\documentclass[a4paper]{article}
\usepackage{a4wide}

\input{packages.sty}
\input{commands.sty}

\begin{document}

\author{Bjoern Andres$^{1,2}$, Silvia Di Gregorio$^2$, Jannik Irmai$^2$ and Jan-Hendrik Lange$^3$}
\date{$^2$TU Dresden $\qquad$ $^3$Max Planck Institute for Informatics}
\title{\textbf{A Polyhedral Study of Lifted Multicuts}}
\maketitle
\footnotetext[1]{Correspondence: \texttt{bjoern.andres@tu-dresden.de}}
\footnotetext[3]{The contributions to this article by Jan-Hendrik Lange are a result of his work at the Max Planck Institute for Informatics, Saarland Informatics Campus, Germany}


\begin{abstract}
Fundamental to many applications in data analysis are the decompositions of a graph, i.e.~partitions of the node set into component-inducing subsets. 
One way of encoding decompositions is by multicuts, the subsets of those edges that straddle distinct components.
Recently, a lifting of multicuts from a graph $G = (V, E)$ to an augmented graph $\widehat G = (V, E \cup F)$ has been proposed in the field of image analysis, with the goal of obtaining a more expressive characterization of graph decompositions in which it is made explicit also for pairs $F \subseteq \tbinom{V}{2} \setminus E$ of non-neighboring nodes whether these are in the same or distinct components.
In this work, we study in detail the polytope in $\mathbb{R}^{E \cup F}$ whose vertices are precisely the characteristic vectors of multicuts of $\widehat G$ lifted from $G$, connecting it, in particular, to the rich body of prior work on the clique partitioning and multilinear polytope.
\end{abstract}

\setcounter{tocdepth}{3}
{
	\hypersetup{hidelinks}
	\tableofcontents
}

\section{Introduction}\label{sec:introduction}

Fundamentally, we are interested in the set of all decompositions of a (finite, simple, undirected) graph.
A decomposition of a graph $G = (V, E)$ is a partition $\Pi$ of the node set $V$ such that, for every $U \in \Pi$, the subgraph of $G$ induced by $U$ is connected, and hence a (not necessarily maximal) component of $G$ (\Cref{def:decomposition}).
An example is depicted in \Cref{fig:multicut-decomposition-example}.
Decompositions of a graph occur in practice, as a mathematical abstraction of different ways of clustering data, and in theory, as a generalization of the partitions of a set, to which they specialize for complete graphs.

We follow \citet{Chopra1993,Chopra1995} in studying the set of all decompositions of a graph through its characterization as a set of multicuts.
The multicut induced by a decomposition is the set of those edges that straddle distinct components (\Cref{def:multicut}).
An example is depicted in \Cref{fig:multicut-decomposition-example}.

Our work is motivated by a limitation of multicuts as a characterization of decompositions:
For a complete graph $K_V = (V, \tbinom{V}{2})$, the characteristic function $x \colon \tbinom{V}{2} \to \{0,1\}$ of a multicut $x^{-1}(1)$ of $K_V$ makes explicit for every pair $\{u,v\} \in \tbinom{V}{2}$ whether the nodes $u$ and $v$ are in the same component, indicated by $x_{\{u,v\}} = 0$, or in distinct components, indicated by $x_{\{u,v\}} = 1$.
For a general graph $G = (V, E)$, however, the characteristic function $x \colon E \to \{0,1\}$ of a multicut $x^{-1}(1)$ of $G$ makes explicit \emph{only for neighboring nodes} $\{u,v\} \in E$ whether $u$ and $v$ are in the same or distinct components.
Hence, the binary linear optimization problem whose feasible solutions are the characteristic functions of the multicuts of a graph is less expressive for general graphs than it is for complete graphs.

In order to make explicit also for non-neighboring nodes, specifically, for all $\{u,v\} \in E \cup F$ with $F \subseteq \tbinom{V}{2} \setminus E$, whether $u$ and $v$ are in distinct components, we consider a \emph{lifting} of the multicuts of $G$ to multicuts of the augmented graph $\widehat{G} = (V, E \cup F)$.
The multicuts of $\widehat{G}$ lifted from $G$ are still in one-to-one relation with the decompositions of $G$. 
Yet, they are a more expressive characterization of these decompositions than the multicuts of $G$.
This expressiveness has applications in the field of image analysis \citep{Beier2017, Tang2017}, as we discuss in \Cref{sec:related-work}.

In this article, we study the polytope in the affine space $\mathbb{R}^{E \cup F}$ whose vertices are the characteristic functions of the multicuts of $\widehat{G}$ lifted from $G$.
We refer to this object as the \emph{lifted multicut polytope} with respect to $G$ and $\widehat G$. 
In this study, we focus separately on the cases of $G$ being a general graph, path, tree and cycle.

\subsection{Contributions}

We make the following contributions:
\begin{itemize}
\item We generalize results of \citet{Chopra1993} for multicut polytopes to lifted multicut polytopes. 
In particular, we establish full-dimensionality and the exact condition under which cycle inequalities define facets. 
Moreover, we establish conditions under which further inequalities of the canonical relaxation of the lifted multicut problem are facet-defining.

\item We establish a new class of facet-defining inequalities that arise from cycles in the graphs. 
These constitute a new class of facet-defining inequalities also for the comprehensively studied \emph{clique partitioning polytope} \cite{Groetschel1990}. 

\item We offer a complete description of the polytopes of the multicuts of a complete graph lifted from a path. 
This geometric description complements the combinatorial results about the \emph{sequential set partition} problem from \citet{Kernighan1971}.

\item We establish a new class of facet-defining inequalities for the polytopes of the multicuts of a general graph lifted from a tree and the exact condition under which these are facet-defining also when lifting from a general graph.

\item We study the relation between the lifted multicut polytope for trees and the multilinear polytope.
To this end, we generalize the inequalities we introduce for the lifted multicut polytope for trees and establish a connection to known inequalities for the multilinear polytope. 
Moreover, we show that the lifted multicut problem in case of lifting from a path corresponds to multilinear optimization over $\beta$-acyclic hypergraphs.

\item We establish further classes of facets that contribute to the understanding of polytopes of multicuts lifted from cycles.
\end{itemize}

\Cref{thm:dimension} on the full-dimensionality of the lifted multicut polytope as well as the study of facets from canonical inequalities in \Cref{sec:facets-from-canonical-ineq} have been published before in a conference article \citep{Hornakova2017}.
Here, we correct and simplify proofs given there, without altering the results. 
In this, we build on the dissertation of \citet{Lange2020phd}.
The proofs of \Cref{thm:dimension,thm:cycle-facets} we offer here employ a different (simpler) construction than the proofs by \citet{Hornakova2017} and \citet{Lange2020phd}.
We reproduce from \citet{Hornakova2017} Figures 2 and 3, from the article, and Figures 3 and 4, from the supplement. 
The results on polytopes of multicuts of a complete graph lifted from a tree or path in \Cref{sec:tree-partition-problem,sec:lmp-trees,sec:lmp-trees-facets,sec:lmp-paths}, except \Cref{thm:intersection-arbitrary-g}, have been published before in a conference article \citep{Lange2020} from which we adapt\footnote{Adapted/Translated by permission from Springer Nature: \citet{Lange2020} according to License No.~5243580658511.}, in particular, Figure 1.

\begin{figure}
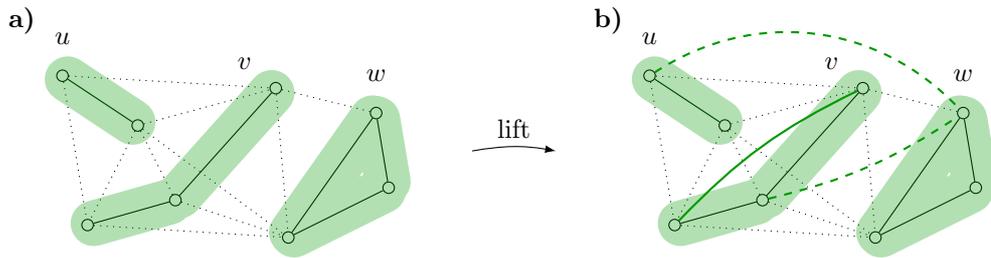

\centering
\begin{minipage}{0.4\textwidth}
\textbf{a)}
\raisebox{2ex}{\imagetop{%
\input{figures/introductory-example/a.tex}
}}
\end{minipage}
\begin{minipage}{0.1\textwidth}
	\tikz{\draw[-latex] (0,0) to[bend left=10] node[above]{lift} (1,0);}
\end{minipage}
\begin{minipage}{0.45\textwidth}
\textbf{b)}
\raisebox{2ex}{\imagetop{%
\input{figures/introductory-example/b.tex}
}}
\end{minipage}
\caption{Depicted above in \textbf{a)} (in green) is a decomposition of a graph $G$, i.e.~a partition of the node set into connected subsets.
Any decomposition of a graph is characterized by the set of those edges (depicted above as dotted lines) that straddle distinct components.
These subsets of edges are called the \emph{multicuts} of the graph.
A multicut $M$ makes explicit for neighboring nodes $u$ and $v$ whether these are in the same component, indicated by $\{u,v\} \notin M$, or in distinct components, indicated by $\{u,v\} \in M$.
In order to make this information explicit also for non-neighboring nodes, e.g.~for $\{u,w\}$, we identify the set of all decompositions of a graph $G$ with a subset of the multicuts of an augmented graph $\widehat{G}$ which we call the multicuts of $\widehat G$ \emph{lifted} from $G$. 
An example can be seen in \textbf{b)} with augmented edges depicted as thick green curves.
In this article, we study the polytopes whose vertices are the characteristic functions of lifted multicuts.
We refer to these as \emph{lifted multicut polytopes}.}
\label{fig:multicut-decomposition-example}
\end{figure}

\subsection{Contents}

This article is organized as follows.
First, in \Cref{sec:related-work}, we discuss related work.
In \Cref{sec:preliminaries}, we introduce basic terminology.
In \Cref{sec:lifting-of-multicuts}, we define the lifted multicut problem and polytope.

In \Cref{sec:lifted-multicut-polytope}, we study the lifted multicut polytope in its most general form. 
We characterize its vertices in terms of linear inequalities and integrality constraints, which yields an integer linear programming (ILP) formulation of the lifted multicut problem.
We show that the polytope is full dimensional and establish relations between different lifted multicut polytopes.
Furthermore, we investigate conditions under which the inequalities of the ILP formulation are facet-defining. 

In \Cref{sec:lmc-trees-and-paths}, we turn to the special case of lifting from a \emph{tree}. 
In this case, the lifted multicut problem can be stated as a binary multilinear optimization problem.
We offer a relaxation of the lifted multicut polytope for trees which is tighter than the standard relaxation.
In addition, we introduce a class of facet-defining inequalities for the lifted multicut polytope for trees and establish the exact condition under which these inequalities define facets also when lifting from an arbitrary graph.
Thanks to this class of inequalities, we obtain a complete description of the polytopes of the multicuts of a complete graph lifted from a \emph{path}.
We further generalize these inequalities for the case of lifting from a tree to a general, possibly incomplete, graph and provide necessary conditions for these inequalities to be facet-defining.
We establish a relation between these inequalities and known valid inequalities for the multilinear polytope. 
Beyond this, we investigate further connections between the lifted multicut polytope and the multilinear polytope.

In \Cref{sec:lmc-for-cycles}, we study the lifted multicut polytope in the complementary case of lifting from a \emph{cycle}.
First, we investigate valid inequalities that are inherited from the multicut polytope of the complete graph. 
It turns out that, with a few canonical exceptions, the known classes of inequalities are not facet-defining for the lifted multicut polytope for cycles.
We establish several new classes of facet-defining inequalities for the lifted multicut polytope for cycles.
From these results, we derive facet-defining inequalities for the lifted multicut polytope for arbitrary graphs that arise from cycles in that graph. 
These inequalities define a new class of facets of the multicut polytope for complete graphs and, hence, also for the isomorphic clique partitioning polytope.

\section{Related work}\label{sec:related-work}

We build on a long line of work in and related to the field of discrete optimization.
Below, we visit this work in chronological order, starting with early combinatorial algorithms and finishing with recent applications in the area of image analysis.

In the late 1960s, the task of partitioning the node set of a graph arises from practical problems in computer science: the laying out of circuits on computer boards \cite{charney1968efficient, russo1971heuristic} and the segmentation of computer programs \cite{kernighan1969some}. 
\citet{Kernighan1970} devise a heuristic algorithm and, for the special case of nodes adhering to a linear order, \citet{Kernighan1971} presents an algorithm that computes an optimal solution in time proportional to the number of edges. 
The corresponding integer linear programming formulation for this special case admits a totally unimodular constraint matrix \cite{Joseph1997}. 

Motivated by the task of clustering data, \citet{Groetschel1989} consider the problem of partitioning the node set of complete graphs. 
As every subset of the nodes of a complete graph induces a clique, they refer to the problem as \emph{clique partitioning}.
In the clique partitioning problem, a feasible solution is the characteristic vector of a subset of edges such that the graph consisting of all nodes and these edges is transitive.
For complete graphs, this characteristic vector is one minus the characteristic vector of a multicut, the subset of edges between cliques.
\citet{Groetschel1990} pioneer the study of the \emph{clique partitioning polytope}, the convex hull of all feasible solutions of the clique partitioning problem.
Due to the simple relation between characteristic vectors of clique partitionings and characteristic vectors of multicuts, properties of the clique partitioning polytope transfer easily to the multicut polytope for complete graphs, and vice versa.

The polyhedral study of the partition problem for general graphs is initiated by \citet{Chopra1993,Chopra1995}.
They introduce the structures that we refer to as a multicut and multicut polytope \citep[Lemma~2.2]{Chopra1993}.
Notably, they consider also additional constraints that restrict the number of components.
A complete description of small multicut polytopes can be found in \citet{Deza1990}, while \citet{Chopra1994} gives a complete description for series-parallel and 4-wheel free graphs.
\citet{Deza1992} present the class of \emph{clique-web inequalities} that generalizes several classes of facets established before by \citet{Groetschel1990,Chopra1993}. 
See also \citet{Sorensen2002}. 
They provide a comprehensive characterization of those clique-web inequalities that define facets, even for cases where additional constraints are imposed on the number of components.
In subsequent studies, \citet{Groetschel1990composition,Bandelt1999lifting,Oosten2001clique} present techniques for composing new classes of facet-defining inequalities from known inequalities. 
In particular, \citet{Oosten2001clique} succeed in classifying all facet-defining inequalities of the clique partitioning polytope with right hand side $1$ or $2$.
Many of the aforementioned results build on prior polyhedral studies of the cut polytope \cite{Barahona1986}, the cut cone \cite{deza1992cutconeI,deza1992cutconeII} and the bipartite subgraph polytope \cite{barahona1985bipartite}.
For an overview of known facets of the multicut polytope for complete graphs, see \Cref{tab:known-facets}.
The separation problem for various classes of facet-defining inequalities is discussed by \citet{Groetschel1989,Deza1992,muller1996partial,caprara1996,Oosten2001clique}.

\begin{table}[b!]
\caption{Classes of valid inequalities for the multicut polytope of a complete graph.}
\centering
\small
\vspace{0.5ex}
\begin{tabular}{>{\raggedright\hangindent=2ex}p{0.25\textwidth} L{0.34\textwidth} L{0.32\textwidth}}
\toprule
Inequality & Original reference & Remarks \\
\midrule
$0 \leq x$ & \citet{Groetschel1990} & Not facet-defining 
\\
$x \leq 1$ & \citet{Groetschel1990} & Facet-defining 
\\
Triangle inequality & \citet{Groetschel1990} & Facet-defining 
\\
$[S,T]$-inequality & \citet{Groetschel1990} & Facet-defining iff $|S| \neq |T|$ 
\\
2-chorded cycle inequality & \citet{Groetschel1990} & See~\Cref{thm:2-chorded-valid-and-facet} 
\\
2-chorded path inequality & \citet{Groetschel1990} & Facet-defining iff the paths has even length 
\\
2-chorded even wheel inequality & \citet{Groetschel1990} & Facet-defining 
\\
General 2-partition inequality & \citet{Groetschel1990composition} & Facet-defining; generalizes $[S,T]$-inequalities
\\
Cycle inequality & \citet{Chopra1993} & Facet-defining precisely for chordless cycles, i.e., for triangles; generalizes triangle inequalities
\\
(Bicycle) wheel inequality & \citet{Chopra1993} & Facet-defining iff wheel has odd length 
\\
Clique-web inequalities & \citet{Deza1992} & See~\Cref{thm:clique-web-facet}; generalizes triangle, $[S,T]$- and (bicycle) wheel inequalities 
\\
(Lifted) weighted (s,T)-inequality & \citet{Oosten2001clique} & Generalizes $[S,T]$-inequalities with $|S|=1$ 
\\
(Lifted) stable set inequality & \citet{Oosten2001clique} & Facet-defining under certain necessary conditions
\\
Generalized 2-chorded cycle inequality & \citet{Oosten2001clique} & Generalizes 2-chorded cycle inequalities by additionally including some $3$-chords and $4$-chords; facet-defining under certain necessary conditions
\\
Generalized 2-chorded path inequality & \citet{Oosten2001clique} & Facet-defining for paths of even length; generalizes 2-chorded path inequalities  
\\
Generalized 2-chorded even wheel inequality & \citet{Oosten2001clique} & Facet-defining; generalizes 2-chorded even wheel inequalities  
\\
\bottomrule
\end{tabular}
\label{tab:known-facets}
\end{table}

Optimization problems closely related to the multicut problem are \emph{correlation clustering} \cite{Bansal2004} and \emph{coalition structure generation} in \emph{weighted graph games} \cite{bachrach2013optimal}.
For correlation clustering, one further distinguishes between three variations: \emph{minimizing disagreement}, \emph{maximizing agreement} and \emph{maximizing correlation}.
All variations of correlation clustering, as well as coalition structure generation in weighted graph games and multicut share the same set of feasible solutions and differ only by constant additive terms in the objective function.
Therefore, these problems are equivalent at optimality, and \textsc{np}-hardness of one implies \textsc{np}-hardness of the others.
\citet{Bansal2004} show that correlation clustering for complete graphs is \textsc{np}-hard even for unit weights. 
Independently, \citet{voice2012coalition} and \citet{bachrach2013optimal} show that the problem remains \textsc{np}-hard for planar graphs.
Complementary to this hardness result, \citet{Klein2015} present a polynomial time approximation scheme for planar graphs via a reduction to the problem of finding a minimal two-edge-connected augmentation.

Although these variations of the problems all have the same solutions, they differ significantly regarding the hardness of approximation.
A survey covering all variations and all restrictions to specific classes of graphs already studied is beyond the scope of this article.
We summarize some important results:
\citet{Bansal2004} offer a constant factor approximation algorithm for minimizing disagreement in unweighted complete graphs.
Their results are strengthened by \citet{Charikar2005} who show \textsc{apx}-hardness and significantly improve the approximation factor to $4$.
Independently, \citet{Charikar2005} and \citet{Demaine2006} develop $\mathcal{O}(\text{log}\; n)$ approximation algorithms for minimizing disagreement in general weighted graphs.
For the problem of maximizing agreement, \citet{Bansal2004} offer a polynomial time approximation scheme in case of unweighted complete graphs while the problem is \textsc{apx}-hard for general weighted graphs \cite{Charikar2005}.
Coalition structure generation in weighted graph games is the hardest of the problem variations in the sense that it cannot be approximated to within $\mathcal{O}(n^{1-\epsilon})$ for all $\epsilon > 0$ unless \textsc{p} = \textsc{np} \cite{bachrach2013optimal,zuckerman2006linear}.
For further results on the hardness of approximation and the study of efficient (approximation) algorithms, the interested reader is referred to \citet{emanuel2003correlation,Charicar2004maximizing,swamy2004correlation,Chawla2006,Chawla2015,Veldt2017,Veldt2021}.

The \emph{lifted multicut problem}, the generalization of the multicut problem in which costs assigned to pairs of \emph{non-neighboring} nodes are taken into account as well, is introduced by \citet[Definition 1]{Keuper2015b} in the context of applications in the fields of image analysis and computer graphics.
Specifically, they consider the problem of decomposing the pixel grid graph of an image into objects based on estimates of object boundaries, and the task of decomposing simplicial surfaces of three-dimensional objects into smooth components based on estimates of their curvature.
Subsequently, \citet{Beier2017} apply the lifted multicut problem for the task of decomposing large volume images of neural tissue into individual cells.
At the same time, \citet{Tang2017} employ lifted multicuts for the task of tracking multiple pedestrians in a monocular video.
Their task is to decide for candidate detections of pedestrians, modeled as nodes in a graph, whether these refer to the same pedestrian or distinct pedestrians.
They use lifting to relate candidate detections across longer distances in time.
Specifically, they relate candidate detections that appear similar in the video by additional edges with positive cost, thus rewarding feasible solutions in which these candidate detections refer to the same pedestrian, but without introducing additional feasible solutions that would link these candidate detections directly via the additional edges.
They empirically quantify an advantage, in the context of this application, of this lifted multicut problem over the multicut problem with respect to just the augmented graph and explain it from the point of view of the application by the simple fact that similar looking pedestrians need not be identical.

The lifted multicut problem in case of lifting from a tree can be stated equivalently as minimization of a \emph{multilinear} objective function in binary variables, i.e.~a \emph{pseudo-Boolean function}.
In this way, the study of the lifted multicut polytope is connected to the field of multilinear optimization.
The combinatorial polytope associated with the linearization of quadratic pseudo-Boolean functions is studied, among others, by \citet{Hammer1984,Barahona1986,Padberg1989,DeSim}. 
Quadratization techniques are commonly used in order to solve pseudo-Boolean optimization problems.
The benefit of reducing the problem to a quadratic one, via the addition of variables and constraints, is the possibility to take advantage of the rich literature available for the quadratic case.
We refer the reader to e.g. \citet{Ros75,Boros2002,BucRin07,Ish11,Boros2012,EllLamLaz19}.
Recent research also considers the linearization of more general multilinear forms \cite{dPKha17,dPKha21,dPDiG21,HojPfeWal19}. 
This line of work is motivated by the idea of avoiding the additional constraints and variables introduced in order to express the problem as a quadratic one, and to exploit the structure of the original problem.
The connections between this line of work and our results is discussed in more detail in \Cref{sec:multilinear}.

While the multicut and lifted multicut problem are non-trivial only for a combination of positive and negative costs attributed to the edges, related problems defined for non-negative edge weights only and additional constraints include 
\emph{$k$-terminal cut} \cite{dahlhaus1992complexity} where $k$ terminals are to be separated by a cut, 
\emph{$k$-multicommodity cut} \cite{leighton1999multicommodity} where $k$ source-sink pairs are to be separated by a cut, and 
\emph{$k$-cut} \cite{goldschmidt1994polynomial} where the graph is to be cut into $k$ components without consideration of specific nodes.
The first two problems are \textsc{np}-hard for $k \geq 3$ while the latter can be solved in polynomial time for fixed $k$. 
Polynomial equivalence of $k$-multicommodity cut and correlation clustering is established by \citet{Demaine2006}.


\section{Preliminaries}\label{sec:preliminaries}

Let $G = (V, E)$ be a connected graph with node set $V$ and edge set $E \subseteq \binom{V}{2}$.
For pairs of nodes $u,v \in V$ with $u \neq v$, we write $uv = \{u,v\} = vu$, for short.
For a subset $A \subseteq E$ of edges, we let $\1_A \in \{0,1\}^E$ denote the \emph{characteristic vector} of the set $A$, i.e. $(\1_A)_e = 1 \Leftrightarrow e \in A$ for all $e \in E$.
For clarity, we distinguish between decompositions of the graph, i.e.~specific partitions of the node set, and (multi)cuts of the graph, i.e.~sets of edges that straddle distinct components.
Examples are depicted in \Cref{fig:multicut-decomposition-example}a.

\begin{definition}
\label{def:decomposition}
A partition $\Pi$ of the node set $V$ of a graph $G = (V, E)$ is called a \emph{decomposition} of $G$ if and only if every $U \in \Pi$ induces a (connected) component of $G$.
For any $k \in \N$ and any decomposition $\Pi$ of $G$, $\Pi$ is called a \emph{$k$-decomposition} if and only if $|\Pi|=k$. 
For any two distinct nodes $u,v \in V$ and any 2-decomposition $\Pi = \{U,V\setminus U\}$, $\Pi$ is called a \emph{$uv$-decomposition} if and only if $u \in U$ and $v \in V \setminus U$, or $u \in V \setminus U$ and $v \in U$.
\end{definition}

\begin{definition}
\label{def:multicut}
A set $M \subseteq E$ is called a \emph{multicut} of a graph $G = (V, E)$ if and only if there exists a decomposition $\Pi$ of $G$ such that $M$ consists of precisely those edges that straddle distinct components of $\Pi$.
\end{definition}

As the decomposition in \Cref{def:multicut} is unique, there is a bijection $\phi_G: D_G \to M_G$ between the set $D_G \subseteq 2^{2^V}$ of all decompositions of $G$ and the set $M_G \subseteq 2^E$ of all multicuts of $G$, specifically
\begin{align}\label{eq:phi}
    \phi_G(\Pi) = \left\{uv \in E \mid \forall U \in \Pi: u \notin U \text{ or } v \notin U \right\} \enspace .
\end{align}
We will frequently switch between decompositions and multicuts and refer to one as being \emph{induced} by the other. For an illustration, see \Cref{fig:multicut-decomposition-example}a.
Specifically, a multicut $M$ of $G$ is called a \emph{$k$-cut} of $G$ if and only if $M$ is induced by a $k$-decomposition.
Similarly, a multicut $M$ of $G$ is called a \emph{$uv$-cut} of $G$ if and only if $M$ is induced by a $uv$-decomposition.
For any $uv$-decomposition $\Pi = \{U,V\setminus U\}$ we let $\delta(U)$ denote the induced $uv$-cut. 
Note that $uv$-cuts are necessarily minimal, for when removing any edge from $\delta(U)$, this set is no longer a $uv$-cut.

We conclude this section with a characterization of multicuts by \citet{Chopra1993}.

\begin{proposition}[Lemma 2.2 of \citet{Chopra1993}]\label{prop:multicuts-statisfy-cycle-ineq}
A subset $M \subseteq E$ of the edge set $E$ of a graph $G = (V, E)$ is a multicut of $G$ if and only if no cycle of $G$ contains precisely one edge of $M$.
\end{proposition}

\section{Lifting of multicuts}\label{sec:lifting-of-multicuts}

For any multicut $M$ of $G = (V,E)$, the characteristic vector $\1_M$ makes explicit for every pair $uv \in E$, whether $u$ and $v$ are in distinct components.
To make explicit also for non-neighboring nodes, specifically, for all $uv \in F$ where $F \subseteq \binom{V}{2} \setminus E$, whether $u$ and $v$ are in distinct components, we define a lifting of the multicuts of $G$ to multicuts of the augmented graph $\widehat G = (V, E \cup F)$.

\begin{definition}
Let $G=(V,E)$ be a connected graph and let $\widehat G = (V, E \cup F)$ with $F \subseteq \binom{V}{2}\setminus E$ be an \emph{augmentation} of $G$, i.e. the graph obtained from $G$ by adding the set of edges $F$.
We call the composed map $\phi_{\widehat G} \circ \phi_G^{-1}$ the \emph{lifting} of multicuts from $G$ to $\widehat G$.
For any multicut $M$ of $G$, we call the set $(\phi_{\widehat G} \circ \phi_G^{-1}) (M)$ a multicut of $\widehat G$ \emph{lifted from} $G$.
\end{definition}

\begin{definition}
    Let $G = (V,E)$ be a connected graph and let $\widehat G = (V, E \cup F)$ be an augmentation. 
    We call the convex hull of characteristic vectors of multicuts of $\widehat G$ lifted from $G$ the \emph{lifted multicut polytope} with respect to $G$ and $\widehat G$, denoted by
    \begin{align}\label{eq:lifted-multicut-polytope}
        \lmc(G,\widehat G) = 
        \conv \big \{ \1_M \mid M \text{ multicut of } \widehat G \text{ lifted from } G \big \} \enspace . 
    \end{align}
    For brevity, we define $\lmc(G):=\lmc(G,K_V)$ where $K_V$ is the complete graph with nodes $V$.
\end{definition}

\begin{definition}
    Let $G = (V,E)$ be a connected graph, let $\widehat{G} = (V,E \cup F)$ be an augmentation and let $\theta \in \R^{E \cup F}$ be a vector associated with the edges of the augmented graph.
    The instance of the \emph{lifted multicut problem} with respect to $G$, $\widehat G$ and $\theta$ consists in finding a minimum cost multicut of $\widehat G$ lifted from $G$ with respect to $\theta$. 
    It has the form
    \begin{align}\label{eq:lifted-multicut-problem}
        \min_{x \in \lmc(G,\widehat{G})} \; \sum_{e \in E \cup F} \theta_e \, x_e \enspace. \tag{LMP}
    \end{align}
\end{definition}

If $F = \emptyset$, then \eqref{eq:lifted-multicut-problem} specializes to the \emph{multicut problem}, i.e., the linear optimization problem over the \emph{multicut polytope} $\mc(G) := \lmc(G, G)$.
If $F \neq \emptyset$, then \eqref{eq:lifted-multicut-problem} differs from the multicut problem with respect to $\widehat G$ and $\theta$, since then it holds that $\lmc(G,\widehat G) \subset \mc(\widehat G)$, cf.~\Cref{prop:inclusion-property}. For an example, see \Cref{fig:multicut-polytope}.
In the lifted multicut problem, the assignment $x_{uv} = 0$ indicates that the nodes $u$ and $v$ are connected in $G$ by a path of edges labeled 0.
This property can be used to penalize (by $\theta_{uv} > 0$) or reward (by $\theta_{uv} < 0$) those decompositions of $G$ for which $u$ and $v$ are in distinct components.

We close this section by establishing \textsc{apx}-hardness for the minimum multicut problem by a simple reduction from the closely related \emph{maximum agreement correlation clustering} problem.

\begin{proposition}\label{prop:mc-apx-hard}
    The multicut problem is \emph{\textsc{apx}}-hard even for costs $c \in \{-1,1\}^E$. 
\end{proposition}

\proofref{prop:mc-apx-hard}

As the minimum lifted multicut problem generalizes the multicut problem, we obtain

\begin{corollary}
    The lifted multicut problem \eqref{eq:lifted-multicut-problem} is \emph{\textsc{apx}}-hard.
\end{corollary}

\section{Lifted multicut polytope}\label{sec:lifted-multicut-polytope}

\begin{figure}
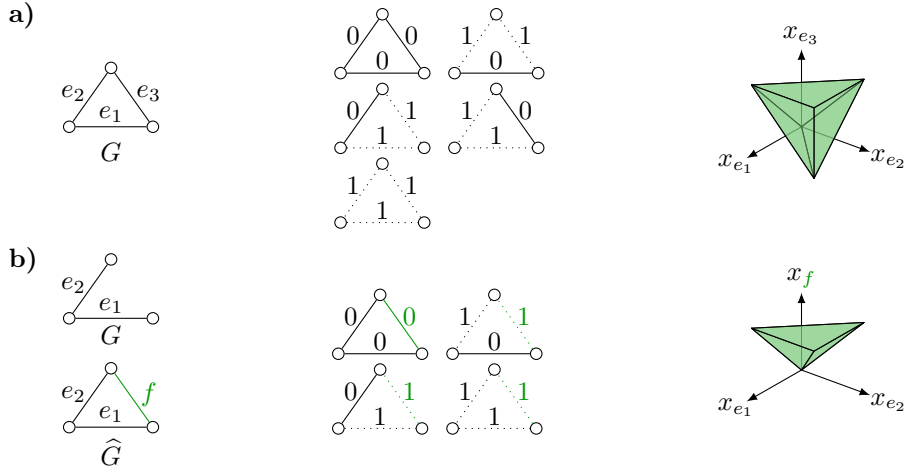

	\centering
    \begin{minipage}[t]{0.8\linewidth}
    \imagetop{\textbf{a)}}\vspace{-2ex}\hspace{3ex}
    \input{figures/triangle-example/triangle-graph.tex}
    \hfill
    \begin{tabular}{@{}l@{\ }l@{}}
    \input{figures/triangle-example/decompositions-a.tex}
    & \input{figures/triangle-example/decompositions-b.tex} \\
    \input{figures/triangle-example/decompositions-c.tex}
    & \input{figures/triangle-example/decompositions-d.tex} \\
    \input{figures/triangle-example/decompositions-e.tex}
    \end{tabular}
    \hfill
    \input{figures/triangle-example/multicut-polytope.tex}
    \end{minipage}
    \\
    \begin{minipage}[t]{0.8\linewidth}
    \imagetop{\textbf{b)}}\vspace{-2ex}\hspace{3ex}
    \begin{tabular}{@{}l@{}}
    \input{figures/triangle-example/lifted-a.tex}
    \\[3ex]
    \input{figures/triangle-example/lifted-b.tex}
    \end{tabular}
    \hfill
    \begin{tabular}{@{}l@{\ }l@{}}
    \input{figures/triangle-example/lifted-decompositions-a.tex}
    & \input{figures/triangle-example/lifted-decompositions-b.tex} \\
    \input{figures/triangle-example/lifted-decompositions-c.tex}
    & \input{figures/triangle-example/lifted-decompositions-d.tex}
    \end{tabular}
    \hfill
    \input{figures/triangle-example/lifted-multicut-polytope.tex}
    \end{minipage}
    \caption{\textbf{a)} For any connected graph $G$ (left), 
    the characteristic vectors of multicuts of $G$ (middle)
    span, as their convex hull in $\R^E$, the multicut polytope of $G$ (right),
    a 01-polytope that is $\abs{E}$-dimensional
    \citep{Chopra1993}.
    \textbf{b)} For any connected graph $G = (V, E)$ (top left)
    and any graph $\widehat G = (V, E \cup F)$ (bottom left),
    the characteristic vectors of multicuts of $\widehat G$ that are lifted from $G$ (middle)
    span, as their convex hull in $\R^{E \cup F}$, the lifted multicut polytope with respect to $G$ and $\widehat G$ (right),
    a 01-polytope that is $\abs{E \cup F}$-dimensional
    (\Cref{thm:dimension}).}
    \label{fig:multicut-polytope}
\end{figure}

In this section, we study the geometry of the lifted multicut polytope $\lmc(G,\widehat{G})$ defined by \eqref{eq:lifted-multicut-polytope}. 
To this end, we first state a description of $\lmc(G,\widehat{G})$ in terms of linear inequalities and integrality constraints.

\begin{proposition}\label{prop:lifted-multicut-polytope-inequalities}
    The polytope $\lmc(G,\widehat{G})$ is the convex hull of all vectors $x \in \{0,1\}^{E \cup F}$ that satisfy the following inequalities:
    \begin{align}
        x_f &\leq \sum_{e \in E_C \setminus \{f\}} x_e 
            && \forall \; \text{cycles } C=(V_C,E_C) \text{ in } G \quad \forall f \in E_C
            \label{eq:lifted-multicut-cycle} \\
        x_{uv} &\leq \sum_{e \in E_P} x_e 
            && \forall \; uv \in F \text{ and all } uv\text{-paths } P \text{ in } G
            \label{eq:lifted-multicut-path} \\
        1 - x_{uv} &\leq \sum_{e \in \delta(U)} 1 - x_e 
            && \forall \; uv \in F \text{ and all } uv\text{-cuts } \delta(U) \text{ in } G
           	\enspace .
            \label{eq:lifted-multicut-cut}
    \end{align}
\end{proposition}

\proofref{prop:lifted-multicut-polytope-inequalities}

We refer to \eqref{eq:lifted-multicut-cycle}--\eqref{eq:lifted-multicut-cut} as the \emph{cycle}, \emph{path} and \emph{cut inequalities}, respectively.
According to \Cref{prop:lifted-multicut-polytope-inequalities}, a vector $x \in \{0,1\}^{E \cup F}$ is the characteristic vector of a multicut of $\widehat G$ lifted from $G$ if and only if, in addition to the cycle inequalities in $G$, it satisfies all path and cut inequalities.

The cycle inequalities \eqref{eq:lifted-multicut-cycle} are introduced by \citet{Chopra1993} for the multicut polytope $\mc(G)$. 
The path inequalities \eqref{eq:lifted-multicut-path} correspond to cycles in $\widehat{G}$ where all edges except $uv$ are edges also of $G$.
Cycle inequalities with respect to cycles in $\widehat{G}$ that do not correspond to a path inequality are satisfied by all points in $\lmc(G,\widehat{G})$, by 
\Cref{prop:multicuts-statisfy-cycle-ineq} and the fact that every multicut of $\widehat G$ lifted from $G$ is a multicut of $\widehat G$. 
Yet, they are redundant in the description of $\lmc(G,\widehat{G})$, as they are implied by the cut inequalities.

\subsection{Dimension and inclusion properties}\label{sec:dimension}

In this section, we show that the lifted multicut polytope $\lmc(G,\widehat{G})$ is full-dimensional as a polytope in $\R^{E \cup F}$. 
Additionally, we describe which lifted multicut polytopes are subsets of another lifted multicut polytope. From this, we derive connections between faces of different lifted multicut polytopes.

\begin{theorem}\label{thm:dimension}    
    Let $G=(V,E)$ be a connected graph and let $\widehat{G}=(V,E \cup F)$ be an augmentation. 
    Then
    \[
        \dim \lmc(G,\widehat{G}) = \abs{E \cup F} \enspace .
    \]
\end{theorem}

\proofref{thm:dimension}

\begin{proposition}\label{prop:inclusion-property}
    Let $G=(V,E)$ be a connected graph, let $\widehat{G}=(V,E\cup F)$ be an augmentation, and let $E' \subseteq E \cup F$ such that $G' = (V, E')$ is connected. Then, $\lmc(G',\widehat{G}) \subseteq \lmc(G,\widehat{G})$ if and only if $E' \subseteq E$.
\end{proposition}

\proofref{prop:inclusion-property}

From the full dimensionality (\Cref{thm:dimension}) and the inclusion property (\Cref{prop:inclusion-property}), we obtain the following lemmata on facet-defining inequalities:

\begin{lemma}\label{lem:valid-and-facet-defining-for-subset}
    Let $G=(V,E)$ be a connected graph, let $\widehat{G}=(V,E \cup F)$ be an augmentation and let $E' \subseteq E$ such that $G'=(V,E')$ is connected. 
    Let $a^\top x \leq b$ be a valid inequality for $\lmc(G,\widehat{G})$. 
    Then, $a^\top x \leq b$ is valid for $\lmc(G',\widehat{G})$. 
    If, furthermore, $a^\top x \leq b$ is facet-defining for $\lmc(G', \widehat{G})$, then $a^\top x \leq b$ is facet-defining for $\lmc(G,\widehat{G})$.
\end{lemma}

\proofref{lem:valid-and-facet-defining-for-subset}

\begin{lemma}\label{lem:facet-for-lifting-to-sub-graph}
    Let $G=(V,E)$ be a connected graph and let $\widehat{G}=(V,E \cup F)$ and $\widehat{G}'=(V,E\cup F')$ be augmentations of $G$ with $F \subseteq F' \subseteq \binom{V}{2} \setminus E$. 
    Let $a^\top x \leq b$ with $a \in \R^{E \cup F'}$ and $b \in \R$ be a valid and facet-defining inequality for $\lmc(G,\widehat{G}')$. 
    Let $E_a = \{e \in \binom{V}{2} \mid a_e \neq 0\}$ be the support of $a$ and let $\bar{a} \in \R^{E \cup F}$ with $\bar{a}_e = a_e$ for $e \in E \cup F$. 
    If $E_a \subseteq E \cup F$, then $\bar{a}^\top y \leq b$ is valid and facet-defining for $\lmc(G,\widehat{G})$.
\end{lemma}

\proofref{lem:facet-for-lifting-to-sub-graph}

\subsection{Facets from canonical inequalities}\label{sec:facets-from-canonical-ineq}

In this section, we investigate which of the inequalities $0 \leq x_e$ and $x_e \leq 1$ for $e \in E \cup F$, and which of the inequalities \eqref{eq:lifted-multicut-cycle}--\eqref{eq:lifted-multicut-cut} are facet-defining for the lifted multicut polytope $\lmc(G,\widehat{G})$. 
Since $\lmc(G,\widehat{G})$ is full dimensional by \Cref{thm:dimension}, its facets are described by inequalities that are unique up to positive scalar multiplication. Moreover, a valid inequality $a^\top x \leq b$ is not facet-defining for $\lmc(G,\widehat{G})$ if and only if every $x$ that satisfies $a^\top x = b$ also satisfies $c^\top x = d$ where $c$ is not a scalar multiple of $a$.

First, we characterize those edges $e \in E \cup F$ for which the inequality $x_e \leq 1$ defines a facet of $\lmc(G,\widehat{G})$.
To this end, we consider separating sets:
For a graph $G = (V,E)$ and nodes $u,v \in V$ a node set $S \subseteq V$ is called a \emph{$uv$-separating set} if any only if every $uv$-path in $G$ contains at least one node in $S$. 
Furthermore, a node $w \in V$ is called a \emph{$uv$-cut-node} if any only if $\{w\}$ is a $uv$-separating set.

\begin{theorem}\label{thm:upper-box-facets}
    Let $G = (V,E)$ be a connected graph and let $\widehat{G} = (V,E\cup F)$ be an augmentation.
    For $e = st \in E \cup F$, the inequality $x_e \leq 1$ defines a facet of $\lmc(G,\widehat{G})$ if and only if there is no $uv \in F \setminus \{e\}$ such that $s$ and $t$ are $uv$-cut-nodes with respect to $G$.
\end{theorem}

\proofref{thm:upper-box-facets}

Next, we give conditions that contribute to identifying those edges $e \in E \cup F$ for which the inequality $0 \leq x_e$ defines a facet of the lifted multicut polytope $\lmc(G,\widehat{G})$.

\begin{theorem}\label{thm:lower-box-facets}
    Let $G = (V,E)$ be a connected graph, let $\widehat{G} = (V,E\cup F)$ be an augmentation and let $e \in E \cup F$. 
    In case $e \in E$, the inequality $0 \leq x_e$ defines a facet of $\lmc(G,\widehat{G})$ if and only if there is no triangle in $\widehat G$ that contains $e$.
    In case $uv = e \in F$, the inequality $0 \leq x_e$ defines a facet of $\lmc(G,\widehat{G})$ only if the following necessary conditions hold:
    \begin{enumerate}[(i)]
        \item \label{cond:box-1}
        There is no triangle in $\widehat G$ that contains $e$.
        \item \label{cond:box-2}
        The distance of any pair of $uv$-cut-nodes except $uv$ itself is at least $3$ in $\widehat G$.
        \item \label{cond:box-3}
        There is no triangle in $\widehat G$ consisting of nodes $s, s', t$ such that $\{s,s'\}$ is a $uv$-separating node set and $t$ is a $uv$-cut-node.
    \end{enumerate}
\end{theorem}

\proofref{thm:lower-box-facets}

Next, we characterize those inequalities of \eqref{eq:lifted-multicut-cycle} and \eqref{eq:lifted-multicut-path} that are facet-defining for $\lmc(G,\widehat{G})$.
\citet{Chopra1993} show that a cycle inequality defines a facet of the multicut polytope $\mc(G)$ if and only if the associated cycle is chordless. 
We establish a similar characterization of those cycle and path inequalities in the description of $\lmc(G,\widehat{G})$ from \Cref{prop:lifted-multicut-polytope-inequalities} that are facet-defining.

\begin{theorem}\label{thm:cycle-facets}
    Let $G = (V,E)$ be a connected graph and let $\widehat{G} = (V,E\cup F)$ be an augmentation.
    The following statements hold true:
    \begin{enumerate}[(a)]
        \item \label{enum:chordless-cycle}
        For any cycle $C=(V_C,E_C)$ in $G$ and any $f \in E_C$, the corresponding cycle inequality \eqref{eq:lifted-multicut-cycle} defines a facet of $\lmc(G,\widehat{G})$ if and only if $C$ is chordless in $\widehat G$. 
        \item \label{enum:chordless-path} For any edge $uv = f \in F$ and any $uv$-path $P=(V_P,E_P)$ in $G$, the corresponding path inequality \eqref{eq:lifted-multicut-path} defines a facet of $\lmc(G,\widehat{G})$ if and only if $E_P \cup \{f\}$ induces a chordless cycle in $\widehat G$.
    \end{enumerate}
\end{theorem}

\proofref{thm:cycle-facets}

Any cycle inequality with respect to some cycle $C=(V_C,E_C)$ in $\widehat G$ and $f \in E_C$ is valid for $\lmc(G,\widehat G)$ as it is valid for $\mc(\widehat G)$ which contains $\lmc(G,\widehat G)$, by \Cref{prop:inclusion-property}.
For any cycle inequality to define a facet of $\lmc(G,\widehat G)$, it is necessary that the associated cycle is chordless, as is shown in the proof of \Cref{thm:cycle-facets}.
In general, however, chordlessness is not a sufficient condition if the inequality is neither of the form~\eqref{eq:lifted-multicut-cycle} for a cycle $C$ in $G$ nor of the form~\eqref{eq:lifted-multicut-path}.
For example, consider the graph $\widehat G$ depicted in \Cref{fig:multicut-polytope}b.
Here, the cycle inequality $x_{e_1} \leq x_f + x_{e_2}$ is dominated by the cut inequality
$x_{e_1} \leq x_f \iff 1-x_f \leq 1-x_{e_1}$ together with $x_{e_2} \geq 0$.

Next, we consider the cut inequalities \eqref{eq:lifted-multicut-cut}.
Our goal is to constrain the class of cuts that gives rise to facet-defining inequalities.
To this end, we define several concepts and apply these in \Cref{thm:cut-facets} to formulate necessary conditions under which cut inequalities define facets.
Non-trivial examples of cuts whose associated inequalities fail to define facets of $\lmc(G,\widehat{G})$ are shown in \Cref{figure:violated-conditions}.

For $uv \in F$ and $U \subseteq V$ such that $u \in U$, $v \notin U$ and $\delta(U)$ is a $uv$-cut we define
\begin{align*}
    S(uv, U) 
    & = \left\{ x \in \lmc(G,\widehat{G}) \cap \Z^{E \cup F} \; \middle| \; 1 - x_{uv} = \sum_{e \in \delta(U)} 1 - x_e \right\} \enspace ,
    \\
    \Sigma(uv, U) 
    & = \conv S(uv, U) \enspace ,
\end{align*}
i.e.~$\Sigma(uv,U)$ is the face defined by the cut inequality with respect to $uv$ and $\delta(U)$, and $S(uv,U)$ are the integral points in that face.

\begin{definition}
    Let $G=(V,E)$ be a connected graph, let $u,v \in V$, and let $U \subseteq V$ such that $u \in U$, $v \notin U$ and $\delta(U)$ is a $uv$-cut.
    A connected, induced subgraph $H = (V_H, E_H)$ of $G$ is called \emph{$(uv, U)$-connected} if 
    \begin{align*}
        u, v \in V_H \text{ and } \abs{E_H \cap \delta(U)} = 1 \enspace .
    \end{align*}
\end{definition}

\begin{figure}
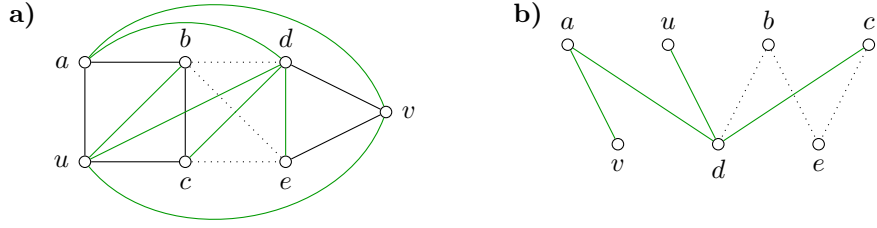

    \centering
    \textbf{a)}
    \raisebox{5ex}{\imagetop{%
    \input{figures/uv-U-connected.tex}
    }}
    \hspace{2em}
    \textbf{b)}
    \raisebox{2ex}{\imagetop{%
    \input{figures/uv-U-connected-G-hat.tex}
    }}
   	\vspace{-2ex}
    \caption{Depicted above in \textbf{a)} are a graph $G$ (black edges) and an augmentation $\widehat{G}$ (black and green edges), together with a $uv$-cut $\delta(U)$ in $G$ (dotted edges) with $U = \{u, a, b, c\}$. 
    The subgraph $H=(V_H,E_H)$ of $G$ that is induced by $V_H=\{u, a, b, e, v\}$ is $(uv,U)$-connected, for example, as $u,v \in V_H$ and $\abs{E_H \cap \delta(U)} = \abs{\{be\}} = 1$. 
    Depicted in \textbf{b)} is the bipartite graph $\widehat{G}(uv,U)$ that consists of all edges in $\delta(U)$ (dashed lines) and $\delta_{F\setminus{uv}}(U)$ (green lines). 
    Here, it holds that $F'_{H'} = \{av\}$.}
    \label{fig:uv-U-connected}
\end{figure}

An example of a $(uv,U)$-connected component is shown in \Cref{fig:uv-U-connected}a. 

\begin{lemma}\label{lem:connectedness}
    Let $G = (V,E)$ be a connected graph and let $\widehat{G} = (V,E\cup F)$ be an augmentation.
    For any $uv \in F$, let $U \subseteq V$ such that $u \in U$, $v \notin U$ and $\delta(U)$ is a $uv$-cut. 
    Every $x \in S(uv, U)$ defines a decomposition of $G$ which contains at most one $(uv, U)$-connected component. 
    That is, at most one maximal component of the graph $(V, \{e \in E \mid x_e = 0\})$ is $(uv, U)$-connected. 
    It exists if and only if $x_{uv} = 0$.
\end{lemma}

\proofref{lem:connectedness}

We denote by $\delta_{F \setminus \{uv\}}(U)$ the set of edges in $F$, except $uv$, that cross the cut, i.e.
\begin{align*}
    \delta_{F \setminus \{uv\}}(U) = 
        \left\{ u'v' \in F \setminus \{uv\} \; \middle|
        \; u' \in U \text{ and } v' \notin U \right\} \enspace .
\end{align*}
Furthermore, let
\begin{align*}
    \widehat G(uv,U) = (V,\ \delta(U) \cup \delta_{F \setminus \{uv\}}(U))
\end{align*}
denote the subgraph of $\widehat G$ that comprises all edges of the cut induced by $U$, except $uv$.
For any $(uv,U)$-connected subgraph $H=(V_H,E_H)$ of G, we denote by 
\begin{align*}
    F'_H = 
        \{ u'v' \in \delta_{F \setminus \{uv\}}(U) \mid 
        u' \in V_H \text{ and } v' \in V_H \}
\end{align*}
the set of those edges $u'v' \in \delta_{F \setminus \{uv\}}(U)$ such that $H$ is also $(u'v',U)$-connected. 
For an exemplary illustration of the above definitions, see \Cref{fig:uv-U-connected}b.

\begin{theorem}\label{thm:cut-facets}
    Let $G = (V,E)$ be a connected graph and let $\widehat{G} = (V,E\cup F)$ be an augmentation.
    For any $uv \in F$ and any $U \subseteq V$ with $u \in U$, $v \notin U$ and $\delta(U)$ a $uv$-cut, the polytope $\Sigma(uv,U)$ is a facet of $\lmc(G,\widehat{G})$ only if the following necessary conditions hold:
    \begin{enumerate}[C1]
        \item \label{cond:cut-1} 
        For any $e \in \delta(U)$, there exists some $(uv,U)$-connected subgraph $H = (V_H, E_H)$ of $G$ such that $e \in E_H$.
        
        \item \label{cond:cut-2} 
        For any $\emptyset \neq F' \subseteq \delta_{F \setminus \{uv\}}(U)$, there exist an edge $e \in \delta(U)$ and $(uv,U)$-connected subgraphs $H = (V_H,E_H)$ and $H'=(V_{H'},E_{H'})$ of $G$ such that
        \begin{align*}
            e \in E_H \text{ and } e \in E_{H'} \text{ and } 
            \abs{F' \cap F'_H} \neq \abs{F' \cap F'_{H'}} \enspace .
        \end{align*}

        \item \label{cond:cut-3} 
        For any $f' \in \delta_{F \setminus \{uv\}}(U)$, any $\emptyset \neq F' \subseteq \delta_{F \setminus \{uv\}}(U) \setminus \{f'\}$ and any $k \in \N$, there exists a $(uv,U)$-connected subgraph $H=(V_H,E_H)$ with $f' \in F'_H$ such that $\abs{F' \cap F'_H} \neq k$ or there exists a $(uv,U)$-connected subgraph $H'=(V_{H'},E_{H'})$ with $f' \notin F'_{H'}$ such that $\abs{F' \cap F'_{H'}} \neq 0$.

        \item \label{cond:cut-4} 
        For any $u' \in U$, any $v' \in V \setminus U$ and any $u'v'$-path $P = (V_P, E_P)$ in $\widehat G(uv,U)$, there exists a $(uv, U)$-connected subgraph $H=(V_H,E_H)$ of $G$ such that
        \begin{align*}
            & (u' \notin V_H \text{ or } \exists v'' \in V_P \setminus U: v'' \notin V_H) \\
            \text{and } \quad & (v' \notin V_H \text{ or } \exists u'' \in V_P \cap U: u'' \notin V_H) \enspace .
        \end{align*}

        \item \label{cond:cut-5} 
        For any cycle $C = (V_C, E_C)$ in $\widehat G(uv,U)$, there exists a $(uv,U)$-connected subgraph $H=(V_H,E_H)$ of $G$ such that
        \begin{align*}
            & (\exists u' \in V_C \cap U : u' \notin V_H) \\
            \text{and} \quad & (\exists v' \in V_C \setminus U : v' \notin V_H) \enspace .
        \end{align*}
    \end{enumerate}
\end{theorem}

\proofref{thm:cut-facets}

\begin{figure}
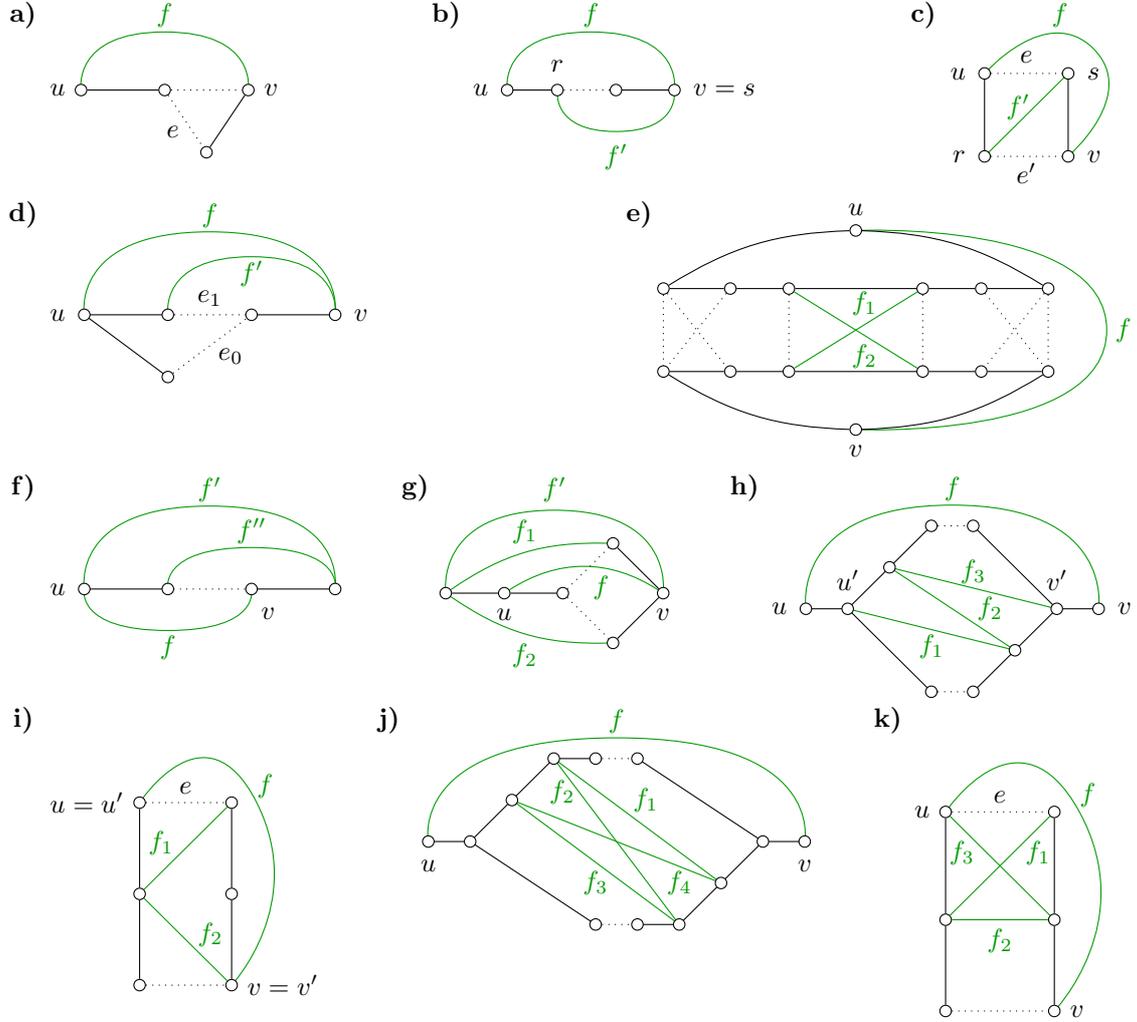

    \makebox(0,0){\textbf{a)}}
    \input{figures/cut-examples/a.tex}
    \hfill
    \makebox(0,0){\textbf{b)}}
    \input{figures/cut-examples/b.tex}
    \hfill
    \makebox(0,0){\textbf{c)}}
    \input{figures/cut-examples/c.tex}
    \\
    \makebox(0,0){\textbf{d)}}
    \input{figures/cut-examples/d.tex}
    \hfill
    \makebox(0,0){\textbf{e)}}
    \input{figures/cut-examples/e.tex}
    \\
    \makebox(0,0){\textbf{f)}}
    \input{figures/cut-examples/f.tex}
    \hfill
    \makebox(0,0){\textbf{g)}}
    \input{figures/cut-examples/g.tex}
    \hfill
    \makebox(0,0){\textbf{h)}}
    \input{figures/cut-examples/h.tex}
    \\
    \makebox(0,0){\textbf{i)}}
    \input{figures/cut-examples/i.tex}
    \hfill
    \makebox(0,0){\textbf{j)}}
    \input{figures/cut-examples/j.tex}
    \hfill
    \makebox(0,0){\textbf{k)}}
    \input{figures/cut-examples/k.tex}
    \caption[Examples of violated necessary conditions for facet-defining cut inequalities]{Depicted above are graphs $G = (V, E)$ (in black)
    and $\widehat G = (V, E \cup F)$ ($F$ in green),
    distinct nodes $u, v \in V$ and a $uv$-cut $\delta(U)$ in $G$ (as dotted lines).
    In any of the above examples, one condition of 
    \Cref{thm:cut-facets}
    is violated and thus, 
    $\Sigma(uv,U)$ is not a facet of the lifted multicut polytope $\lmc(G,\widehat{G})$.
    \textbf{a)} Condition~\ref{cond:cut-1} is violated for $e$.
    \textbf{b)} Condition~\ref{cond:cut-2} is violated as $r$ and $s$ are connected in any $(uv,U)$-connected component.
    \textbf{c)} Condition~\ref{cond:cut-2} is violated as $r$ and $s$ are not connected in any $(uv,U)$-connected component.
    \textbf{d)} Condition~\ref{cond:cut-2} is violated. Specifically, $\delta^0(U) = \{e_0\}$ and $\delta^1(U) = \{e_1\}$ in the proof of
    \Cref{thm:cut-facets}.
    \textbf{e)} Condition~\ref{cond:cut-2} is violated for $F' = \{f_1,f_2\}$.
    \textbf{f)} Condition~\ref{cond:cut-3} is violated.
    \textbf{g)} Condition~\ref{cond:cut-3} is violated for $F' = \{f_1,f_2\}$ and $k=1$.
    \textbf{h)} Condition~\ref{cond:cut-4} is violated for the $u'v'$-path with edges $\{f_1, f_2, f_3\}$.
    \textbf{i)} Condition~\ref{cond:cut-4} is violated for the $u'v'$-path with edges $\{e, f_1, f_2\}$.
    \textbf{j)} Condition~\ref{cond:cut-5} is violated for the cycle with edges $\{f_1, f_2, f_3, f_4\}$.
    \textbf{k)} Condition~\ref{cond:cut-5} is violated for the cycle with edges $\{e, f_1, f_2, f_3\}$.
    }
    \label{figure:violated-conditions}
\end{figure}

Examples in which one of the Conditions \ref{cond:cut-1}--\ref{cond:cut-5} is violated are depicted in \Cref{figure:violated-conditions}.
On the contrary, the example depicted in \Cref{fig:uv-U-connected} satisfies all conditions from \Cref{thm:cut-facets}. 
And indeed, the respective cut inequality defines a facet of the corresponding lifted multicut polytope, which we have verified numerically by computing the affine dimension of all feasible solutions that satisfy the cut inequality with equality.

\section{Multicuts lifted from trees and paths}\label{sec:lmc-trees-and-paths}

In this section, we study multicuts of a tree $T=(V,E)$ lifted to the complete graph. 
For any pair of distinct nodes $u,v \in V$, we denote by $P_{uv}=(V_{uv}, E_{uv})$ the unique path from $u$ to $v$ in $T$.
Moreover, we write $d(u,v)$ for the distance of $u$ and $v$ in $T$, i.e. the length of $P_{uv}$.

\begin{proposition}\label{cor:standard-relaxation-lmc-trees}
    The lifted multicut polytope $\lmc(T)$ with respect to a tree $T=(V,E)$ is the convex hull of all $x \in \{0,1\}^{\binom{V}{2}}$ that satisfy 
    \begin{align}
        x_{uv} & \leq \sum_{e \in E_{uv}} x_e && \forall u,v \in V, \: d(u,v) \geq 2  \label{eq:path-ineq-tree} \\
        x_e & \leq x_{uv} && \forall u,v \in V, \: d(u,v) \geq 2, \; \forall e \in E_{uv} 
        \enspace . \label{eq:cut-ineq-tree}
    \end{align}
\end{proposition}

\begin{proof}
    The claim follows from \Cref{prop:lifted-multicut-polytope-inequalities}:
    Since $T$ is a tree, there are no cycle inequalities \eqref{eq:lifted-multicut-cycle}.
    Moreover, the path and cut inequalities \eqref{eq:lifted-multicut-path} and \eqref{eq:lifted-multicut-cut} simplify to \eqref{eq:path-ineq-tree} and \eqref{eq:cut-ineq-tree}.
\end{proof}

\subsection{Tree partition problem}\label{sec:tree-partition-problem}

The lifted multicut problem \eqref{eq:lifted-multicut-problem} with respect to a tree $T$ can be stated equivalently as the minimization of a particular multilinear polynomial over binary inputs, which we refer to as the \emph{tree partition problem}.

\begin{definition}[Tree Partition Problem]\label{def:tree-partition-problem}
    Let $T = (V,E)$ be a tree and $\bar \theta \in \R^{\binom{V}{2}}$.
    The optimization problem
    \begin{align}\label{eq:tpp-pbo}
        \min_{z \in \{0,1\}^E}
            \sum_{uv \in \binom{V}{2}}
            \bar \theta_{uv} 
            \prod_{e \in E_{uv}} z_e
        \tag{TPP}
    \end{align}
    is called the instance of the \emph{tree partition problem} with respect to~$T$ and 
    $\bar \theta$.
    If $T$ is a path, then we also refer to \eqref{eq:tpp-pbo} as the \emph{path partition problem} with respect to $T$ and $\bar \theta$.
\end{definition}

It is straightforward to see, by a change of variables, that the tree partition problem \eqref{eq:tpp-pbo} and the lifted multicut problem \eqref{eq:lifted-multicut-problem} for trees are equivalent (up to a constant):

\begin{proposition} \label{prop:tree-lmp}
    The vector $z \in \{0,1\}^E$ is a solution to the instance of \eqref{eq:tpp-pbo} with respect to the tree $T=(V,E)$ and costs $\bar \theta \colon \binom{V}{2} \to \R$ if and only if the unique $x \in \lmc(T)$ such that $x_e = 1 - z_e$ for all $e \in E$ is a solution to the instance of \eqref{eq:lifted-multicut-problem} with respect to $T$ and the cost vector $\theta = - \bar \theta$.
\end{proposition}

\proofref{prop:tree-lmp}

By \Cref{prop:tree-lmp}, \eqref{eq:tpp-pbo} corresponds to the minimization of a specific class of binary multilinear functions.
More precisely, we call any $n$-variate binary multilinear function \emph{tree-sparse} if it can be aligned with a tree such that $n=\abs{E}$ and every non-zero coefficient corresponds to the edge set of a path in the tree.
Similarly, we call it \emph{path-sparse} if the tree is a path.
Tree-sparse binary multilinear functions are exactly those multilinear functions that correspond to tree partition problems \eqref{eq:tpp-pbo}.

The tree partition problem, and thus \eqref{eq:lifted-multicut-problem} for trees, is \textsc{np}-hard in general (\Cref{prop:complexity} below).
However, the path partition problem is solvable in strongly polynomial time \cite{Kernighan1971}. 

\begin{proposition}\label{prop:complexity}
    The tree partition \eqref{eq:tpp-pbo} problem is \emph{\textsc{np}}-hard.
\end{proposition}

\begin{proof}
If $T$ is a star (see \Cref{fig:lifted-trees}a for an example), then \eqref{eq:tpp-pbo} is equivalent to the unconstrained binary quadratic program with $\abs{E}$ variables, which is well-known to be \textsc{np}-hard.
\end{proof}

\subsection{Lifted multicut polytope for trees}\label{sec:lmp-trees}

In this section, we study the facial structure of the lifted multicut polytope $\lmc(T)$ for a tree $T=(V,E)$.
We characterize all canonical facets and offer a relaxation of $\lmc(T)$ that is tighter than the standard relaxation given by \Cref{cor:standard-relaxation-lmc-trees}.
In \Cref{sec:lmp-paths}, we show that our results yield a complete totally dual integral description of the lifted multicut polytope for paths.

\begin{figure}
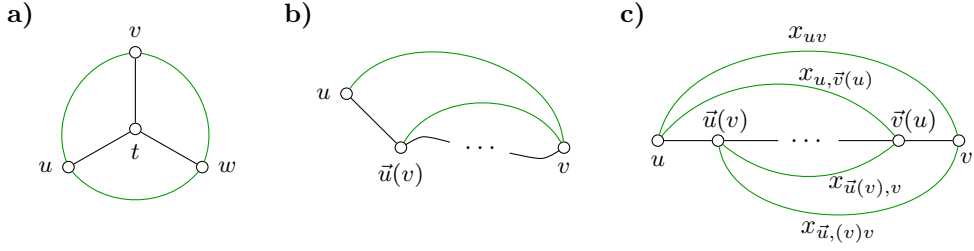

    \center
    \makebox(0,0){\textbf{a)}}
    \imagetop{\input{figures/tree/lifted-star.tex}}
    \hspace{1em}
    \makebox(0,0){\textbf{b)}}
    \imagetop{\input{figures/tree/vec-u-v-example.tex}}
    \hspace{1em}
    \makebox(0,0){\textbf{c)}}
    \imagetop{\input{figures/tree/intersection-inequality-example.tex}}
    \caption{\textbf{a)} A star augmented by additional (green) edges between non-neighboring nodes.
    \textbf{b)} The node $\vec{u}(v)$ is the first internal node on the path~$P_{uv}$. 
    \textbf{c)} A path of length at least three gives rise to an \emph{intersection inequality} \eqref{eq:intersection-inequality}.}
    \label{fig:lifted-trees}
\end{figure}

We denote the standard relaxation of $\lmc(T)$ by
\begin{align*}
    \tpp(T) = 
        \left\{ x \colon \tbinom{V}{2} \to [0,1] \ \middle|\  x \text{ satisfies \eqref{eq:path-ineq-tree} and \eqref{eq:cut-ineq-tree}}\right\} \enspace ,
\end{align*}
which is obtained by dropping the integrality constraints from the definition of $\lmc(T)$.
Given two nodes $u,v \in V$ such that $d(u,v) \geq 2$, let $\vec{u}(v)$ be the first node on the path $P_{uv}$ that is different from both $u$ and $v$ (cf.\ \Cref{fig:lifted-trees}b), and consider the polytope
\begin{align*}
    \tppt(T) = 
        \Big \{ x \colon \tbinom{V}{2} \to [0,1] \; \Big| \; 
            & x_{uv} \leq x_{u,\vec{u}(v)} + x_{\vec{u}(v),v} 
                && \hspace{-1em} \forall u,v \in V, \: d(u,v) \geq 2, \\
            & x_{\vec{u}(v),v} \leq x_{uv} 
                && \hspace{-1em} \forall u,v \in V, \: d(u,v) \geq 2 \Big \} \enspace .
\end{align*}
This description is compact in the sense that it only considers a quadratic number of node triplets, namely those which feature two neighboring nodes and an arbitrary third node. The first inequality in the description of $\tppt(T)$ is depicted in \Cref{fig:lifted-trees}b. The following lemma states that $\tppt(T)$ is indeed a relaxation of $\lmc(T)$ that is at least as tight as $\tpp(T)$.

\begin{proposition} \label{prop:tree-polytope}
    For a tree $T=(V,E)$, we have $\lmc(T) \subseteq \tppt(T) \subseteq \mathsf \tpp(T)$.
\end{proposition}

\proofref{prop:tree-polytope}

In general, the inclusions in \Cref{prop:tree-polytope} are strict, in particular $\tppt(T)$ is a strictly tighter relaxation than $\tpp(T)$. 
For example, consider $T=(V,E)$, the path of length $3$, with $V=\{0,\dots,3\}$ and $E=\{\{0,1\},\{1,2\},\{2,3\}\}$. 
For $x \colon \tbinom{V}{2} \to [0,1]$ with $x_{01} = 0.5$, $x_{12}=0.5$, $x_{23}=0$, $x_{02} = 0.5$, $x_{13}=0.5$ and $x_{03}=1$, we have $x \in \tpp(T)$ but $x_{03} > x_{02} + x_{23}$, i.e.~$x \notin \tppt(P)$. 
For $x \colon \tbinom{V}{2} \to [0,1]$ with $x_{03} = 1$ and $x_e = 0.5$ for all other edges $e \neq \{0,3\}$, we have $x \in \tppt(T)$ but $x \notin \lmc(T)$ (this can be seen, e.g., from the fact that $x$ violates an intersection inequality \eqref{eq:intersection-inequality} which is valid for $\lmc(T)$, see \Cref{lem:tree-intersection-valid}).

\subsection{Facets}\label{sec:lmp-trees-facets}

In this section, we show which inequalities in the definition of $\tppt(T)$ define facets of $\lmc(T)$.
Moreover, we present another type of inequalities associated with paths in~$T$, which define facets of $\lmc(T)$.
We note that further facets can be established through the connection of $\lmc(T)$ to the multilinear polytope and, as a special case, the Boolean quadric polytope \cite{Padberg1989}. 

\begin{proposition}\label{prop:tree-path-facet}
    Let $T = (V,E)$ be a tree and let $u,v \in V$ with $d(u,v) \geq 2$. 
    The inequality
    \begin{align}\label{eq:tree-path}
        x_{uv} \leq x_{u,\vec{u}(v)} + x_{\vec{u}(v),v} 
    \end{align}
    defines a facet of $\lmc(T)$ if and only if $d(u,v) = 2$.
\end{proposition}

\proofref{prop:tree-path-facet}

\begin{proposition}\label{prop:tree-cut-facet}
    Let $T = (V,E)$ be a tree and let $u,v \in V$ with $d(u,v) \geq 2$. 
    The inequality 
    \begin{align}\label{eq:tree-cut}
        x_{\vec{u}(v),v} \leq x_{uv}
    \end{align}
    defines a facet of $\lmc(T)$ if and only if $v$ is a leaf of $T$.
\end{proposition}

\proofref{prop:tree-cut-facet}

\begin{proposition}\label{prop:tree-box-facets}
    Let $T = (V,E)$ be a tree.
    For any distinct $u,v \in V$, the inequality $x_{uv} \leq 1$ defines a facet of $\lmc(T)$ if and only if both $u$ and $v$ are leaves of $T$.
    Moreover, none of the inequalities $0 \leq x_{uv}$ define facets of $\lmc(T)$.
\end{proposition}

\proofref{prop:tree-box-facets}

Next, we present an additional class of facets of $\lmc(T)$.
For any $u,v \in V$ with $d(u,v) \geq 3$ consider the inequality
\begin{align}\label{eq:intersection-inequality}
    x_{uv} + x_{\vec{u}(v),\vec{v}(u)}
    \leq x_{u,\vec{v}(u)} + x_{\vec{u}(v),v} \enspace, 
\end{align}
which we refer to as the \emph{intersection inequality}.
For an illustration, see \Cref{fig:lifted-trees}c.

\begin{lemma}\label{lem:tree-intersection-valid}
    Let $T = (V,E)$ be a tree. Any intersection inequality is valid for $\lmc(T)$.
\end{lemma}

\proofref{lem:tree-intersection-valid}

\begin{theorem}\label{thm:tree-intersection-facet}
    Let $T = (V,E)$ be a tree. Any intersection inequality defines a facet of $\lmc(T)$.
\end{theorem}

\proofref{thm:tree-intersection-facet}

Next, we establish the exact condition under which the intersection inequalities are valid and facet-defining for the lifted multicut polytope $\lmc(G)$, for an arbitrary graph $G$. 
For an example, see \Cref{fig:intersection-arbitrary-g}.

\begin{theorem}\label{thm:intersection-arbitrary-g}
    Let $G=(V,E)$ be a connected graph and let $uu',vv' \in E$ such that $u,u',v,v'$ are pairwise distinct. Then, the intersection inequality 
    \begin{align}\label{eq:intersection-arbitrary-g}
        x_{uv} + x_{u'v'} \leq x_{uv'} + x_{vu'}
    \end{align}
    is valid and facet-defining for $\lmc(G)$ if and only if $\{u,v'\}$ is a $vu'$-separating node set and $\{v,u'\}$ is a $uv'$-separating node set. 
\end{theorem}

\proofref{thm:intersection-arbitrary-g}

\begin{figure}
    \center
    \input{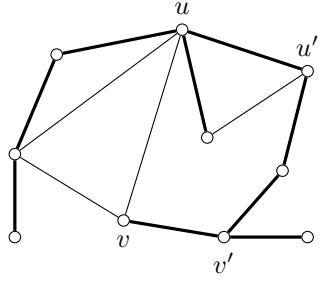}
    \caption{Depicted above are a graph $G$ and four nodes $u,u',v,v'$ such that $uu'$ and $vv'$ are edges in $G$. The conditions of \Cref{thm:intersection-arbitrary-g} are satisfied and the corresponding intersection inequality \eqref{eq:intersection-arbitrary-g} defines a facet of $\lmc(G)$. The thick edges form a spanning tree $T$ of $G$. Inequality \eqref{eq:intersection-arbitrary-g} with respect to $u,u',v,v'$ for $\lmc(G)$ coincides with the intersection inequality \eqref{eq:intersection-inequality} with respect to $u$ and $v$ for $\lmc(T)$.}
    \label{fig:intersection-arbitrary-g}
\end{figure}

\subsection{Lifted multicut polytope for paths}\label{sec:lmp-paths}

In this section, we show that the facets established in the previous section yield a complete description of $\lmc(T)$ when $T$ is a path.
To this end, suppose that $V = \{ 0, \dotsc, n \}$ and $E = \big \{ \{i,i+1\} \mid i \in \{0, \dotsc, n -1 \} \big \}$ are linearly ordered.
Therefore, $T = (V,E)$ is path.
We consider only paths of length $n \geq 2$, since for $n = 1$, the polytope $\lmc(T) = [0,1]$ is simply the unit interval.
Let $\ppp(n)$ be the convex hull of all $x \in \R^{\binom{V}{2}}$ that satisfy the system
\begin{align}
    x_{0n} & \leq 1 \label{eq:path-box} 
    \\
    x_{in} & \leq x_{i-1,n} & \forall i \in \{1, \dotsc, n-1\} \label{eq:path-cut-right} \\
    x_{0i} & \leq x_{0,i+1} & \forall i \in \{1, \dotsc, n-1\} \label{eq:path-cut-left} \\
    x_{i-1,i+1} & \leq x_{i-1,i} + x_{i,i+1} & \forall i \in \{1, \dotsc, n-1\} \label{eq:path-triangle} \\
    x_{j,k} + x_{j+1,k-1} & \leq x_{j+1,k} + x_{j,k-1} & \forall j,k \in \{0, \dotsc, n\}, \: j < k-2 \label{eq:path-intersection} 
    \enspace .
\end{align}
Note that this system consists precisely of those inequalities that we have shown to define facets of $\lmc(T)$ in the previous section.
We first prove that $\ppp(n)$ indeed yields a relaxation of $\lmc(T)$.

\begin{lemma} \label{lem:path-polytope}
    For a path $T$ of length $n$, we have $\lmc(T) \subseteq \ppp(n) \subseteq \tppt(T)$.
\end{lemma}

\proofref{lem:path-polytope}

As our main result of this section, we prove that the system defining $\ppp(n)$ is in fact a complete description of $\lmc(T)$ and, moreover, it is \emph{totally dual integral}.
For an extensive reference on the subject of total dual integrality, we refer the reader to \citet{conforti2014integer}.

\begin{theorem}\label{thm:path-tdi}
    The system \eqref{eq:path-box} -- \eqref{eq:path-intersection} is totally dual integral.
\end{theorem}

\proofref{thm:path-tdi}

\begin{remark}
    The constraint matrix corresponding to the system \eqref{eq:path-box}--\eqref{eq:path-intersection} is in general \emph{not} totally unimodular. 
    A minimal example is the path of length $4$.
\end{remark}

\begin{corollary}\label{cor:complete-description}
    For a path $T$ of length $n$, we have $\lmc(T) = \ppp(n)$.
\end{corollary}

\begin{proof}
    As the polytope corresponding to any totally unimodular system is integral \cite{Edmonds1977}, the claim follows from \Cref{thm:path-tdi} and \Cref{lem:path-polytope}.
\end{proof}

The path partition problem admits a smaller representation as a set partition problem where the variables do not correspond to edges but to connected subpaths. 

\begin{proposition}\label{prop:sequential-set-partition-problem}
    The path partition problem $\min \; \{\theta^\top x \mid x \in \ppp(n)\}$ is equivalent to the \emph{sequential set partition} problem
    \begin{align}
        \min 
            && \Theta_{0n} - \Theta^\top \lambda \tag{SSP} \label{eq:ssp} \\
        \textup{subject to} 
            && \sum_{0 \leq i \leq k \leq j \leq n} \lambda_{ij} &= 1 && \forall k \in \{1,\dots,n-1\} \nonumber \\
            && \lambda_{ij} & \geq 0 && \forall i,j \in \{0,\dots,n\}, i \leq j \nonumber 
    \end{align}
    with $\Theta_{ij} = \sum_{i \leq k < \ell \leq j} \theta_{k\ell}$ for all $i,j \in \{0,\dots,n\}$ with $i \leq j$.
\end{proposition}

\proofref{prop:sequential-set-partition-problem}

Each variable $\lambda_{ij}$ in \eqref{eq:ssp} corresponds to the subpath containing nodes $i$ to $j$. 
Problem \eqref{eq:ssp} is precisely the formulation of the path partition problem used by \citet{Joseph1997}. 
It admits a quadratic number of variables and a linear number of constraints (in contrast to a quadratic number of constraints in the description of $\ppp(n)$).
The constraint matrix satisfies the \emph{consecutive-ones} property with respect to its columns.
Therefore, the integrality constraint need not be enforced, since the constraint matrix is totally unimodular.

Before closing this section, we observe that Theorem~\ref{thm:path-tdi} leads to a second way of solving the path partition problem by means of linear programing, alongside the formulation \eqref{eq:ssp} from \citet{Joseph1997}.
This complements the fully combinatorial result of \citet{Kernighan1971}.
Moreover, the system \eqref{eq:path-box}--\eqref{eq:path-intersection} satisfies the assumptions of \citet{Tardos1986} which imply that, for any linear objective function, the corresponding linear program over $\ppp(n)$ can be solved in strongly polynomial time.

\subsection{Lifting to an arbitrary graph} \label{sec:arb-graph}

So far, we have lifted multicuts of a tree $T = (V,E)$ to the complete graph with the node set $V$.
In this section, we study the lifted multicut problem for multicuts lifted from a tree $T$ to an arbitrary augmented graph $\widehat{G}$.
This means that we are not interested in understanding whether $u$ and $v$ are in distinct components (of the decomposition defined by multicut) for \emph{every} pair $\{u, v\}$ of nodes, but we are interested in understanding this property \emph{only for a subset} of pairs of nodes.

More formally, let $T = (V, E)$ be a tree and let $F \subseteq \binom V2 \setminus E$ contain all the pairs of non-neighboring nodes for which we wish to make explicit whether they are in the same component.
The characteristic vector $x \in \{0,1\}^E$ of a multicut is lifted to the set $\{0,1\}^{E \cup F}$ with the identical meaning as in \Cref{sec:tree-partition-problem}.
Hence, the lifted multicut polytope with respect to $T$ and $\widehat{G}$ becomes the convex hull of the vectors $x \in \{0,1\}^{E \cup F}$ that satisfy the following path and cut inequalities:
\begin{align}
    x_{uv} 
        & \leq \sum_{e \in E_{uv}} x_e 
        && \forall uv \in F \label{eq:path-lifted} \\
    x_e 
        & \leq x_{uv} 
        && \forall uv \in F \quad \forall e \in E_{uv} \label{eq:cut-lifted} 
\end{align}

As before, we denote the lifted multicut polytope with respect to $T$ and $\widehat{G} = (V, E \cup F)$ by $\lmc(T,\widehat{G})$.
The lifted multicut problem, tree partition problem, and path partition problem are then defined accordingly.

Two remarks are in order.
On the one hand, the tree partition problem when lifting to an arbitrary graph is \textsc{np}-hard, since lifting to the complete graph is a special case.
On the other hand, the path partition problem is still solvable in polynomial time, for instance, by lifting to the complete graph and setting the coefficients to zero for all edges in $\tbinom{V}{2} \setminus F$.
Nevertheless, we study the case of lifting to a general graph from a polyhedral perspective.

\subsubsection{Generalized intersection inequalities}

In this section, we introduce a class of valid inequalities for the lifted multicut polytope when lifting to an arbitrary augmented graph.
We call these the \emph{generalized intersection inequalities}.
In fact, we will see that the intersection inequalities \eqref{eq:intersection-inequality} are a special case, the one in which all pairs of nodes in \eqref{eq:intersection-inequality} are in $E \cup F$.
We remark that also the simpler inequalities \eqref{eq:tree-path} and \eqref{eq:tree-cut} can be written as generalized intersection inequalities, granted that all pairs of nodes that occur in \eqref{eq:tree-path} and \eqref{eq:tree-cut}, respectively, are in $E \cup F$.

To begin with, let $uv \in E \cup F$, let $K$ be a finite, ordered index set, i.e.~$K=\{1,...,|K|\}$, and for $k \in K$ let $\{u_k, v_k\} \in E \cup F$ such that $E_{uv} \cap E_{u_kv_k} \neq \emptyset$, i.e.~the paths $P_{uv}$ and $P_{u_k v_k}$ must share at least one edge. 
Starting from these paths, we define the following sets: 
\begin{align}\label{eq:def-N}
    & N_1 = \emptyset \nonumber \\
    & N_k = E_{uv} \cap E_{u_k v_k} \cap 
        \left( \bigcup_{0 < i < k} E_{u_i v_i}\right) 
        \quad \forall k \in \{2,\dots,|K|\} 
\end{align}
Furthermore, for all $k \in K$ for which $N_k \neq \emptyset$, let $f_k \in E \cup F$ such that $E_{f_k} \subseteq N_k$ (where $E_{f_k}$ is the set of edges on the unique path in $T$ between the endpoints of $f_k$).
Then, we define the corresponding \emph{generalized intersection inequality} as
\begin{equation}\label{eq:gen-ineq}
    x_{uv} + \sum_{\substack{k \in K : \\N_k \neq \emptyset}} x_{f_k} 
    \leq \sum_{k \in K} x_{u_k v_k} + 
    \sum_{e \in E_{uv} \setminus \bigcup_{k \in K} E_{u_k v_k}} x_e \enspace . 
\end{equation}

\begin{proposition}\label{prop:general-intersection-valid}
    The generalized intersection inequalities are valid for $\lmc(T,\widehat{G})$.
\end{proposition}

\proofref{prop:general-intersection-valid}

Next, we present two conditions every facet-defining generalized intersection inequality satisfies.

\begin{proposition}\label{prop:nec-cond-gen-ineq}
    A generalized intersection inequality defines a facet of $\lmc(T,\widehat{G})$ only if the following necessary conditions hold:
    \begin{enumerate}[(i)]
        \item \label{cond:generalized-intersection-1}
        For all $k \in K$ such that $N_k \neq \emptyset$, $E_{f_k}$ is maximal with respect to set inclusion.
        \item \label{cond:generalized-intersection-2}
        For all $k,k' \in K$, $k \neq k'$ such that $E_{f_k}, E_{f_{k'}} \subseteq N_k \cap N_{k'}$, we have $f_k = f_{k'}$. 
    \end{enumerate}
\end{proposition}

\proofref{prop:nec-cond-gen-ineq}

In the beginning of this section, we have claimed that the generalized intersection inequalities contain, as a special case, the intersection inequalities for trees.
Let us recall now the form of the intersection inequalities~\eqref{eq:intersection-inequality}.
For $u,v \in V$ such that $d(u,v) \geq 3$, the intersection inequality is $x_{uv} + x_{\vec{u}(v),\vec{v}(u)} \leq x_{u,\vec{v}(u)} + x_{\vec{u}(v),v}$.
We can assume that $\{u,v\}$, $\{\vec{u}(v),\vec{v}(u)\}$, $\{u,\vec{v}(u)\}$, and $\{\vec{u}(v),v\}$ all belong to $E \cup F$.
Otherwise, the intersection inequality would not be well-defined.
Let $|K| = 2$ and let $\{u_1,v_1\} = \{u,\vec{v}(u)\}$ and $\{u_2,v_2\} = \{\vec{u}(v),v\}$.
Note that $E_{uv} \setminus (E_{u,\vec{v}(u)} \cup E_{\vec{u}(v),v}) = \emptyset$.
Thus, the last sum in \eqref{eq:gen-ineq} is vacuous.
By \eqref{eq:def-N}, we have $N_1 = \emptyset$ and $N_2 = E_{uv} \cap E_{\vec{u}(v),v} \cap E_{u,\vec{v}(u)} = E_{\vec{u}(v),\vec{v}(u)}$.
Hence, $f_2$ can be chosen as the edge $\{\vec{u}(v), \vec{v}(u)\}$.
In this case, \eqref{eq:gen-ineq} becomes precisely $x_{uv} + x_{\vec{u}(v),\vec{v}(u)} \leq x_{u,\vec{v}(u)} + x_{\vec{u}(v),v}$.

We conclude this section by discussing an example of a generalized intersection inequality that is not an intersection inequality.

\begin{figure}
    \centering
    \input{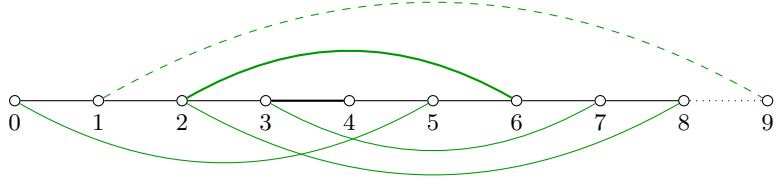}
    \caption{Illustrated above is the generalized intersection inquality of \Cref{example:gen-int-ineq}.
    Depicted in black is a path, augmented by additional edges $F$, depicted in green, to form the augmented graph $\widehat{G}$.
    The edge $P_{uv}$ defined in \Cref{example:gen-int-ineq} is depicted as a dashed green line. 
    The edges corresponding to some $f_k$ are depicted as thick lines (black and green). 
    The edges representing paths indexed in $K$ in \Cref{example:gen-int-ineq} are depicted as solid thin lines in green.
    Lastly, the only edge in $P_{uv}$ not contained in any other paths is the edge $\{8,9\}$ depicted as a dotted line in black.}
	\label{fig:example-gen-intersec}
\end{figure}

\begin{example}
\label{example:gen-int-ineq}
Let $T = (V, E)$, where $V = \{0, \dots, 9\}$, $E = \{ \{i, i+1\} \mid i = 0, \dots, 8 \}$, and $F = \{ \{0,5\}$, $\{1,9\}$, $\{2,8\}$, $\{2,6\}$, $\{3,7\} \}$.
The augmented graph $\widehat{G} = (V, E \cup F)$ is depicted in \Cref{fig:example-gen-intersec}.
Consider the generalized intersection inequality
\begin{equation}\label{eq:gen-example}
    x_{19} + x_{34} + x_{26} \leq x_{05} + x_{37} + x_{28} + x_{89}
    \enspace .
\end{equation}
Here, $uv = \{1,9\}$, $|K| = 3$ and $\{u_1,v_1\} = \{0,5\}$, $\{u_2,v_2\} = \{3,7\}$, and $\{u_3,v_3\} = \{2,8\}$.
By the definition of the sets $N_k$ in \eqref{eq:def-N}, it follows that $N_1 = \emptyset$, $N_2 = E_{19} \cap E_{37} \cap E_{05} = \{ \{3,4\}$, $\{4,5\} \}$, and $N_3 = E_{19} \cap E_{28} \cap (E_{05} \cup E_{37}) = \{ \{2,3\}$, $\{3,4\}$, $\{4,5\}$, $\{5,6\}$, $\{6,7\} \}$.
Then, we must pick $f_2$ to be either an edge in $N_2$ or an edge in $F$ such that the path between the endpoints of $f$ is completely contained in $N_2$. 
Here, for example, $f_2$ is chosen to be just one edge in $N_2$, namely the edge $\{3,4\}$. 
On the other hand, we choose $f_3$ to be an edge in $F$ such that the endpoints are completely contained in $N_3$.
We observe that $\{2,6\} \in F$, and that $E_{26} \subseteq N_3$.
Therefore, we set $f_3 = \{2, 6\}$.

In \Cref{fig:example-gen-intersec}, we can observe the relation between the ``center'' of the inequality, i.e.~$uv = \{1,9\}$, and the other terms.
We remark that the union of the paths indexed by $K$ does not need to be a subpath of $P_{uv}$. In fact, $\{0,1\} \in \bigcup_{k \in K} E_{u_k v_k} \setminus E_{uv}$. 
Nor must the union of these paths cover $P_{uv}$.
This last fact is highlighted in the example, as we see that $\{8,9\} \in E_{uv}$ but $\{8,9\}$ is not on any of the paths $P_{u_k v_k}$.
Hence, in \eqref{eq:gen-example}, we have
\begin{align*}
    & x_{uv} = x_{19} \enspace , 
    & \sum_{\substack{k \in K : \\ N_k \neq \emptyset}} x_{f_k} = x_{34} + x_{26} \enspace , \\
    & \sum_{k \in K} x_{u_k v_k} = x_{05} + x_{37} + x_{28} \enspace , 
    & \sum_{e \in E_{uv} \setminus \bigcup_{k \in K} E_{u_k v_k}} x_e = x_{89} \enspace .
\end{align*}

By means of Fourier-Motzkin elimination (see e.g.~\citet{Ziegler2000}, implemented e.g.~in \citet{porta}), we compute all facets of this lifted multicut polytope with respect to $T$ and $\widehat{G}$. 
It has a total of 56 facets.
Of these, 11 are given by the bounds on the variables, 
16 are cut inequalities \eqref{eq:cut-lifted}, 
3 are path inequalities \eqref{eq:path-lifted}, and 
26 facets are induced by generalized intersection inequalities. 
In this example, it is impossible to construct ``regular'' intersection inequalities \eqref{eq:intersection-inequality}.
One of the 26 facets is given by \eqref{eq:gen-example}.
In this example, the generalized intersection inequalities, together with the bounds on the variables, the path and cut inequalities, provide a perfect formulation of $\lmc(T,\widehat{G})$.
We close the example by observing that this does not happen when $T$ is a tree.
In fact, consider $T$ and $\widehat{G}$ as depicted in \Cref{fig:tree-not-beta}.
Then, it can be checked that the inequality $x_{14}+x_{25}+x_{46} \leq x_{12} + x_{13} + x_{34} + x_{36} + 1$ defines a facet of $\lmc(T,\widehat{G})$ and is not a generalized intersection inequality. 
\end{example}

\subsection{Connections to the multilinear polytope}\label{sec:multilinear}

In~\Cref{sec:tree-partition-problem}, we have seen that the lifted multicut problem on trees is equivalent to the tree partition problem, which in turn is a special case of binary multilinear optimization.
Here, we explore this connection in more detail.
To begin with, we recall elementary notions from the field of binary multilinear optimization:

\paragraph{Notation.}
Given a binary multilinear optimization problem, it is common practice to introduce an extra variable for each monomial of length at least 2, see, e.g., \citet{For60,GloWoo74}.
Once this is done, the object $\mathcal{S} = \{ (x,y) \mid y_I = \prod_{i \in I} x_i, I \in \mathcal{I}, x \in \{0,1\}^n \}$ is studied.
In order to exploit the structure of the problem better, a hypergraph representation can be used, where $H = (\bar V,\bar E)$ is a hypergraph if $\bar E$ consists of subsets of $\bar V$. 
We will say that $\bar V$ is the set of nodes of $H$, while $\bar E$ is the set of edges of $H$.
Note that if all the edges contain only two nodes, then the hypergraph is simply a regular graph. 
Each node of $H$ represents a variable of the multilinear form, and each edge of $H$ corresponds to a monomial.
By using the hypergraph representation, the multilinear polytope \mpg{H} is then defined by \citet{dPKha17} as the convex hull of $\{ z \in \{0,1\}^{\bar V \cup \bar E} \mid z_e = \prod_{v \in e} z_v, e \in \bar E \}$.

With this notation set in place, we observe from \Cref{sec:tree-partition-problem} that each edge of the tree $T$ in the lifted multicut problem becomes a node in the hypergraph $H$ representing the tree partition problem.
Furthermore, every augmented edge $f \in F$ in the lifted multicut problem becomes an edge in $H$ containing the edges on the unique path between the endpoint of $f$.
Hence, $\bar V = E$, where $E$ is the set of edges of the tree $T$, and $\bar E = \{ E_{uv} \mid uv \in F \}$, where $F$ has the same meaning as in \Cref{sec:arb-graph}.
\Cref{prop:tree-lmp} states that the lifted multicut problem on trees is equivalent to the tree partition problem when applying the affine transformation $x_e = 1 - z_e$ for $e \in E = \bar V$, and $x_{uv} = 1 - z_{E_{uv}}$ for $uv \in F$ and hence $E_{uv} \in \bar E$.

\subsubsection{Inequalities}\label{sec:multilinea-ineq}

We begin our analysis by establishing correspondences in the multilinear setting of the inequalities defined in \Cref{sec:tree-partition-problem,sec:lmp-trees}, starting with path and cut inequalities \eqref{eq:path-ineq-tree} and \eqref{eq:cut-ineq-tree}.
It is easy to see that these sets of inequalities lead to the inequalities in the standard linearization of the multilinear sets that are different from the bounds $0 \leq z \leq 1$, see, e.g., \citet{dPKha18SIOPT}.
In particular, inequalities \eqref{eq:path-ineq-tree} become $\sum_{e \in E_{uv}} z_e - z_{E_{uv}} \leq |E_{uv}| - 1$ for every $E_{uv} \in \bar E$, after applying the affine transformation described in the previous paragraph.
Similarly, inequalities \eqref{eq:cut-ineq-tree} correspond to the inequalities $z_{E_{uv}} \leq z_e$ defined for all $E_{uv} \in \bar E$, for all $e \in E_{uv}$.

Now, we move on to considering the facets identified in \Cref{sec:lmp-trees}, namely inequalities \eqref{eq:tree-path}, \eqref{eq:tree-cut}, \eqref{eq:intersection-inequality}, and their affine transformations in the multilinear setting. 
Let us start from inequality \eqref{eq:tree-path}. 
It is easy to see that the affine transformation results in $z_{E_{u,\vec{u}(v)}} + z_{E_{\vec{u}(v),v}} - z_{E_{uv}} \leq 1$. 
This inequality is a flower inequality with the edge $E_{uv}$ as a center and only one adjacent edge, which is $E_{\vec{u}(v),v}$.
In particular, this edge contains all the nodes in $E_{uv}$ except for one.
This node is precisely the one corresponding to $\{u,\vec{u}(v)\}$. 
We refer to \citet{dPKha18SIOPT} for a thorough introduction to flower inequalities.

Consider now inequality \eqref{eq:tree-cut}. 
After the transformation, this becomes $z_{E_{uv}} - z_{E_{\vec{u}(v),v}} \leq 0$, a flower inequality as well. 
Here, the center is $E_{\vec{u}(v),v}$ and the only adjacent edge is $E_{uv}$ that strictly contains $E_{\vec{u}(v),v}$.
Note that this is, in particular, also a 2-link inequality as defined by \citet{CraRod17}. 
We remind the reader that 2-link inequalities have been generalized by the flower inequalities.

Next, we look into inequality \eqref{eq:intersection-inequality}.
This becomes $-z_{E_{\vec{u}(v),\vec{v}(u)}} + z_{E_{u,\vec{v}(u)}} + z_{E_{\vec{u}(v),v}} - z_{E_{uv}} \leq 0$, once the transformation is applied.
We observe that this inequality resembles the running intersection inequalities introduced by \citet{dPKha21}.
In fact, there is a center, $E_{uv}$, and two neighbors, $E_{u,\vec{v}(u)}$, $E_{\vec{u}(v),v}$, that satisfy the running intersection property.
However, $E_{\vec{u}(v),\vec{v}(u)}$ is in general an edge contained in $E_{u,\vec{v}(u)} \cap E_{\vec{u}(v),v}$ rather than a node, unless $d(u,v) = 3$ in $T$.
We discuss this connection more in depth now, in connection with the generalized intersection inequalities that contain the intersection inequalities, as we have seen in the previous section. 

Consider the generalized intersection inequalities~\eqref{eq:gen-ineq} that are valid for the more general case in which we lift the tree $T$ to an arbitrary augmented graph $\widehat{G}$.
When we apply the affine transformation from the lifted multicut polytope to the multilinear polytope, an arbitrary inequality of the type~\eqref{eq:gen-ineq} becomes
\begin{multline}\label{eq:gen-run-int}
    - \sum_{\substack{k \in K : \\ N_k \neq \emptyset}} z_{f_k} 
    + \sum_{e \in E_{uv} \setminus \bigcup_{k \in K} E_{u_k v_k}} z_e 
    + \sum_{k \in K} z_{E_{u_k v_k}} - z_{E_{uv}} \leq \\
    \left|E_{uv} \setminus \bigcup_{k \in K} E_{u_k v_k} \right| 
    + |K| - \bigl|\{k \in K \mid N_k \neq \emptyset\}\bigr| - 1 \enspace .
\end{multline}
Note that \eqref{eq:gen-run-int} includes the running intersection inequalities.
The interested reader will have noticed the similarities between our definition of the sets $N_k$
in \eqref{eq:def-N} and the definition of $N(e_0 \cap e_k)$ by \citet{dPKha21}.
Moreover, our proof of the validity of the generalized intersection inequalities, \Cref{prop:general-intersection-valid}, uses the same technique as the proof by \citet{dPKha21} of the validity of the running intersection inequalities.
At the same time, there are several differences between the two classes of inequalities:
1.~In the generalized intersection inequalities, we do not take into account the number of components of $\tilde G$, as we do not need it in order to determine the right-hand side. 
In fact, all terms in the generalized intersection inequalities are variables. 
2.~We do not assume the running intersection property.
This also implies that, in general, the right-hand side of \eqref{eq:gen-run-int} is not equal to the number of components of $\tilde G$ minus $1$, like instead happens when the running intersection property holds for the sets $N(e_0 \cap e_k)$. 
3.~Some of the paths might comprise only one edge (which would become simply one node in the multilinear setting and not an edge, as required instead by the definition of running intersection inequalities).
4.~We allow $u_k \in N(e_0 \cap e_k)$ to possibly be an edge of the hypergraph (not just a node).

We close this section by showing that the generalized intersection inequalities are not a trivial generalization of the running intersection inequalities. 
For example, the inequality \eqref{eq:gen-example} of \Cref{example:gen-int-ineq} is not a running intersection inequality, after the affine transformation is applied.
Recall, however, that it defines a facet of $\lmc(T,\widehat{G})$.
Observe that the running intersection property does not hold for the order of the paths chosen, i.e. $P_{05}$, $P_{37}$, and finally $P_{28}$.
That property holds only if we change the order in which we consider the paths, but then it would be impossible to have $x_{26}$ in the left-hand side and hence obtain the facet-defining inequality \eqref{eq:gen-example}.
Moreover, out of all the necessary conditions and sufficient conditions for a running intersection to be facet-defining, only one of the necessary conditions still holds for the generalized intersection inequalities.
Namely, this is Condition~\ref{cond:generalized-intersection-2} in \Cref{prop:nec-cond-gen-ineq}.
All other necessary conditions, as well as the sufficient condition, described by \citet{dPKha21} turn out to be violated in this case.

\subsubsection{Path partition problem}\label{sec:multilinear-path-beta}

When we assume that $T$ is a path, we can establish more:
We can prove that the path partition problem is a special case of binary multilinear optimization on $\beta$-acyclic hypergraphs, which are hypergraphs that do not contain any $\beta$-cycle. 
For a complete introduction to the different types of cycles that can arise in hypergraphs, we refer the reader to \citet{Fag83}.
\citet{dPDiG20} show that binary multilinear optimization on $\beta$-acyclic hypergraphs can be solved in strongly polynomial time.
However, their approach is not polyhedral, and a complete description of the multilinear polytope in the $\beta$-acyclic setting is currently unknown.

In order to show that the path partition problem can be formulated as a $\beta$-acyclic binary multilinear problem, we use the characterization of $\beta$-acyclic hypergraphs that can be found in \citet{Dur12}.
This characterization is based on the definition of a nest point.
In any hypergraph $H = (\bar V, \bar E)$, a node $\bar v$ is called a \emph{nest point} if, for every two edges $\bar e, \bar f$ that contain $\bar v$, we have $\bar e \subseteq \bar f$ or $\bar f \subseteq \bar e$.

\begin{theorem}[\citet{Dur12}]\label{thm:charact-beta}
    A hypergraph $H$ is $\beta$-acyclic if and only if, after removing successively a nest point, we obtain the empty hypergraph.
\end{theorem}

\begin{proposition}\label{prop:path-partition-beta-acyclic}
    Let $T = (V,E)$ be a path, and let $\widehat{G} = (V, E \cup F)$ be an augmentation of $T$. 
    Then, the corresponding path partition problem is represented by a $\beta$-acyclic hypergraph.
\end{proposition}

\proofref{prop:path-partition-beta-acyclic}

\begin{figure}
    \centering
    \input{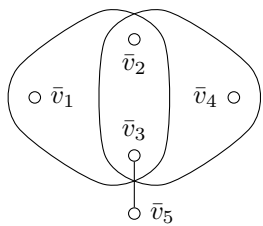}
    \caption{Depicted above is a $\beta$-acyclic hypergraph $H$ for which the $\beta$-acyclic binary multilinear problem represented by $H$ is not a path partition problem.}
    \label{fig:example-beta}
\end{figure}

On other hand, not all $\beta$-acyclic binary multilinear problems can be represented via a path partition problem.
For example, consider the hypergraph $H = (\bar V, \bar E)$ depicted in \Cref{fig:example-beta}, where $\bar V = \{\bar v_1$, $\bar v_2$, $\bar v_3$, $\bar v_4$, $\bar v_5\}$ and $\bar E = \{ \{ \bar v_1, \bar v_2, \bar v_3 \}$, $\{ \bar v_2, \bar v_3, \bar v_4 \}$, $\{ \bar v_3, \bar v_5\} \}$.
By \Cref{thm:charact-beta}, it can be checked easily that $H$ is $\beta$-acyclic.
Assume that $H$ is the hypergraph corresponding to a path partition problem.
Then, $T = (V,E)$ is a path with 6 nodes $\{0,1,2,3,4,5\}$ and 5 edges $\{\{0,1\}$, $\{1,2\}$, $\{2,3\}$, $\{3,4\}$, $\{4,5\} \}$.
Observe that the paths represented by $\{ \bar v_1, \bar v_2, \bar v_3 \}$, $\{ \bar v_2, \bar v_3, \bar v_4 \}$ contain three edges each, two of them in common.
Without loss of generality, we can hence assume that the path represented by $\{ \bar v_1, \bar v_2, \bar v_3 \}$ is $P_{03}$, and similarly, that the path corresponding to $\{ \bar v_2, \bar v_3, \bar v_4 \}$ is $P_{14}$.
The third edge of $H$ is $\{ \bar v_3, \bar v_5\}$, so its corresponding path in $T$ must contain the only edge in $T$ that is neither on $P_{03}$ nor on $P_{14}$, together with one edge that is in common between $P_{03}$ and $P_{14}$.
Thus, $\{ \bar v_3, \bar v_5\}$ corresponds to either $\{ \{1,2\}, \{4,5\} \}$ or $\{ \{2,3\}, \{4,5\} \}$.
However, none of them induces a path, since the two edges in both possibilities are not adjacent.
We conclude that the $\beta$-acyclic binary multilinear problem represented by $H$ is not a path partition problem.

Lastly, we observe that the tree partition problem cannot be represented by $\beta$-acyclic hypergraphs in general.
In fact, we can consider $T = (V,E)$ with $V = \{ 1,2,3,4,5,6 \}$ and $E = \{ \{1,2\}$, $\{1,3\}$, $\{3,4\}$, $\{4,5\}$, $\{5,6\} \}$, and $\widehat{G} = (V, E \cup F)$ where $F = \{ \{1,4\}$, $\{2,5\}$, $\{3,6\}$, $\{4,6\} \}$.
Then, we construct the hypergraph $H$ representing the corresponding tree partition problem.
It remains to observe that the three nodes in $H$ corresponding to $\{1,3\}$, $\{3,4\}$, $\{3,5\}$ together with the three edges in $H$ related to the paths $P_{14}$, $P_{46}$, $P_{25}$ form a $\beta$-cycle in $H$.
Both $\widehat{G}$ and $H$ are depicted in \Cref{fig:tree-not-beta}.

\begin{figure}
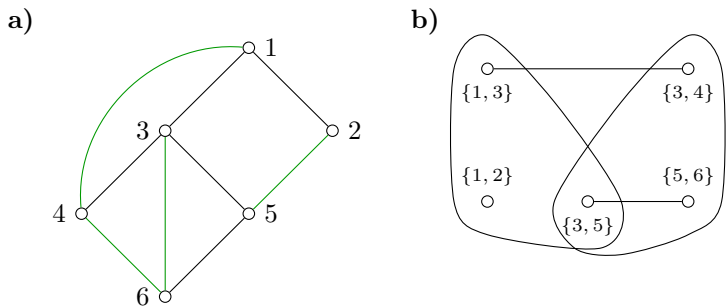

    \centering
    \textbf{a)}
    \imagetop{\input{figures/tree/tree-partition-beta-cycle-a.tex}}
    \hspace{0.5em}
    \textbf{b)}
    \imagetop{\input{figures/tree/tree-partition-beta-cycle-b.tex}}
    \caption{Depicted above are \textbf{a)} an augmented graph $\widehat{G} = (V, E \cup F)$ with augmented edges $F$ depicted in green. 
    \textbf{b)} The hypergraph $H$ representing the corresponding tree partition problem.}
    \label{fig:tree-not-beta}
\end{figure}


\section{Multicuts lifted from cycles}\label{sec:lmc-for-cycles}

In this section, we study the facial structure of the lifted multicut polytope $\lmc(C)$ for a cycle $C=(V,E)$. 
To begin with, we note that the lifted multicut problem for a cycle $C$, in the following called \emph{cycle partition problem}, can be solved in strongly polynomial time, by solving a linear number of path partition problems:
For any given edge in $C$, remove that edge from $C$ and take it to be an augmented edge in the resulting path partition problem.
The feasible solutions of this path partition problem are also feasible for the cycle partition problem.
Conversely, each feasible solution of the cycle partition problem is a feasible solution of at least one of the path partition problems (the all zeros solution is feasible for all path partition problems and all other solutions are feasible for those path partition problems that are obtained by removing an edge that is cut by the solution).
As the coefficients of the objective function remain the same, any feasible solution provides the same objective value for both the path partition problem and the cycle partition problem.
Therefore, the optimal solution of the cycle partition problem can be obtained by solving all path partition problems that are obtained by removing one edge respectively and picking a solution with minimal objective value.
The path partition problem can be solved in strongly polynomial time, as discussed in \Cref{sec:lmp-paths}.

Despite this simple reduction to the path partition problem and despite the simple description of the lifted multicut polytope for paths (cf. \Cref{cor:complete-description}), the lifted multicut polytope for cycles does not admit a simple description. 
For a path $P$ with $8$ nodes, $\lmc(P)$ has $34$ facets. For a cycle $C$ with $8$ node, $\lmc(C)$ has $37815$ facets.

We proceed as follows:
First and for clarity, we introduce some notation for describing a cycle $C$ of size $n$ and its properties.
Then, we state which of the cycle, path, cut and box inequalities are facet-defining for $\lmc(C)$. 
Thereafter, we show that most of the known facet-defining inequalities of the multicut polytope $\mc(K_n)$ of the complete graph with $n$ nodes, while valid, are not facet-defining for $\lmc(C)$. 
As our main contribution in this section, we establish several large classes of facet-defining inequalities for $\lmc(C)$. 
Of particular significance are the \emph{half-chorded odd cycle inequalities}, as they are also valid and facet-defining for $\mc(K_n)$.

\paragraph*{Notation.}
Let $n$ be the size of the cycle $C=(V,E)$.
For convenience, we identify the $n$ nodes with $\Z_n$, the ring of integers modulo $n$, such that $E = \{\{v, v+1\} \mid v \in \Z_n\}$. 
Then, we have $C = (\Z_n,E)$. 
Consistent with the previous sections, we let $F := \binom{\Z_n}{2} \setminus E$ and let $K_n = (\Z_n, E \cup F) = (\Z_n, \binom{\Z_n}{2})$ be the complete graph on the nodes $\Z_n$. 
For brevity, let
\[X_n := \lmc(C) \cap \{0,1\}^{\binom{\Z_n}{2}}\] 
denote the set of all characteristic vectors of multicuts of $K_n$ lifted from $C$.
For any $v,w \in \Z_n$, we define the interval from $v$ to $w$ as
$[v, w] := \{u \in \Z_n \mid u - v \leq w - v\} = \{v, v+1, \dots, w-1,w\}.$
Additionally, we define the (half) open intervals $]v,w[ \; := [v,w] \setminus \{v,w\}$ and $]v,w] := [v,w] \setminus \{v\}$ and $[v,w[ \; := [v,w] \setminus \{w\}$.
For any integer $1 \leq k \leq n$, we call a cyclic sequence $v \colon \Z_k \to \Z_n$ \emph{true to the cycle $C$} if one of the following two conditions holds:
\begin{enumerate*}[(a)]
\item \label{item:definition-true-a}
$v_i \in [v_0,v_{i+1}[$ for all $i=0,\dots,k-2$,
\item \label{item:definition-true-b}
$v_i \in \; ]v_{i+1},v_0]$ for all $i=0,\dots,k-2$. 
\end{enumerate*}
If \ref{item:definition-true-a} holds, the node $v_i$ comes before the node $v_{i+1}$ in $v_0,v_0+1,v_0+2,\dots$ for all $i=0,\dots,k-2$. 
If \ref{item:definition-true-b} holds, the node $v_i$ comes before the node $v_{i+1}$ in $v_0,v_0-1,v_0-2,\dots$ . 
In order to make such a sequence explicit, we write $\langle v_0,\dots,v_{k-1} \rangle$.
For example, let $n=k=4$. 
The cyclic sequences $\langle 0,1,2,3 \rangle$, $\langle 2,3,0,1 \rangle$ and $\langle 1,0,3,2 \rangle$ are true to $C$, while $\langle 0,2,1,3 \rangle$ and $\langle 3,2,0,1 \rangle$ are not. 
Notice: If $\langle v_0,\dots,v_{k-1} \rangle$ satisfies \ref{item:definition-true-a}, the reverse cyclic sequence $\langle v_{k-1},\dots,v_0 \rangle$ satisfies \ref{item:definition-true-b}, and vice versa. 
In the following, we assume without loss of generality that \ref{item:definition-true-a} is satisfied when we say a cyclic sequence is true to $C$.

With this notation, the cycle, path and cut inequalities \eqref{eq:lifted-multicut-cycle}, \eqref{eq:lifted-multicut-path} and \eqref{eq:lifted-multicut-cut} can be written compactly for the lifted multicut polytope for cycles. 

\begin{proposition}
    The lifted multicut polytope $\lmc(C)$ for a cycle $C=(\Z_n,E)$ is the convex hull of all vectors $x \in \{0,1\}^{\binom{\Z_n}{2}}$ that satisfy the cycle, path and cut inequalities for cycles:
    \begin{align}
        x_e &\leq \sum_{e' \in E \setminus \{e\}} x_{e'} 
            && \forall e \in E \label{eq:cycle} \\
        x_{vw} &\leq \sum_{u \in [v,w[} x_{u, u+1} \enspace, \quad
            x_{vw} \leq \sum_{u \in [w,v[} x_{u, u+1} 
            && \forall vw \in F \label{eq:path} \\
        1-x_{vw} &\leq \left(1-x_{s, s+1}\right) + \left(1-x_{t, t+1}\right) 
            && \forall vw \in F \quad \forall s \in [v, w[ \quad \forall t \in [w, v[ \label{eq:cut}
    \end{align}
\end{proposition}

\begin{proof}
    The claim follows from \Cref{prop:lifted-multicut-polytope-inequalities}:
    Since $C$ is the only cycle in $C$, the cycle inequalities \eqref{eq:lifted-multicut-cycle} can be written as \eqref{eq:cycle}.
    For any $vw \in F$, there are precisely two $vw$-paths, namely the paths along the nodes $[v,w]$ and $[w,v]$, respectively.
    Thus, the path inequalities \eqref{eq:lifted-multicut-path} can be written as \eqref{eq:path}.
    For any $vw \in F$, any $vw$-cut in $C$ consists of precisely two edges $e=\{s,s+1\}$ and $e'=\{t,t+1\}$ for some $s \in [v, w[$ and $t \in [w, v[$. 
    Thus, the cut inequalities \eqref{eq:lifted-multicut-cut} can be written as \eqref{eq:cut}.
\end{proof}

\paragraph*{Notation.}
Every (connected) component of $C = (\Z_n, E)$ that is not the entire graph is of the form $[v,w]$ for some $v,w \in \Z_n$. 
Every decomposition of $C$ into $k \geq 2$ components is a partition of $\Z_n$ of the form $\{]v_0,v_1],\allowbreak\;]v_1,v_2],\allowbreak\dots,\allowbreak\;]v_{k-2},v_{k-1}],\allowbreak\;]v_{k-1},v_0]\}$ for some $v \colon Z_k \to \Z_n$ true to $C$. 
For brevity, we introduce the following notation:
For any set of nodes $W \subseteq \Z_n$ with $k:=|W| \geq 2$, 
let $w \colon \mathbb{Z}_k \to W$ bijective and true to $C$, 
let $\Pi_W := \{]w_i,w_{i+1}] \mid i \in \Z_k \}$, 
and let
\begin{align}\label{eq:defn-chi}
    \chi(W) := \1_{\phi_{K_n}(\Pi_W)} \enspace ,
\end{align}  
the characteristic vector corresponding to the multicut of $K_n$ lifted from $C$ where the edges that are cut in $C$ are precisely the $w_iw_{i+1}$ for which $i \in \Z_k$.
For an example, see \Cref{fig:chi-example}. 
With the convention $\chi(\emptyset) := 0$, the zero vector, this notation allows us to describe precisely all characteristic vectors of multicuts of $K_n$ lifted from $C$, as $X_n = \{\chi(W) \mid W \subseteq \Z_n, |W| \neq 1\}$. 

\begin{figure}
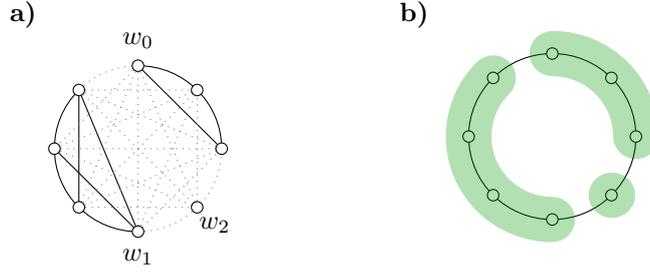

    \center
    \textbf{a)} \imagetop{\input{figures/cycle/chi-example.tex}}
    \hspace{5em}
    \textbf{b)} \imagetop{\input{figures/cycle/chi-example-decomposition.tex}}
    \caption{Depicted above are \textbf{a)} the lifted multicut corresponding to the vector $\chi(\{w_0,w_1,w_2\})$, with cut edges drawn as dotted lines, and in \textbf{b)} the corresponding decomposition of $C$.}
    \label{fig:chi-example}
\end{figure}

\subsection{Canonical facets}

In the following, we apply our results on the facet-defining\-ness of canonical inequalities for general graphs from \Cref{sec:facets-from-canonical-ineq} to the special case of cycles.

\begin{corollary}[of \Cref{thm:upper-box-facets}]\label{cor:upper-box-facet}
	For any integer $n \geq 3$,
	the cycle $C = (\Z_n, E)$ and 
	any $e \in \binom{\Z_n}{2}$, the inequality $x_e \leq 1$ defines a facet of $\lmc(C)$.
\end{corollary}

\begin{corollary}[of \Cref{thm:lower-box-facets}]\label{cor:lower-box-not-facet}
	For any integer $n \geq 3$,
	the cycle $C = (\Z_n, E)$ and 
    any $e \in \binom{\Z_n}{2}$, the inequality $0 \leq x_e$ does not define a facet of $\lmc(C)$.
\end{corollary}

\begin{corollary}[of \Cref{thm:cycle-facets}]
	For any integer $n \geq 4$ and the cycle $C = (\Z_n, E)$, none of the cycle inequalities \eqref{eq:cycle} defines a facet of $\lmc(C)$.
\end{corollary}

\begin{corollary}[of \Cref{thm:cycle-facets}] \label{cor:path-length-2-facet}
	For any integer $n \geq 3$ and the cycle $C = (\Z_n, E)$, a path inequality \eqref{eq:path} defines a facet of $\lmc(C)$ if and only if the path has length $2$.
\end{corollary}

Note that, for $n=3$, cycle inequalities coincide with path inequalities for paths of length $2$.

\begin{proposition}\label{prop:cut-not-facet}
	For any integer $n \geq 3$,
	the cycle $C = (\Z_n, E)$ and 
	any $f=vw \in F$, 
	let $\delta(U) = \{e, e'\}$ be a $vw$-cut. 
    The cut inequality $1-x_f \leq (1-x_e) + (1-x_{e'})$
    defines a facet of $\lmc(C)$ only if $e$ and $e'$ share a node. 
\end{proposition}

\proofref{prop:cut-not-facet}

\Cref{prop:cut-not-facet} gives only a necessary conditions for a cut inequality to be facet-defining. 
The only $vw$-cuts that satisfy this necessary condition are $\{\{v-1,v\},\{v,v+1\}\}$ and $\{\{w-1,w\},\{w,w+1\}\}$. 
In \Cref{sec:star-glider}, we establish classes of inequalities that generalize those inequalities and show that these classes of inequalities are indeed facet-defining, i.e.~the condition in \Cref{prop:cut-not-facet} is not only necessary but also sufficient (cf. \Cref{cor:cut_facet_equivalence}).

\subsection{Facets inherited from the multicut polytope for complete graphs} \label{sec:facets-of-mc}

Since $\lmc(C) \subseteq \mc(K_n)$ (cf. \Cref{prop:inclusion-property}), any inequality that is valid for $\mc(K_n)$ is valid also for $\lmc(C)$. 
This raises the question which of the known classes of facet-defining inequalities of $\mc(K_n)$ are facet-defining also for $\lmc(C)$. 
Below, we show that, with a few exceptions that coincide with canonical inequalities, many of the known classes of facet-defining inequalities of $\mc(K_n)$ are not facet-defining for $\lmc(C)$.
For clarity, we state those inequalities that were introduced originally for the isomorphic clique partitioning polytope equivalently here for $\mc(K_n)$, by substituting $1 - x$ for $x$.

\begin{definition}[Definition 1.1 of \citet{Deza1992}]
    Let $p,q \geq 1$, $r \geq 0$ with $p - q \geq 2r + 1$ and let $S,T \subset \Z_n$ with $S \cap T = \emptyset$, $|S| = q$ and $|T|=p$. 
    Let $v\colon \Z_p \to T$ bijective. 
    The \emph{anti-web} with respect to $r$ and $v$ is defined as $AW := \{\{v_i,v_{i+\ell}\} \mid i\in \Z_p, \ell \in \{1,\dots,r\} \}$, the \emph{web} with respect to $r$ and $v$ is defined as $W := \binom{T}{2} \setminus AW$. 
    The \emph{clique-web inequality} with respect to $r$, $S$, $T$ and $W$ is defined as 
    \begin{align}\label{eq:clique-web-inequality}
        \sum_{e \in W} x_e 
            + \sum_{v,w \in S, v \neq w} x_{vw} 
            - \sum_{v \in S, w \in T} x_{vw} 
        \leq \frac{(p-q)(p-q-2r-1)}{2} \enspace.
    \end{align}
\end{definition}

\begin{lemma}[Proposition 1.7 of \citet{Deza1992}]\label{lem:clique-web-valid}
    The clique-web inequality \eqref{eq:clique-web-inequality} is valid for the multicut polytope of the complete graph $\mc(K_n)$.
\end{lemma}

\begin{theorem}[Theorem 1.11 of \citet{Deza1992} and Proposition 3 of \citet{Sorensen2002}]\label{thm:clique-web-facet}
    For any integers $p, q \geq 1$ and $r \geq 0$ with $p + q \leq n$, $p - q \geq 2r + 1$, $q \geq \frac{p - 1}{2} - r$ if $r \geq 1$, and $q \geq 2$ if $p - q = 2r + 1$, the clique-web inequality \eqref{eq:clique-web-inequality} defines a facet of $\mc(K_n)$.
\end{theorem}

\begin{proposition}\label{prop:clique-web-not-facet}
	For any integer $n \geq 3$ and the cycle $C=(\Z_n,E)$, the clique-web inequality \eqref{eq:clique-web-inequality} is facet-defining for $\lmc(C)$ if and only if it coincides with a path inequality \eqref{eq:path} of a path of length $2$, i.e.~$p=2$, $q=1$, $r=0$, $S=\{u\}$ and $T=\{u-1,u+1\}$ for some $u \in \Z_n$.
\end{proposition}

\proofref{prop:clique-web-not-facet}

The clique-web inequalities \eqref{eq:clique-web-inequality} generalize several other classes of facet-defining inequalities of $\mc(K_n)$:
For $q=2$, $p=1$ and $r=0$, the clique-web inequality \eqref{eq:clique-web-inequality} is a triangle inequality from \citet{Groetschel1990}. 
For $r \geq 0$, $p=2r+3$ and $q=1$ (respectively $q=2$), the clique-web inequality \eqref{eq:clique-web-inequality} coincides with a \emph{wheel inequality} (respectively \emph{bicycle wheel inequality}) from \citet{Chopra1993} (cf. Section 1.3 of \citet{Deza1992}). 
For $r=0$, the clique-web inequality \eqref{eq:clique-web-inequality} is a \emph{$[S,T]$-inequality} (also called $2$-partition inequality) from \citet{Groetschel1990}. 
By \Cref{prop:clique-web-not-facet}, all those inequalities are not facet-defining for the polytope $\lmc(C)$, unless they coincide with a path inequality \eqref{eq:path} for a path of length $2$.

\begin{definition}[Section 5 of \citet{Groetschel1990}]
    Let $5 \leq k \leq n$ and let $v \colon \Z_k \to \Z_n$ be injective.
    The \emph{2-chorded cycle inequality} with respect to $v_0,\dots,v_{k-1}$ is defined as
    \begin{align}\label{eq:2-chorded-cycle}
        \sum_{i \in \Z_k} \left(x_{v_iv_{i+2}} - x_{v_iv_{i+1}}\right)
            \leq \left\lfloor \frac{k}{2} \right\rfloor \enspace.
    \end{align}
\end{definition}

\begin{theorem}[Theorem 5.1 of \citet{Groetschel1990}]\label{thm:2-chorded-valid-and-facet}
    The 2-chorded cycle inequality \eqref{eq:2-chorded-cycle} is valid for $\mc(K_n)$. It defines a facet of $\mc(K_n)$ if and only if $k$ is odd.
\end{theorem}

\proofref{thm:2-chorded-valid-and-facet}

\begin{proposition}\label{prop:2-chorded-not-facet}
	For any integers $5 \leq k \leq n$,
	the cycle $C=(\Z_n,E)$ 
	and any $v \colon \Z_k \to \Z_n$ injective,
    the 2-chorded cycle inequality \eqref{eq:2-chorded-cycle} can only be facet-defining for $\lmc(C)$ if the following necessary conditions are satisfied:
    \begin{enumerate}[(i)]
        \item \label{item:2-chorded-not-true-not-facet}
        $v$ is true to $C$
        \item \label{item:2-chorded-geq-k-not-facet}
        $k = 5$.
    \end{enumerate}
\end{proposition}

\proofref{prop:2-chorded-not-facet}

For $k=5$, the 2-chorded cycle inequality coincides with a half-chorded odd cycle inequality from \Cref{sec:half-chorded}. 
By \Cref{thm:half-chorded-facet}, it is facet-defining for $\lmc(C)$ if and only if $v \colon \Z_5 \to \Z_n$ is true to $C$.
Therefore the conditions in \Cref{prop:2-chorded-not-facet} are not only necessary but also sufficient.

In the literature, several other classes of facet-defining inequalities of $\mc(K_n)$ are obtained by lifting and/or patching the facet-defining inequalities discussed above \cite{Groetschel1990composition,Chopra1995,Bandelt1999lifting,Oosten2001clique}. 
For example, the \emph{general 2-partition inequalities} from \citet{Groetschel1990composition} are obtained by patching 2-partition inequalities, the special case of the clique-web inequalities \eqref{eq:clique-web-inequality} with $r=0$. 
The proofs showing that the derived classes of inequalities are facet-defining rely on the fact that the initial classes of inequalities are facet-defining. 
We conjecture that the derived inequalities are not facet-defining for $\lmc(C)$.
This conjecture is supported by computational experiments, cf. \Cref{rem:facets-of-both-mc-and-lmc}.

\subsection{Intersection inequalities}\label{sec:intersection-ineq-cycle}

In this section, we apply the results about the intersection inequalities for arbitrary graphs from \Cref{thm:intersection-arbitrary-g} to the special case of the cycle.

\begin{corollary}\label{cor:intersection-cycle}
	For any integer $n \geq 3$, the cycle $C=(\Z_n,E)$ and any $vw \in F$, the \emph{intersection inequality} with respect to $vw$, written below, is valid and facet-defining for $\lmc(C)$.
    \begin{align}\label{eq:cycle-intersection-inequality}
        x_{v,w+1} + x_{v+1,w} \leq x_{vw} + x_{v+1,w+1}
    \end{align}
\end{corollary}

\begin{proof}
    The claim follows from \Cref{thm:intersection-arbitrary-g}, since $\{v+1,w+1\}$ is a $vw$-separating node set and $\{v,w\}$ is a $v+1,w+1$-separating node set.
\end{proof}

\subsection{Half-chorded odd cycle inequalities}\label{sec:half-chorded}

In this section, we take three steps. 
Firstly, we introduce a class of inequalities that are valid and facet-defining for $\mc(K_n)$, the multicut polytope of the complete graph. 
Secondly, we establish conditions under which these inequalities are facet-defining for $\lmc(C)$. 
Thirdly, we obtain facet-defining inequalities for the lifted multicut polytope $\lmc(G)$ for arbitrary graphs $G$ that arise from cycles in $G$.

\begin{definition}
    Let $5 \leq k \leq n$ with $k$ odd, let $d=\frac{k-1}{2}$, and let $v \colon \Z_k \to \Z_n$ be injective. 
    The \emph{half-chorded odd cycle inequality} with respect to $v$ is defined as
    \begin{align}\label{eq:half-chorded-odd-cycle-ineq}
        \sum_{i \in \Z_k} \left( x_{v_iv_{i+d}} - x_{v_iv_{i+1}} \right) 
        \leq k - 3 \enspace.
    \end{align}
\end{definition}

Herein, the edges $v_iv_{i+1}$ for $i \in \Z_k$ form a cycle in $K_n$, while the edges $v_iv_{i+d}$ are \emph{half-chords} of that cycle, i.e.~chords that halve the cycle into two parts with $d$ and $d+1$ edges, respectively. 
For $d=2$, the half-chords are 2-chords and the half-chorded odd cycle inequality coincides with the 2-chorded cycle inequality \eqref{eq:2-chorded-cycle} from \citet{Groetschel1990}. 
Note that the edges $v_iv_{i+d}$ for $i \in \Z_k$ also form a cycle in $K_n$, and that the edges $v_iv_{i+1}$ are the 2-chords of that cycle. 
Therefore, the left hand side of \eqref{eq:half-chorded-odd-cycle-ineq} is the negative of the left hand side of \eqref{eq:2-chorded-cycle}, and together the inequalities can be written as
\begin{align}\label{eq:2-chorded-and-half-chorded}
    - d \leq 
    \sum_{i \in \Z_k} \left( x_{v_iv_{i+d}} - x_{v_iv_{i+1}} \right) 
    \leq k-3 \enspace .
\end{align}
For an illustration of the support graph of a half-chorded odd cycle inequality, see \Cref{fig:half-chorded-odd-cycle}. 

\begin{figure}
    \centering
    \input{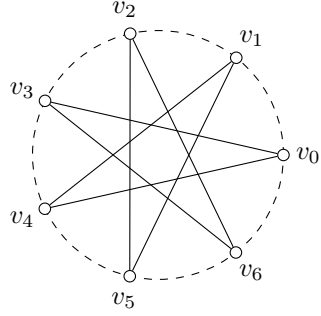}
    \caption{Depicted above is the support graph of a half-chorded odd cycle inequality \eqref{eq:half-chorded-odd-cycle-ineq} for $d=3$. 
    The edges $\{v_i,v_{i+d}\}$ corresponding to coefficients $+1$ are depicted as solid lines, the edges $\{v_i,v_{i+1}\}$ corresponding to coefficients $-1$ are depicted as dashed lines.}
    \label{fig:half-chorded-odd-cycle}
\end{figure}

\begin{proposition}\label{prop:num-half-chorded-odd-cycle}
    For any $n \in \mathbb{N}$, there are 
    \[    
        \sum_{k=5 \text{ odd}}^n \frac{n!}{(n-k)! \; 2k}
    \]
    distinct half-chorded odd cycle inequalities.
\end{proposition}

\proofref{prop:num-half-chorded-odd-cycle}

We show that the half-chorded odd cycle inequalities are Chv\'atal inequalities for $\mc(K_n)$ and, thus, are valid for $\mc(K_n)$. We refer to \citet{conforti2014integer} for a thorough introduction to Chv\'atal inequalities. 

\begin{lemma}\label{lem:half-chorded-valid}
    The half-chorded odd cycle inequalities \eqref{eq:half-chorded-odd-cycle-ineq} are Chv\'atal inequalities for the multicut polytope $\mc(K_n)$ of the complete graph. 
    In particular, they are valid for $\mc(K_n)$.
\end{lemma}

\proofref{lem:half-chorded-valid}

By \Cref{prop:inclusion-property}, the lifted multicut polytope $\lmc(C)$ is a subset of the multicut polytope $\mc(K_n)$. 
Thus, we obtain the following corollary.

\begin{corollary}[of \Cref{lem:half-chorded-valid}]\label{cor:half-chorded-valid-lmc}
    For any integer $n \geq 3$ and the cycle $C=(\Z_n,E)$, the half-chorded odd cycle inequalities \eqref{eq:half-chorded-odd-cycle-ineq} are valid for $\lmc(C)$.
\end{corollary}

In order to show that the half-chorded odd cycle inequalities are facet-defining, we establish which characteristic vectors $\chi(W)$, defined in \eqref{eq:defn-chi}, satisfy \eqref{eq:half-chorded-odd-cycle-ineq} with equality:

\begin{lemma}\label{lem:roots-of-half-chorded}
    For any integers $5 \leq k \leq n$ with $k$ odd, $d=\frac{k-1}{2}$, and for the cycle $C=(\Z_n,E)$, let $v \colon \Z_k \to \Z_n$ be injective and true to $C$. 
    Furthermore, let $W \subseteq \Z_n$ with $|W|\neq 1$, and let $x:=\chi(W)$. 
    Define $I := \bigl\{i \in \Z_k \mid W \cap [v_i,v_{i+1}[ \; \neq \emptyset \bigr\}$ the set of indices $i$ such that the path along the nodes $[v_i,v_{i+1}]$ is cut with respect to $x$. 
    Then, $x$ satisfies the half-chorded odd cycle inequality \eqref{eq:half-chorded-odd-cycle-ineq} with respect to $v$ with equality if and only if precisely one of the following conditions is satisfied:
    \begin{enumerate}[(a)]
        \item \label{cond:roots-half-chorded-a}
        $\abs{I} = 2$ and for all $j,\ell \in I$ with $j \neq \ell$, we have $j - \ell \in \{d, d+1\}$ 
        \item \label{cond:roots-half-chorded-b}
        $\abs{I} = 3$ and for all $j,\ell \in I$ with $j \neq \ell$, we have $j-\ell \in \{1,\dots,d\}$ or $\ell - j \in \{1,\dots,d\}$
    \end{enumerate}
\end{lemma}

\proofref{lem:roots-of-half-chorded}

For an illustration of a lifted multicut corresponding to $\chi(W)$ that satisfies a half-chorded odd cycle inequality \eqref{eq:half-chorded-odd-cycle-ineq} with equality, see \Cref{fig:root-of-half-chorded-odd-cycle}.

\begin{figure}
    \centering
    \input{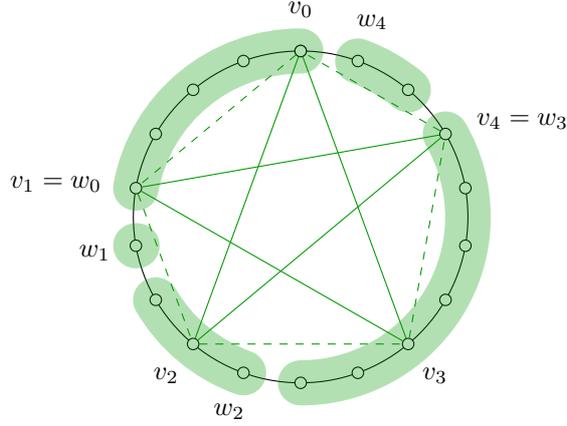}
    \caption{Depicted above is a half-chorded odd cycle inequality \eqref{eq:half-chorded-odd-cycle-ineq} for $n=18$ and $k=5$, along with the decomposition of the cycle $C$ (black edges) that corresponds to $x:=\chi(W)$ with $W=\{w_0,\dots,w_4\}$. 
    The support graph of the inequality is depicted by the green edges where the edges with coefficients $+1$ and $-1$ are depicted by continuous and dashed lines, respectively. 
    The components of the decomposition of $C$ are depicted as green areas. 
    The edges $\{v_i,v_{i+2}\}$ for $i=0,\dots,4$ are all cut with respect to $x$. 
    The edges $\{v_0,v_1\}$ and $\{v_3,v_4\}$ are not cut with respect to $x$, while the edges $\{v_1,v_2\}$, $\{v_2,v_3\}$ and $\{v_0,v_4\}$ are. 
    Therefore, $x$ satisfies the inequality with equality. 
    Indeed, Condition~\ref{cond:roots-half-chorded-b} of \Cref{lem:roots-of-half-chorded} is satisfied, since we have $I=\{1,2,4\}$. 
    For $W=\{w_0,w_1, w_3\}$, for instance, the vector $\chi(W)$ also satisfies the inequality with equality because $I=\{1, 4\}$ and Condition~\ref{cond:roots-half-chorded-a} of \Cref{lem:roots-of-half-chorded} is satisfied. 
    For $W=\{w_0,w_1\}$, for instance, the vector $\chi(W)$ does not satisfy the inequality with equality, since $I=\{1\}$.}
    \label{fig:root-of-half-chorded-odd-cycle} 
\end{figure}

\begin{lemma}\label{lem:half-chorded-facet}
	For any integers $5 \leq k \leq n$ with $k$ odd, and for the cycle $C=(\Z_n,E)$,
	let $v \colon \Z_k \to \Z_n$ be injective.
    If $v$ is true to $C$, then the half-chorded odd cycle inequality \eqref{eq:half-chorded-odd-cycle-ineq} with respect to $v$ is facet-defining for $\lmc(C)$.
\end{lemma}

\proofref{lem:half-chorded-facet}

\begin{lemma}\label{lem:half-chorded-not-true-not-facet}
	For any integers $5 \leq k \leq n$ with $k$ odd, and for the cycle $C=(\Z_n,E)$,
	let $v \colon \Z_k \to \Z_n$ be injective.
    If $v$ is not true to $C$, then the half-chorded odd cycle inequality \eqref{eq:half-chorded-odd-cycle-ineq} with respect to $v$ is not facet-defining for $\lmc(C)$.
\end{lemma}

\proofref{lem:half-chorded-not-true-not-facet}

\begin{theorem}\label{thm:half-chorded-facet}
	For any integers $5 \leq k \leq n$ with $k$ odd, and for the cycle $C=(\Z_n,E)$,
	let $v \colon \Z_k \to \Z_n$ be injective.
    The half-chorded odd cycle inequality \eqref{eq:half-chorded-odd-cycle-ineq} with respect to $v$ is facet-defining for $\lmc(C)$ if and only if $v$ is true to $C$.
\end{theorem}

\begin{proof}
    By \Cref{lem:half-chorded-facet} and \Cref{lem:half-chorded-not-true-not-facet}.
\end{proof}

\begin{proposition}\label{prop:num-half-chorded-odd-cycle-facet}
	For any integer $n \geq 3$ and the cycle $C = (\Z_n, E)$, 
	there are precisely $2^{n-1} - n - \frac{n(n-1)(n-2)}{6}$ distinct half-chorded odd cycle inequalities that are facet-defining for $\lmc(C)$.
\end{proposition}

\proofref{prop:num-half-chorded-odd-cycle-facet}

Exploiting the inclusion properties (\Cref{prop:inclusion-property}) of different lifted multicut polytopes, we now derive conditions under which the half-chorded odd cycle inequalities are facet-defining also for the lifted multicut polytope $\lmc(G)$ for an arbitrary graph $G$:

\begin{corollary}\label{cor:half-chorded-facet-hamiltonian}
    Let $G$ be a graph with $n$ nodes, 
    let $C = (V_C, E_C)$ be a hamiltonian cycle in $G$,
    let $5 \leq k \leq n$ with $k$ odd, and let $v \colon \Z_k \to V_C$ such that $v$ is injective and true to $C$. 
    Then, the half-chorded odd cycle inequality \eqref{eq:half-chorded-odd-cycle-ineq} with respect to $v$ is valid and facet-defining for $\lmc(G)$.
\end{corollary}

\begin{proof}
    By \Cref{prop:inclusion-property}, $\lmc(C) \subseteq \lmc(G) \subseteq \mc(K_n)$. 
    By \Cref{lem:half-chorded-valid}, the inequality is valid for $\mc(K_n)$.
    By \Cref{thm:half-chorded-facet}, it is facet-defining for $\lmc(C)$. 
    The claim follows by \Cref{lem:valid-and-facet-defining-for-subset}.
\end{proof}

The multicut polytope $\mc(K_n)$ of the complete graph is the special case of $\lmc(G,\widehat{G})$ with $G=\widehat{G}=K_n$. 
As an application of \Cref{cor:half-chorded-facet-hamiltonian}, we obtain that the class of half-chorded odd cycle inequalities is facet-defining for $\mc(K_n)$. 
This establishes a new class of facet-defining inequalities for the isomorphic and intensively studied clique partitioning polytope.

\begin{proposition}\label{cor:half-chorded-facet-mc}
    Any half-chorded odd cycle inequality \eqref{eq:half-chorded-odd-cycle-ineq} is facet-defining for the multicut polytope $\mc(K_n)$ of the complete graph.
\end{proposition}

\begin{proof}
    Let $5 \leq k \leq n$ with $k$ odd.
    For any $v\colon \Z_k \to \Z_n$, there exists a hamiltonian cycle $C$ in $K_n$ such that $v$ is true to $C$. 
    The claim follows from \Cref{cor:half-chorded-facet-hamiltonian}.
\end{proof}

From \Cref{thm:2-chorded-valid-and-facet} and \Cref{cor:half-chorded-facet-mc} together follows that both the lower and upper inequality of \eqref{eq:2-chorded-and-half-chorded} are facet-defining for $\mc(K_n)$.

One can impose various constraints on multicuts and study the polytopes of the characteristic vectors of the multicuts that satisfy these constraints.
One example is the polytope corresponding to multicuts of the complete graph into at most $r$ components, for some $r \leq n$ (cf. \citet{Chopra1995,Deza1992}). 
This polytope is denoted by $\mc^r_\leq(K_n)$. 

\begin{proposition}
    Any half-chorded odd cycle inequality \eqref{eq:half-chorded-odd-cycle-ineq} is facet-defining for $\mc^r_\leq(K_n)$ for all $r \geq 4$.
\end{proposition}

\begin{proof}
    Clearly, $\mc^r_\leq(K_n) \subseteq \mc(K_n)$.
    The claim follows from the proof of \Cref{lem:half-chorded-facet} in which only multicuts into at most $4$ components are considered.
\end{proof}

Next, we generalize the above results to arbitrary connected graphs. 
In particular, we show that, for a connected graph $G$, a half-chorded odd cycle inequality is facet-defining for $\lmc(G)$ if $v$ is true to any (not necessarily hamiltonian) cycle in $G$.

\begin{theorem}\label{thm:half-chorded-facet-general-g}
    Let $G=(V,E)$ be a connected graph and let $C = (V_C,E_C)$ be a cycle in $G$. 
    Let $5 \leq k \leq |V_C|$ with $k$ odd and $v \colon \Z_k \to V_C$ such that $v$ is injective and true to $C$. 
    Then, the half-chorded odd cycle inequality \eqref{eq:half-chorded-odd-cycle-ineq} with respect to $v$ is facet-defining for $\lmc(G)$.
\end{theorem}

\proofref{thm:half-chorded-facet-general-g}

By \Cref{thm:half-chorded-facet-general-g} and \Cref{lem:facet-for-lifting-to-sub-graph}, the half-chorded odd cycle inequalities are not only facet-defining for the multicut polytope of the complete graph but also for the multicut polytope of an arbitrary graph:

\begin{corollary}
    Let $G=(V,E)$ be a connected graph. 
    Let $5 \leq k \leq |V|$ with $k$ odd, $d=\frac{k-1}{2}$, and $v \colon \Z_k \to V$ such that $v$ is injective and such that for all $i \in \Z_k$, we have $v_i v_{i+1} \in E$ and $v_i v_{i+d} \in E$.
	Then, the half-chorded odd cycle inequality with respect to $v$ is facet-defining for $\mc(G)$.
\end{corollary}

Independently, \citet{muller1996partial} and \citet{caprara1996} show that the separation problem for the 2-chorded cycle inequalities \eqref{eq:2-chorded-cycle} can be solved in polynomial time. 
The result by \citet{caprara1996} relies on the fact that the 2-chorded odd cycle inequalities are $\{0,\tfrac{1}{2}\}$- Chv\'atal inequalities.
As suggested by the proof of \Cref{lem:half-chorded-valid}, this is not the case for the half-chorded odd cycle inequalities.
The separation algorithm by \citet{muller1996partial}, on the other hand, searches for a shortest weighted walk in a directed auxiliary graph. 
This algorithm can be easily adapted to also separate the half-chorded odd cycle inequalities.

\begin{proposition}\label{prop:half-chorded-separation}
    The separation problem for the half-chorded odd cycle inequalities \eqref{eq:half-chorded-odd-cycle-ineq} can be solved in polynomial time.
\end{proposition}

\proofref{prop:half-chorded-separation}

We conclude this section with observations and open questions regarding the half-chorded odd cycle inequalities:

\begin{remark}
    By \Cref{thm:half-chorded-facet-general-g}, the condition that $v$ is true to a cycle in $G$ is sufficient for the half-chorded odd cycle inequality \eqref{eq:half-chorded-odd-cycle-ineq} with respect to $v$ to be facet-defining for $\lmc(G)$. 
    However, this condition is not necessary. 
    For an example, see \Cref{fig:half-chorded-not-true-still-facet}.
\end{remark}

\begin{figure}
    \centering
    \input{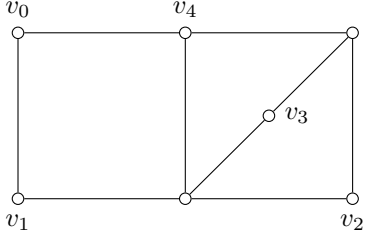}
    \caption{For the graph $G$ depicted above, there is no cycle $C=(V_C,E_C)$ in $G$ such that $v \colon \Z_5 \to V_C$ as depicted is true to $C$. 
    Still, the half-chorded odd cycle inequality \eqref{eq:half-chorded-odd-cycle-ineq} with respect to $v$ is facet-defining for the polytope $\lmc(G)$.}
    \label{fig:half-chorded-not-true-still-facet}
\end{figure}

\begin{remark}\label{rem:facets-of-both-mc-and-lmc}
    Using \citet{porta}, we have computed all facets of the lifted multicut polytope $\lmc(C)$ for cycles of size $n = 3,\dots,8$. 
    For those $n$, the box inequalities $x_e \leq 1$ (cf. \Cref{cor:upper-box-facet}), the path inequalities \eqref{eq:path} for paths of length $2$ (cf. \Cref{cor:path-length-2-facet}) and the half-chorded odd cycle inequalities \eqref{eq:half-chorded-odd-cycle-ineq} (cf. \Cref{thm:half-chorded-facet}) are the only inequalities that are facet-defining for both, the lifted multicut polytope $\lmc(C)$ for a cycle $C$ with $n$ nodes and the multicut polytope $\mc(K_n)$ of complete graph $K_n$. 
    This raises the question whether this holds for all $n$.
\end{remark}

\subsection{Star-glider inequalities}\label{sec:star-glider}

In the following, we introduce three large classes of facet-defining inequalities of $\lmc(C)$. 
Firstly, we define the classes of star inequalities and glider inequalities, both of which are of exponential size. 
Then, we generalize these classes to the class of star-glider inequalities and show that this class of inequalities is indeed valid and facet-defining for $\lmc(C)$.

\begin{definition}\label{defn:star-inequality}
	For any cycle $C=(\Z_n,E)$,
	any odd number $k$,
	any $v \in \Z_n$ and 
	any $w_1,\dots,w_k \in \Z_n \setminus \{v\}$ distinct such that $\langle v,w_1,\dots,w_k \rangle$ is true to $C$, 
    the \emph{star inequality} with respect to $v$ and $w_1,\dots,w_k$ is defined as 
	\begin{align}\label{eq:star-inequality}
		\sum_{i=1}^k (-1)^{i-1} x_{vw_i} \leq 1 \enspace.
	\end{align}
\end{definition}

The support graph of a star inequality is a star graph with center $v$, hence the name. 
For $k = 1$, a star inequality is merely a box inequality $x_e \leq 1$ for $e=\{v,w_1\}$. 
For $k=3$, $w_1=v+1$ and $w_3 = v-1$ the star inequality
\begin{align*}
	& x_{v,v+1} - x_{vw_2} + x_{v,v-1} \leq 1 \\
    \Leftrightarrow \quad &
    1 - x_{vw_2} \leq (1-x_{v,v+1}) + (1 - x_{v,v-1})
\end{align*}
is merely a cut inequality \eqref{eq:cut} with respect to the $vw_2$-cut $\{\{v-1,v\},\{v,v+1\}\}$. 
For an illustration of the support graph of a star inequality with $k=5$, see \Cref{fig:star-glider-inequalities}a.

\begin{proposition}\label{prop:num-star-inequalities}
	For any integer $n \geq 3$ and the cycle $C = (\Z_n, E)$, there are precisely $n2^{n-2} - n(n-1)/2$ distinct star inequalities.
\end{proposition}

\proofref{prop:num-star-inequalities}

\begin{definition}\label{defn:glider-inequality}
	For any integer $n \geq 3$, 
	the cycle $C = (\Z_n, E)$, 
    any $0 \leq k \leq n-3$, 
    any $w \in \Z_n$ and 
    any $v_0,\dots,v_{k+1} \in \Z_n \setminus \{w\}$ distinct such that $\langle w,v_0,\dots,v_{k+1} \rangle$ is true to $C$,
    the \emph{glider inequality} with respect to $w$ and $v_0,\dots,v_{k+1}$ is defined as 
    \begin{align}\label{eq:glider-inequality}
        \sum_{i=0}^k x_{v_iv_{i+1}} - \sum_{i=1}^k x_{v_iw} \leq 1 \enspace.
    \end{align}
    We call the edge set $\{v_iv_{i+1} \mid i \in \{0,\dots,k\}\}$ the \emph{sail} of the glider inequality and the edge set $\{v_iw \mid i \in \{1,\dots,k\}\}$ the \emph{strings} of the glider inequality.
\end{definition}

For $k=0$, a glider inequality is merely a box inequality $x_e \leq 1$ with $e=v_0v_1$. 
For $k=1$, a glider inequality is merely a star inequality (with $k=3$). 
For an illustration of the support graph of a glider inequality with $k=3$, see \Cref{fig:star-glider-inequalities}b. 

\begin{proposition}\label{prop:num-glider-inequalities}
    For the $n$-cycle, there are precisely $n2^{n-1} - n(n + (n-1)(n-3)/2)$ distinct glider inequalities. 
\end{proposition}

\proofref{prop:num-glider-inequalities}

The star and glider inequalities are generalized by the following definition:

\begin{definition}\label{defn:star-glider-inequality}
	Let $C=(\Z_n,E)$ an $n$-cycle, let $0 \leq k \leq n-3$ and let $v_0,\dots,v_{k+1} \in \Z_n$ distinct such that $\langle v_0,\dots,v_{k+1} \rangle$ is true to $C$. 
    For $i=1,\dots,k$, let $m_i \in \N$ odd, let $w_i^1,\dots,w_i^{m_i} \in \;]v_{k+1}, v_0[$ distinct such that $\langle w_i^1,\dots,w_i^{m_i} \rangle$ is true to the cycle $C$ and such that $w_{i+1}^{m_{i+1}} \in \;]{v_{k+1}}, w_i^1]$ for $i=1,\dots,k-1$. 
    The \emph{star-glider inequality} with respect to $v_0,\dots,v_{k+1}$ and $w_1^1,\dots,w_1^{m_1}, \dots, w_k^1,\dots,w_k^{m_k}$ is defined as
    \begin{align}\label{eq:star-glider-inequality}
        \sum_{i=0}^k x_{v_iv_{i+1}} + \sum_{i=1}^k \sum_{j=1}^{m_i} (-1)^j x_{v_iw_i^j} \leq 1 \enspace.
    \end{align}
    Similarly to \Cref{defn:glider-inequality}, we call the edge set $\{v_iv_{i+1} \mid i \in \{0,\dots,k\}\}$ the \emph{sail} of the star-glider inequality and the edge set $\{v_iw_i^j \mid i\in \{1,\dots,k\}, j\in \{1,\dots,m_i\}\}$ the \emph{strings} of the star-glider inequality.
\end{definition}

The conditions on the orderings of the $v_i$ and $w_i^j$ from above definition ensure that the sequence $\langle v_0,\dots,v_{k+1},w_k^1,\dots,w_k^{m_k},\dots,w_1^1,\dots,w_1^{m_1} \rangle$ is true to the cycle $C$ and that all these nodes are distinct, except potentially $w_i^{m_i} = w_{i+1}^1$, for $i=1,\dots,k-1$.

For $k=0$, a star-glider inequality is merely a box inequality $x_e \leq 1$ with $e = v_0v_1$. 
The star-glider inequalities with $k=1$ are precisely the star inequalities according to \Cref{defn:star-inequality}. 
The star-glider inequalities with $m_i=1$ and $w_i^1 = w$ for all $i=1,\dots,k$ for a fixed $w \in \Z_n$ are precisely the glider inequalities according to \Cref{defn:glider-inequality}.
For an illustration of the support graph of a star-glider inequality that is neither a star inequality nor a glider inequality, see \Cref{fig:star-glider-inequalities}c.

\begin{figure}
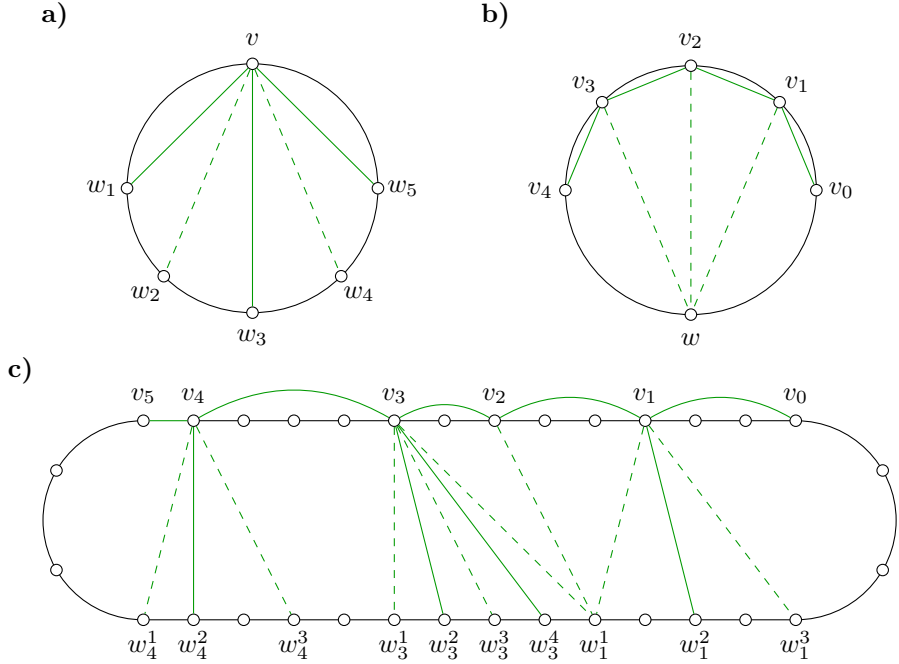

    \center
    \textbf{a)} \imagetop{\input{figures/cycle/star-inequality.tex}} 
    \hspace{1em}
    \textbf{b)} \imagetop{\input{figures/cycle/glider-inequality.tex}}
    \linebreak
    \textbf{c)} \imagetop{\input{figures/cycle/star-glider-inequalities.tex}}
    \caption{Depicted in green are the support graphs of a star inequality from \Cref{defn:star-inequality} in \textbf{(a)}, of a glider inequality from \Cref{defn:glider-inequality} in \textbf{(b)} and of a star-glider inequality from \Cref{defn:star-glider-inequality} in \textbf{(c)}. 
    The edges corresponding to coefficients $+1$ and $-1$ are depicted as continuous and dashed lines, respectively. 
    With a bit of imagination, the support graph of a glider inequality looks like a hang glider consisting of a sail and strings. 
    Note for the depicted star-glider inequality that $w_1^1=w_2^1=w_3^5$.}
    \label{fig:star-glider-inequalities}
\end{figure}

\begin{lemma}\label{lem:star-glider-valid}
    For any $n$-cycle $C=(\Z_n,E)$, any star-glider inequality \eqref{eq:star-glider-inequality} is valid for $\lmc(C)$.
\end{lemma}

\proofref{lem:star-glider-valid}

\begin{theorem}\label{thm:star-glider-facet}
    For the $n$-cycle $C=(\Z_n,E)$, any star-glider inequality \eqref{eq:star-glider-inequality} defines a facet of $\lmc(C)$.
\end{theorem}

\proofref{thm:star-glider-facet}

As the star-glider inequalities generalize the star inequalities and the glider inequalities we obtain the following two corollaries.

\begin{corollary}\label{cor:star-glider-facet}
	For any $n$-cycle $C=(\Z_n,E)$, any star inequality \eqref{eq:star-inequality} is valid and defines a facet of $\lmc(C)$.
\end{corollary}

\begin{corollary}\label{cor:glider_facet_defining}
	For any $n$-cycle $C=(\Z_n,E)$, any glider inequality \eqref{eq:glider-inequality} is valid and defines a facet of $\lmc(C)$.
\end{corollary}

As we have seen before, for $vw \in F$, the cut inequality \eqref{eq:cut} with respect to the $vw$-cut $\{\{v-1,v\},\{v,v+1\}\}$ is a special case of the star inequality \eqref{eq:star-inequality}. 
Thanks to \Cref{prop:cut-not-facet}, we obtain a complete characterization of the condition under which a cut inequality \eqref{eq:cut} is facet-defining for $\lmc(C)$:

\begin{corollary}\label{cor:cut_facet_equivalence}
    For the cycle $C=(\Z_n,E)$, a cut inequality \eqref{eq:cut} with respect to a $vw$-cut $\{e,e'\}$ defines a facet of $\lmc(C)$ if and only if $e$ and $e'$ are adjacent.
\end{corollary}


\section{Conclusion}

We define and analyze the lifted multicut polytope, establishing conditions under which the canonical box, cycle, path and cut inequalities define facets.
In particular, we characterize which cycles and paths in the graph give rise to facet-defining inequalities, thereby generalizing a classic result of \citet{Chopra1993}. 

For the special case of lifting from a tree, we analyze the connections between the lifted multicut polytope and the multilinear polytope, identifying two new classes of valid inequalities:
the \emph{intersection inequalities} that are facet-defining, and the \emph{generalized intersection inequalities} that generalize the previously known cutting planes and are facet-defining under two necessary conditions we establish.
For the further specialization of lifting from a path to the complete graph, the intersection inequalities, together with some canonical inequalities, constitute a totally dual integral formulation of the lifted multicut polytope.
This result relates the geometry of the path partition problem to the combinatorial properties of the sequential set partition problem. 

Complementary to trees, we study the lifted multicut polytope for lifting from cycles, starting from two observations: 
1.~While the lifted multicut problem for cycles can be solved efficiently, the lifted multicut polytope for cycles defies a simple description.
2.~While the lifted multicut polytope for cycles is a subset of the multicut polytope for the complete graph, known classes of facets of the multicut polytope for complete graphs, with only a few canonical exceptions, are not facet-defining for the lifted multicut polytope for cycles.
Motivated by these observations, we establish several large classes of facet-defining inequalities for the lifted multicut polytope for cycles. 
In particular, we introduce the class of \emph{half-chorded odd cycle inequalities} that are facet-defining for both, the lifted multicut polytope for cycles and the multicut polytope for complete graphs. 
Under isomorphic transformation, these constitute a new class of facet-defining inequalities also for the much studied \emph{clique partitioning polytope} \citep{Groetschel1990}. 
Finally, we establish the first non-canonical facet-defining inequalities for the lifted multicut polytope for arbitrary graphs that arise from cycles in that graph.

Directions for future work arise from two conjectures concerning cutting planes originating from specific structures.
The first involves the generalized intersection inequalities. 
In fact, numerical experiments such as \Cref{example:gen-int-ineq} suggest that these valid inequalities might be sufficient to describe completely the lifted multicut polytope for paths lifted to arbitrary graphs.
Secondly, we speculate that the half-chorded odd cycle inequalities are the only non-trivial facet-defining inequalities of the lifted multicut polytope of the cycle inherited form the multicut polytope of the complete graph (cf. \Cref{rem:facets-of-both-mc-and-lmc}). 
Trueness of this second conjecture would imply that the vast knowledge about the facial structure of the multicut polytope of the complete graph cannot be transferred to the lifted multicut polytope.

\paragraph{Acknowledgements} 
The authors acknowledge a contribution by Andrea Hor\v{n}\'akov\'a who has stated Condition~2 of Theorem~\ref{thm:cut-facets} in its present general form.

\appendix
\section{Proofs}\label{sec:proofs}

\begin{delayedproof}{prop:mc-apx-hard}
    Given an undirected graph $G=(V,E)$, a partition of $E$ into \emph{attractive} edges $E^+$ and \emph{repulsive} edges $E^-$, and non-negative edge weights $w \in \R_{\geq 0}^E$ the \emph{maximum agreement correlation clustering} \cite{Bansal2004} problem is the following optimization problem:
    \begin{align}\label{eq:max-agree-cc}
        \max_{x \in \mc(G)} \sum_{e \in E^+} w_e (1 - x_e) + \sum_{e \in E^-} w_e x_e \enspace .
    \end{align}
    Here, the objective is to maximize the sum of the weights of the edges that agree with the signature of the edges, i.e. the edges in $E^+$ that are not cut with respect to $x$ and the edges in $E^-$ that are cut with respect to $x$.
    Unless \textsc{p=np}, \citet{Charikar2005} show that for all $\epsilon > 0$, \eqref{eq:max-agree-cc} cannot be approximated within factor $\frac{79}{80}+\epsilon$ for arbitrary non-negative weights and within factor $\frac{115}{116} + \epsilon$ for unit weights $w_e = 1$ for all $e \in E$.

    We define the cost vector $c \in \R^E$ with $c_e = w_e$ for $e \in E^+$ and $c_e=-w_e$ for $e \in E^-$.
    The objective function of the multicut problem with respect to graph $G$ and costs $c$ can be written as
    \begin{align}\label{eq:mc-cc-identity}
        \sum_{e \in E} c_e x_e = 
        \sum_{e \in E^+} w_e 
        - \left(\sum_{e \in E^+} w_e (1-x_e) + \sum_{e \in E^-} w_e x_e \right)
    \end{align}
    where the term inside the parentheses is precisely the objective of \eqref{eq:max-agree-cc}.
    Now, suppose there exists an $\alpha$-approximation algorithm for the multicut problem for some $0 < \alpha \leq 1$. 
    Let $\text{OPT}_{\text{MC}}$ and $\text{OPT}_{\text{CC}}$ be the optimal solution of the multicut problem and the maximum agreement correlation clustering problem for the given instance and let $\text{A}_{\text{MC}}$ and $\text{A}_{\text{CC}}$ be the respective objective values that are archived by the solution of the approximation algorithm.
    By definition, it holds that $\text{A}_{\text{MC}} \leq \alpha \text{OPT}_{\text{MC}}$ (note that $\text{OPT}_{\text{MC}} \leq 0$ as $x=0$ is always feasible).
    Using the identity \eqref{eq:mc-cc-identity} we obtain
    \begin{align*}
        W^+ - \text{A}_{\text{CC}} &\leq \alpha (W^+ - \text{OPT}_{\text{CC}}) \\
        \Rightarrow \qquad \qquad
        \text{A}_{\text{CC}} &\geq \alpha \text{OPT}_{\text{CC}} + (1-\alpha) W^+ \\
        \Rightarrow \qquad \qquad
        \text{A}_{\text{CC}} &\geq \alpha \text{OPT}_{\text{CC}}
    \end{align*}
    where $W^+ := \sum_{e \in E^+} w_e \geq 0$.
    Therefore, an $\alpha$-approximation algorithm for the multicut problem is also an $\alpha$-approximation algorithm for the maximum agreement correlation clustering problem.
    By the above hardness results from \citet{Charikar2005} the claim follows.
\end{delayedproof}

\begin{delayedproof}{prop:lifted-multicut-polytope-inequalities}
Let $x \in \{0,1\}^{E \cup F}$ be such that $x = \1_M$ for a multicut $M$ of $\widehat G$ lifted from $G$.
Every cycle in $G$ is a cycle in $\widehat G$. 
Moreover, for any $\{u,v\} = f \in F$ and any $uv$-path $P=(V_P,E_P)$ in $G$, it holds that $E_P \cup \{f\}$ induces a cycle in $\widehat G$.
Therefore, $x$ satisfies all inequalities \eqref{eq:lifted-multicut-cycle} and \eqref{eq:lifted-multicut-path} by \Cref{prop:multicuts-statisfy-cycle-ineq}.
Assume $x$ violates some inequality of \eqref{eq:lifted-multicut-cut}.
Then, there is an edge $uv \in F$ and some $uv$-cut $\delta(U)$ in $G$ such that $x_{uv} = 0$ and for all $e \in \delta(U)$ we have $x_e = 1$.
Let $\Pi = \phi_{\widehat{G}}^{-1}(M)$ be the partition of $V$ that is induced by $M$.
Due to $x_{uv} = 0$, there exists some $W \in \Pi$ with $u \in W$ and $v \in W$. 
However, for any $ww' \in \delta(U)$ it holds that $w \notin W$ or $w' \notin W$ by $x_{ww'}=1$. 
This means the induced subgraph $(W,E \cap \binom{W}{2} )$ is not connected, as $\delta(U)$ is a $uv$-cut.
Hence, $\Pi$ is not a decomposition of $G$, which is a contradiction.

Now, suppose $x \in \{0,1\}^{E \cup F}$ satisfies the inequalities \eqref{eq:lifted-multicut-cycle}--\eqref{eq:lifted-multicut-cut} as specified. 
We show first that $M = x^{-1}(1)$ is a multicut of $\widehat G$.
Assume the contrary, then, by \Cref{prop:multicuts-statisfy-cycle-ineq}, there is a cycle $C=(V_C,E_C)$ in $\widehat G$ and some edge $e$ such that $E_C \cap M = \{e\}$.
For every $uv = f \in F \cap E_C \setminus \{e\}$ there exists a $uv$-path $P=(V_P,E_P)$ in $G$ such that $x_{e'} = 0$ for all $e' \in E_P$.
Otherwise there would be some $uv$-cut in $G$ violating \eqref{eq:lifted-multicut-cut}, as $G$ is connected.
If we replace every such $f$ with its associated path $P$ in $G$, then the resulting cycle violates either \eqref{eq:lifted-multicut-cycle} (if $e \in E$) or \eqref{eq:lifted-multicut-path} (if $e \in F$).
Thus, $M$ is a multicut of $\widehat G$ and the induced partition $\Pi = \phi_{\widehat G}^{-1}(M)$ is a decomposition of $\widehat{G}$.
Assume $\Pi$ is not a decomposition of $G$. 
Then, there exists a component $U \in \Pi$ that is not connected in $G$, i.e. there exist $vw \in F$ with $v,w \in U$ such that every $vw$-path in $G$ is cut with respect to $x$.
However, in that case the path inequality \eqref{eq:lifted-multicut-path} with respect to $f$ and that $vw$-path is violated, contradicting our assumption.
Therefore, $\Pi$ is also a decomposition of $G$ and hence, $M$ is indeed lifted from $G$.
\end{delayedproof}

\begin{delayedproof}{thm:dimension}
    As $0 \in \lmc(G,\widehat{G})$ the dimension of $\lmc(G,\widehat{G})$ is equal to the dimension of the vector space that is spanned by $\lmc(G,\widehat{G})$.
    We prove the claim by constructing all unit vectors $\1_{\{uv\}}$ for $uv \in E \cup F$ as linear combinations of characteristic vectors of multicuts of $\widehat{G}$ lifted from $G$.

    For $uv \in E \cup F$ let $P=(V_P,E_P)$ be a $uv$-path in $G$. 
    We define $V_1 = V_P$, $V_2 = V_P \setminus \{u,v\}$, $V_3 = V_P \setminus \{v\}$ and $V_4 = V_P \setminus \{u\}$. 
    For $i=1,\dots,4$, let $\Pi_i = \{V_i\} \cup \{\{w\} \mid w \in V \setminus V_i\}$ be the partition of $V$ that consists of the set $V_i$ and otherwise singular nodes. 
    Clearly the partitions $\Pi_i$, for $i=1,\dots,4$, are decompositions of $G$. 
    For $i=1,\dots,4$, let $x^i := \1_{\phi_{\widehat{G}}(\Pi_i)} \in \lmc(G,\widehat{G})$  be the characteristic vector of the multicut of $\widehat{G}$ lifted from $G$ that is induced by the decomposition $\Pi_i$. 
    For $st \in E \cup F$ the following holds:
    \begin{itemize}
        \item $x^i_{st} = 1$ for $i=1,\dots,4$ if $\{s,t\} \not\subseteq V_P$,
        \item $x^i_{st} = 0$ for $i=1,\dots,4$ if $\{s,t\} \subseteq V_P \setminus \{u,v\}$,
        \item $x^i_{st} = 0$ for $i=1, 3$ and $x^i_{st}=1$ for $i=2,4$ if $\{s,t\} \subseteq V_P$, $u \in \{s,t\}$, $v \notin \{s,t\}$,
        \item $x^i_{st} = 0$ for $i=1,4$ and $x^i_{st} = 1$ for $i=2,3$ if $\{s,t\} \subseteq V_P$, $v \in \{s,t\}$, $u \notin \{s,t\}$,
        \item $x^1_{st} = 0$ and $x^i_{st} = 1$ for $i=2,3,4$ if $\{s,t\} = \{u,v\}$.
    \end{itemize}
    Altogether we have $\1_{\{uv\}} = -x^1 - x^2 + x^3 + x^4$ which concludes the proof.
\end{delayedproof}

\begin{delayedproof}{prop:inclusion-property}
The set of all multicuts of $\widehat{G}$ lifted from $G$ is $\phi_{\widehat{G}}(D_G)$, which is the set of all multicuts of $\widehat{G}$ that are induced by decompositions of $G$. 
With this we can write
\[
    \lmc(G,\widehat{G}) = 
    \conv \big \{ \1_{\phi_{\widehat{G}}(\Pi)} \mid \Pi \in D_G \big \} \enspace.
\]
Therefore, to prove the Lemma it suffices to show $D_{G'} \subseteq D_G$ if and only if $E' \subseteq E$.

In case $E' \subseteq E$, any set $U \subseteq V$ that is connected in $G'$ is also connected in $G$. 
Thus, any decomposition of $G'$ is also a decomposition of $G$, i.e. $D_{G'} \subseteq D_G$. 
If otherwise it holds that $E' \not\subseteq E$ there exists $uv \in E'$ with $uv \notin E$. 
Then, the partition $\Pi = \{\{u,v\}\} \cup \{ \{w\} \mid w \in V \setminus \{u,v\}\}$ is a decomposition of $G'$ but not of $G$ because the node set $\{u,v\}$ is not connected in $G$, i.e. $D_{G'} \not\subseteq D_{G}$.
\end{delayedproof}

\begin{delayedproof}{lem:valid-and-facet-defining-for-subset}
Let $a^\top x \leq b$ be valid for $\lmc(G,\widehat{G})$. 
By \Cref{prop:inclusion-property}, the inequality is also valid for $\lmc(G',\widehat{G})$. 
By \Cref{thm:dimension}, both polytopes $\lmc(G',\widehat{G})$ and $\lmc(G,\widehat{G})$ have full dimension $m:=|E \cup F|$. 
If $a^\top x \leq b$ is facet-defining for $\lmc(G',\widehat{G})$, there exist $m$ affinely independent vectors $x^1,\dots,x^m \in \lmc(G',\widehat{G})$ that satisfy $a^\top x \leq b$ with equality. 
By \Cref{prop:inclusion-property}, it holds that $x^1,\dots,x^m \in \lmc(G,\widehat{G})$ and hence $a^\top x \leq b$ is also facet-defining for $\lmc(G,\widehat{G})$.
\end{delayedproof}

\begin{delayedproof}{lem:facet-for-lifting-to-sub-graph}
For a decomposition $\Pi$ of $G$, let $x = \1_{\phi_{\widehat{G}'}(\Pi)}$ and $y = \1_{\phi_{\widehat{G}}(\Pi)}$ be the characteristic vectors of the multicuts of $\widehat{G}'$ and $\widehat{G}$, respectively lifted from $G$, induced by the decomposition $\Pi$. 
It holds that $y_e = x_e$ for $e \in E \cup F$ and, thus, $\lmc(G,\widehat{G})$ is obtained from $\lmc(G,\widehat{G}')$ by projecting out the variables $x_e$ for $e \in F' \setminus F$.
By the assumption $E_a \subseteq E \cup F$ it follows $a^\top x = \bar{a}^\top y$.
Therefore, validity of $a^\top x \leq b$ for $\lmc(G,\widehat{G}')$ implies the validity of $\bar{a}^\top y \leq b$ for $\lmc(G,\widehat{G})$. 
Further if $x$ satisfies $a^\top x = b$ then $y$ satisfies $\bar{a}^\top y = b$.
By \Cref{thm:dimension}, it holds that $\dim\lmc(G,\widehat{G}') = m := |E \cup F'|$ and, since $a^\top x \leq b$ is facet-defining, there exist $m$ affinely independent vectors $x^1,\dots,x^m \in \lmc(G,\widehat{G}')$ satisfying $a^\top x = b$. 
Let $y^1,\dots,y^m \in \lmc(G,\widehat{G})$ be the vectors that are obtained by deleting the dimensions $e$ for $e \in F' \setminus F$ from $x^1,\dots,x^m$.
Since the matrix $A$ with rows $x^2-x^1,\dots,x^m-x^1$ has rank $m-1$ the matrix with rows $y^2-y^1,\dots,y^m-y^1$ that is obtained by deleting $|F'\setminus F|$ many columns of $A$ has at least rank $m-1-|F'\setminus F| = |E \cup F| - 1$.
Therefore the set $\{y^1,\dots,y^m\}$ contains at least $|E \cup F|$ affine independent vectors all satisfying $\bar{a}^\top y \leq b$ with equality. Hence, the inequality is indeed facet-defining for $\lmc(G,\widehat{G})$.
\end{delayedproof}

\begin{delayedproof}{thm:upper-box-facets}
Let $S = \{x \in \lmc(G,\widehat{G}) \cap \Z^{E \cup F} \mid x_e=1\}$ and put $\Sigma = \conv S$. 

To show necessity, suppose there is some $uv \in F \setminus \{e\}$ such that $s$ and $t$ are $uv$-cut-nodes.
Then, for any $uv$-path $P=(V_P,E_P)$ in $G$, it holds that $s,t \in V_P$, i.e. either $e \in E_P$ or $e$ is a chord of $P$.
We claim that we have $x_{uv}=1$ for any $x \in S$.
This gives $\dim \Sigma \leq \abs{E \cup F} - 2$, so the inequality $x_e \leq 1$ cannot define a facet of $\lmc(G,\widehat{G})$.
If there are no $uv$-paths that have $e$ as a chord, then $\{e\}$ is a $uv$-cut and the claim follows from the corresponding cut inequality \eqref{eq:lifted-multicut-cut}.
Otherwise, every $uv$-path $P$ that has $e$ as a chord contains a $st$-subpath $P'=(V_{P'},E_{P'})$ such that $E_{P'} \cup \{e\}$ induces a cycle.
Thus, for any $x \in S$, the inequalities \eqref{eq:lifted-multicut-cycle} or \eqref{eq:lifted-multicut-path} (for $e \in E$ or $e \in F$, respectively) imply the existence of some $e_{P'} \in E_{P'}$ such that $x_{e_{P'}} = 1$.
Let $\mathcal{P}$ denote the set of all such paths $P'$.
It is easy to see that the collection $\bigcup_{P' \in \mathcal{P}} \{e_{P'}\} \cup \{e\}$ contains a $uv$-cut.
This gives $x_{uv} = 1$ via the corresponding cut inequality \eqref{eq:lifted-multicut-cut}.

We turn to the proof of sufficiency. 
By \Cref{thm:dimension}, we have to show that $\dim \Sigma = \abs{E \cup F} - 1$.
The dimension of $\Sigma$ is equal to the dimension of the vector space spanned by $L = \{x - y \mid x,y \in S\}$.
We prove the claim by showing that $L$ contains $\abs{E \cup F}-1$ unit vectors.
Assume there is no $uv \in F \setminus \{e\}$ such that $s$ and $t$ are $uv$-cut-nodes in $G$.
By this assumption, for every $uv \in E \cup F \setminus \{e\}$ there exists a $uv$-path $P=(V_P,E_P)$ in $G$ with $e \not\subseteq V_P$. 
Let $V_i$, $\Pi_i$ and $x^i$ for $i=1,\dots,4$ be defined as in the proof of \Cref{thm:dimension}.
Then, we have $x^i \in S$ for $i=1,\dots,4$ and it holds that $\1_{\{uv\}} = -x^1 - x^2 + x^3 + x^4 \in L$ which concludes the proof.
\end{delayedproof}

\begin{delayedproof}{thm:lower-box-facets}
Let $S = \{x \in \lmc(G,\widehat{G}) \cap \Z^{E \cup F} \mid x_e=0\}$ and put $\Sigma = \conv S$.

Consider the case that $e \in E$. 
Let $G_{[e]}$ and $\widehat G_{[e]}$ be the graphs obtained from $G$ and $\widehat G$, respectively, by contracting the edge $e$ (and subsequently merging parallel edges). 
The lifted multicuts $x^{-1}(1)$ for $x \in S$ correspond bijectively to the multicuts of $\widehat G_{[e]}$ lifted from $G_{[e]}$. 
This implies $\dim \Sigma = \dim \lmc(G_{[e]}, \widehat G_{[e]})$. 
The claim follows from \Cref{thm:dimension} and the fact that $\widehat G_{[e]}$ has $\abs{E \cup F} - 1$ many edges if and only if $e$ is not contained in any triangle in $\widehat G$.

Now, suppose $\{u,v\} = e \in F$. 
We show necessity of the Conditions~\ref{cond:box-1}--\ref{cond:box-3} by proving that if any of them is violated, then all $x \in S$ satisfy some additional equation and thus, $\dim \Sigma \leq \abs{E \cup F} - 2$.

First, assume that \ref{cond:box-1} is violated. 
Hence, there are edges $e',e'' \in E \cup F$ such that $\{e,e',e''\}$ induces a triangle in $\widehat G$. 
The triangle inequalities
\begin{align}
x_{e'} & \leq x_e + x_{e''} \enspace, \label{eq:proof-facets-box-lower-1} \\
x_{e''} & \leq x_e + x_{e'} \enspace, \label{eq:proof-facets-box-lower-2} 
\end{align}
are cycle inequalities for cycles of length three.
By \Cref{prop:lifted-multicut-polytope-inequalities} these inequalities are valid for $\mc(\widehat{G})$ and by \Cref{prop:inclusion-property} they are valid for $\lmc(G,\widehat{G})$, in particular every $x \in S$ satisfies the triangle inequalities.
Thus, by \eqref{eq:proof-facets-box-lower-1}, \eqref{eq:proof-facets-box-lower-2} and $x_e = 0$, every $x \in S$ satisfies $x_{e'} = x_{e''}$.

Next, assume that \ref{cond:box-2} is violated. 
Consider a violating pair $u'v' \neq uv, u' \neq v'$ of $uv$-cut-nodes. 
For every $x \in S$, there exists a $uv$-path $P=(V_P,E_P)$ in $G$ with $x_f = 0$ for all $f \in E_P$, as $x_e = 0$.
Any such path $P$ has a sub-path $P'=(V_{P'},E_{P'})$ from $u'$ to $v'$ because $u'$ and $v'$ are $uv$-cut-nodes. We distinguish the following cases.
\begin{itemize}
\item If the distance of $u'$ and $v'$ in $\widehat G$ is 1, then $u'v' \in E \cup F$.
If $u'v' \in E_P$, then $x_{u'v'} = 0$ because $x_f = 0$ for all $f \in E_P$.
If otherwise $u'v' \not\in E_P$, we obtain $x_{u'v'} = 0$ by $x_f = 0$ for all $f \in E_{P'}$ and the cycle/path inequality
\begin{align*}
x_{u'v'} \leq \sum_{f \in E_{P'}} x_{f} \enspace.
\end{align*}
Thus $x_{u'v'} = 0$ for all $x \in S$.

\item If the distance of $u'$ and $v'$ in $\widehat G$ is 2, there is a $u'v'$-path in $\widehat G$ consisting of two distinct edges $e',e'' \in E \cup F$.
We show that all $x \in S$ satisfy $x_{e'} = x_{e''}$:
\begin{itemize}
\item If $e' \in E_P$ and $e'' \in E_P$, then $x_{e'} = x_{e''} = 0$ because $x_f = 0$ for all $f \in E_P$.
\item If $e' \in E_P$ and $e'' \notin E_P$ then $x_{e'} = x_{e''} = 0$ by $x_f = 0$ for all $f \in E_{P'}$ and the cycle/path inequality
\begin{align*}
x_{e''} \leq \sum_{f \in E_{P'} \setminus \{e'\}} x_{f} \enspace.
\end{align*}

\item If $e' \notin E_P$ and $e'' \notin E_P$  then $x_{e'} = x_{e''}$ by $x_f = 0$ for all $f \in E_{P'}$ and the cycle/path inequalities
\begin{align*}
x_{e''} & \leq x_{e'} + \sum_{f \in E_{P'}} x_{f} \\
x_{e'} & \leq x_{e''} + \sum_{f \in E_{P'}} x_{f}
\end{align*}
which are valid for $x$ as they are valid for $\mc(\widehat{G})$ by \Cref{prop:lifted-multicut-polytope-inequalities}.
\end{itemize}
\end{itemize}

Finally, assume that \ref{cond:box-3} is violated.
Hence, there exists a $uv$-cut-node $t$ and a $uv$-separating set of nodes $\{s,s'\}$ such that $\{ts, ts', ss'\}$ induces a triangle in $\widehat G$. 
We have that all $x \in S$ satisfy $x_{ss'} = x_{ts} + x_{ts'}$ as follows. 
At most one of $x_{ts}$ and $x_{ts'}$ is $1$, because $t$ is a $uv$-cut-node and $ss'$ is $uv$-separating as well. 
Moreover, it holds that $x_{ts} + x_{ts'} = 0$ if and only if $x_{ss'} = 0$ by the associated triangle inequalities.
\end{delayedproof}

\begin{delayedproof}{thm:cycle-facets}
By Theorem 3.2 of \cite{Chopra1993}, for any chordal cycle $\widehat{C}=(V_{\widehat{C}},E_{\widehat{C}})$ in $\widehat G$ and any $f \in E_{\widehat{C}}$, the associated cycle inequality
\begin{align}\label{eq:chordal-cycle}
    x_f \leq \sum_{e \in E_{\widehat{C}} \setminus \{f\}} x_e 
\end{align}
is not facet-defining for $\mc(\widehat G)$. 
This implies that \eqref{eq:chordal-cycle} is not facet-defining for $\lmc(G,\widehat G)$ as $\lmc(G,\widehat G) \subseteq \mc(\widehat G)$ by \Cref{prop:inclusion-property} and $\dim \lmc(G,\widehat G) = \dim \mc(\widehat G)$ by \Cref{thm:dimension}.
Hence, this shows necessity for both \ref{enum:chordless-cycle} and \ref{enum:chordless-path}.

For the proof of sufficiency, suppose the cycle $C=(V_C,E_C)$ of $G$ is chordless in $\widehat G$ and let $f \in E_C$. 
Let 
\[
    S = \left\{ x \in \lmc(G,\widehat{G}) \cap \Z^{E \cup F} 
        \;\middle|\; x_f = \sum\nolimits_{e \in E_C \setminus \{f\}} x_e \right\}
\] 
and define $\Sigma = \conv S$.
Let $\Sigma'$ be a facet of $\lmc(G,\widehat{G})$ such that $\Sigma \subseteq \Sigma'$ and suppose it is induced by the inequality $a^\top x \leq \alpha$ with $a \in \R^{E \cup F}$ and $\alpha \in \R$, i.e., $\Sigma' = \conv S'$, where
\begin{align*}
    S' = \left\{ x \in \lmc(G,\widehat{G}) \cap \Z^{E \cup F} \mid a^\top x = \alpha \right\} \enspace.
\end{align*}
As $0 \in S \subseteq S'$, we have $\alpha = 0$. 
We show that $a^\top x \leq \alpha$ is a scalar multiple of a cycle inequality \eqref{eq:lifted-multicut-cycle} and thus $\Sigma = \Sigma'$.

Let $y \in \{0,1\}^{E \cup F}$ be defined by $y_e = 0$ for $e \in E_C$ and $y_e = 1$ for $e \notin E_C$, i.e.\ all edges except $E_C$ are cut. 
Then, $y \in S \subseteq S'$, since $C$ is chordless. 
For any $e \in E_C \setminus \{f\}$, the vector $x \in \{0,1\}^{E \cup F}$ with $x_{e'} = 0$ for $e' \in E_C \setminus \{f,e\}$ and $x_{e'} = 1$ for $e' \notin E_C \setminus \{f,e\}$ satisfies $x \in S \subseteq S'$.
Therefore, $a^\top (y-x) = 0$ and thus
\begin{align}\label{eq:proof-cycle-facets-1}
    a_{e} = -a_f \quad \forall e \in E_C \setminus \{f\} \enspace. 
\end{align}

It remains to show that $a_{uv} = 0$ for all edges $uv \in E \cup F \setminus E_C$. 

First, suppose $u,v \notin V_C$.
Let $P = (V_P,E_P)$ be a $uv$-path in $G$. 
We proceed similarly to the proof of \Cref{thm:dimension}. 
If $V_P \cap V_C = \emptyset$ let $V_1 = V_P$, otherwise let $V_1 = V_P \cup V_C$. 
Let $V_2 = V_1 \setminus \{u,v\}$, $V_3 = V_1 \setminus \{u\}$ and $V_4 \setminus \{v\}$. 
If $V_P \cap V_C = \emptyset$ let $\Pi_i = \{V_i, V_C\} \cup \{\{w\} \mid w \in V \setminus (V_i \cup V_C)\}$, otherwise let $\Pi_i = \{V_i\} \cup \{\{w\} \mid w \in V \setminus V_i\}$ for $i=1,\dots,4$. 
Let $x^i = \1_{\phi_{\widehat{G}}(\Pi_i)}$ be the characteristic vector of the multicut induced by the decomposition $\Pi_i$ for $i=1,\dots,4$. 
By construction, it holds that $x^i \in S \subseteq S'$. 
As in the proof of \Cref{thm:dimension}, for 
\[
    x = - x^1 - x^2 + x^3 + x^4
\]
it holds that $x_e = 1$ and $x_{e'} = 0$ for all other $e' \in E \cup F \setminus \{e\}$. It holds that $a^\top x = 0$, which yields $a_e = 0$.

Next, for $v \in V \setminus V_C$ we show $a_e = 0$ for all $e = vu \subseteq E \cup F$ with $u \in V_C$. 
Let $w \in V_C$ be such that there exists a $vw$-path $P=(V_P,E_P)$ with $V_P \cap V_C = \{w\}$.
We pick a direction on $C$ and traverse $C$ from one endpoint of $f$ to the other endpoint of $f$ according to that direction. 
Let $e_i = vu_i$ for $i=1,...,k$ be an ordering of all edges $e = vu \in E \cup F$ with $u \in V_C$ such that $u_i$ comes before $u_{i+1}$ on the traversal of $C$.
Let $m \in \{1,\dots,k\}$ such that $w=u_m$ or $w$ comes after $u_m$ but before $u_{m+1}$ on the traversal of $C$. 
For $i \in \{1,\dots,k\}$ with $i < m$ let $e \in E_C$ be an edge between $u_i$ and $u_{i+1}$. 
Let $U, U' \subseteq V_C$ be the two components of $C$ that are obtained when cutting the edges $f$ and $e$. 
We may assume $w \in U$ by potentially interchanging $U$ and $U'$. 
Let $V^i_1 = V_P \cup U$ and let $V^i_2 = V^i_1 \setminus \{v\}$ for $i=1,\dots,m-1$. We define $\Pi^i_j = \{V^i_j,U'\} \cup \{\{s\}\mid s \in V \setminus (V^i_j \cup U')\}$ for $j=1,2$. 
For an illustration, see \Cref{fig:proof-cycle-facet}.
Additionally we define $V^0_1 = V_P \cup V_C$, $V^0_2 = V^0_1 \setminus \{v\}$ and $\Pi^0_j = \{V^0_j\} \cup \{\{s\} \mid s \in V \setminus V^0_j\}$ for $j=1,2$.
For fixed $i \in \{1,...,k\}$, $i < m$ we define $x^1 = \1_{\phi_{\widehat{G}}(\Pi^{i-1}_1)}$, $x^2 = \1_{\phi_{\widehat{G}}(\Pi^i_2)}$, $x^3 = \1_{\phi_{\widehat{G}}(\Pi^{i}_1)}$ and $x^4 = \1_{\phi_{\widehat{G}}(\Pi^{i-1}_2)}$.
By construction, it holds that $x^j \in S \subseteq S'$ for $j=1,\dots,4$. As in the proof of \Cref{thm:dimension}, for
\[
    x = - x^1 - x^2 + x^3 + x^4
\]
it holds that $x_{vu_i} = 1$ and $x_{e'} = 0$ for all other edges $e' \in E \cup F \setminus \{vu_i\}$. 
Therefore, $a^\top x = 0$ yields $a_{vu_i} = 0$. 
This holds for all $i \in \{1,\dots,k\}$ with $i < m$. 
By reversing the direction on $C$ we also obtain $a_{vu_i} = 0$ for all $i \in \{1,\dots,k\}$ with $i > m$.
It remains to show $a_{vu_m} = 0$.
To that end let $x^1 = \1_{\phi_{\widehat{G}}(\Pi^0_1)}$, $x^2 = \1_{\phi_{\widehat{G}}(\Pi^0_2)}$. 
Again, it holds that $x^1,x^2 \in S \subseteq S'$. 
For $x = x^2 - x^1$ it holds that $x_{vs} = 1$ for all $vs \in E \cup F$ with $s \in V^0_2$ and $x_{e'} = 0$ for all other edges $e'$. 
By above we have that $a_{vs} = 0$ for all $vs \in E \cup F$ with $s \in V^0_2 \setminus \{u_m\}$ and $a^\top x = 0$ yields the desired $a_{vu_m} = 0$.

This concludes the proof of sufficiency of the first assertion.
The proof of sufficiency in the second assertion is completely analogous (consider the cycle $C$ that is obtained by adding $f$ to the path $P$).
The chosen multicuts remain valid, because $f$ is the only edge in the cycle that is not contained in $E$.
\end{delayedproof}

\begin{figure}
    \center
    \input{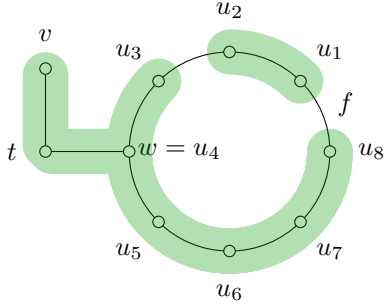}
    \caption{Depicted is the decomposition $\Pi^2_1$ from the proof of \Cref{thm:cycle-facets}. Here the path $P$ is the $vw$-path along the nodes $v$, $t$ and $w$. The cycle is cut at the edges $f$ and $e=\{u_2,u_3\}$ and the component $U$ that contains $w$ is $\{u_3,...,u_8\}$, the other component is $U' = \{u_1,u_2\}$.}
    \label{fig:proof-cycle-facet}
\end{figure}

\begin{delayedproof}{lem:connectedness}
Take some $x \in S(uv, U)$.
Let $E_0 = \{e \in E \mid x_e = 0\}$ and consider $G_0 = (V, E_0)$.

If $x_{uv} = 1$ then for all $e \in \delta(U)$ it holds that $x_e = 1$.
Thus, no component of $G_0$ is $(uv, U)$-connected.

If $x_{uv} = 0$ then, due to $x \in S(uv, U)$, there is some $e \in \delta(U)$ such that
\begin{align}\label{eq:proof-lem-connectedness-1}
    x_e = 0 \quad \text{ and } \quad 
        x_{e'} = 1 \quad \forall e' \in \delta(U) \setminus \{e\} \enspace.
\end{align}
Let $H=(V_H,E_H)$ be the maximal component of $G_0$ with
\begin{align}\label{eq:proof-lem-connectedness-2}
    e \in E_H.
\end{align}
Clearly,
\begin{align}\label{eq:proof-lem-connectedness-3}
    e' \notin E_H \quad \forall e' \in \delta(U) \setminus \{e\}
\end{align}
by \eqref{eq:proof-lem-connectedness-1} and definition of $G_0$.
There is no $uv$-cut $\delta(W)$ with $x_{e'} = 1$ for all $e' \in \delta(W)$,
because this would imply $x_{uv} = 1$.
Thus, there exists a $uv$-path $P=(V_P,E_P)$ in $G$ with $x_{e'} = 0$ for all $e' \in E_P$, as $G$ is connected.
Any such path $P$ has $e \in E_P$, as $E_P \cap \delta(U) \neq \emptyset$ and $\delta(U) \cap E_0 = \{e\}$ and $E_P \subseteq E_0$.
Thus,
\begin{align}\label{eq:proof-lem-connectedness-4}
    u \in V_H \text{ and } v \in V_H
\end{align}
by \eqref{eq:proof-lem-connectedness-2}.
Therefore, $H=(V_H,E_H)$ is $(uv, U)$-connected, by \eqref{eq:proof-lem-connectedness-2}, \eqref{eq:proof-lem-connectedness-3} and \eqref{eq:proof-lem-connectedness-4}.
Any other component of $G_0$ does not cross the cut $\delta(U)$, by \eqref{eq:proof-lem-connectedness-1}, \eqref{eq:proof-lem-connectedness-2} and definition of $G_0$, and is not $(uv, U)$-connected.
\end{delayedproof}

\begin{delayedproof}{thm:cut-facets}
We show that if any of the Conditions~\ref{cond:cut-1}--\ref{cond:cut-5} is violated, then all $x \in S(uv,U)$ satisfy some additional equation and thus $\dim \Sigma(uv,U) \leq \abs{E \cup F} - 2$, which implies that $\Sigma(uv,U)$ cannot be a facet of $\lmc(G,\widehat{G})$, by \Cref{thm:dimension}.

Assume that Condition~\ref{cond:cut-1} does not hold.
Then, there exists an $e \in \delta(U)$ such that no $(uv,U)$-connected subgraph of $G$ contains $e$.
Thus, for all $x \in S(uv,U)$ it holds by \Cref{lem:connectedness} that $x_e = 1$.

Assume that Condition~\ref{cond:cut-2} does not hold. 
Then, there exits $\emptyset \neq F' \subseteq \delta_{F \setminus \{uv\}}(U)$ such that for any $e \in \delta(U)$ there exists some number $m \in \N$ such that for all $(uv,U)$-connected subgraphs $H=(V_H,E_H)$ with $e \in E_H$ it holds that $\abs{F' \cap F'_H}= m$. 
Thus, we can write
\begin{align*}
    \delta(U) = \bigcup_{m=0}^{\abs{F'}} \delta^m(U) \enspace,
\end{align*}
where
\begin{align*}
    \delta^m(U) = 
    \big\{e \in \delta(U) \mid
    m = \abs{F' \cap F'_H} \text{ for all } (uv,U)\text{-connected } (V_H,E_H) \\ 
    \text{ with } e \in E_H \big \}.
\end{align*} 
It follows that for all $x \in S(uv,U)$ we have the equality
\begin{align}\label{eq:proof-facet-cond2}
    \sum_{m=0}^{\abs{F'}} m \sum_{e \in \delta^m(U)} (1-x_e) = \sum_{f' \in F'} (1 - x_{f'})
\end{align}
by the following argument:
\begin{itemize}
    \item If $x_e = 1$ for all $e \in \delta(U)$, then $x_{f'} = 1$ for all $\{u', v'\} = f' \in F'$, since $\delta(U)$ is also a $u'v'$-cut.
    Thus, \eqref{eq:proof-facet-cond2} evaluates to $0=0$.
    \item Otherwise there exists precisely one edge $e \in \delta(U)$ such that $x_e = 0$. 
    Let $m$ be such that $e \in \delta^m(U)$.
    By definition of $\delta^m(U)$, there are exactly $m$ edges $f' \in F'$ with $x_{f'} = 0$.
    Thus, \eqref{eq:proof-facet-cond2} evaluates to $m=m$.
\end{itemize}

Assume that Condition~\ref{cond:cut-3} does not hold.
Then, there exists an $f' \in \delta_{F \setminus \{uv\}}(U)$, a set $\emptyset \neq F' \subseteq \delta_{F \setminus \{uv\}}(U) \setminus \{f'\}$ and some $k \in \mathbb{N}$ such that for all $(uv,U)$-connected subgraphs $H=(V_H,E_H)$ and $H'=(V_{H'},E_{H'})$ with $f' \in F'_H$ and $f' \notin F'_{H'}$ it holds that
\begin{align*}
    \abs{F' \cap F'_H} = k \text{ and } \abs{F' \cap F'_{H'}} = 0 \enspace.
\end{align*}
In other words, for all $x \in S(uv,U)$ it holds that $x_{f'} = 0$ if and only if there are exactly $k$ edges $f'' \in F'$ such that $x_{f''} = 0$. 
Similarly, it holds that $x_{f'} = 1$ if and only if for all $f'' \in F'$ we have $x_{f''} = 1$.
Therefore, all $x \in S(uv,U)$ satisfy the additional equation
\begin{align*}
    k(1-x_{f'}) = \sum_{f'' \in F'} 1 - x_{f''} \enspace.
\end{align*}

\begin{figure}
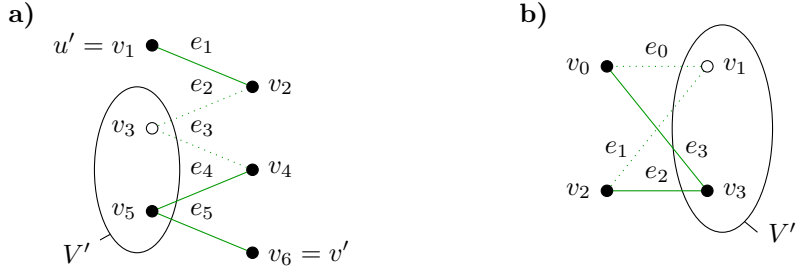

\center
\textbf{a)} \imagetop{\input{figures/cuts-proof-a.tex}}
\hspace{5em}
\textbf{b)} \imagetop{\input{figures/cuts-proof-b.tex}}
\caption[Illustration of the proof argument for Conditions~\ref{cond:cut-4} and \ref{cond:cut-5} in \Cref{thm:cut-facets}]{Depicted are the nodes (in black) and edges (in green) on a path \textbf{a)} and on a cycle \textbf{b)}, respectively.
Nodes in the set $V'$ are are either in $V_H$ (filled circle) or not in $V_H$ (empty circle).
Consequently, pairs of consecutive edges are either cut (dotted lines) or not cut (solid lines).}
\label{figure:proof-facet-aux}
\end{figure}

Assume that Condition~\ref{cond:cut-4} does not hold.
Then, there exist $u' \in U$ and $v' \in V \setminus U$ and a $u'v'$-path $P = (V_P, E_P)$ in $\widehat G(uv,U)$ such that for every $(uv,U)$-connected subgraph $H=(V_H,E_H)$ of $G$ it holds that
\begin{align}
    & (u' \in V_H \text{ and }  V_P \setminus U \subseteq V_H) 
    \label{eq:proof-facet-aux2} \\
    \text{or} \quad & (v' \in V_H \text{ and }  V_P \cap U \subseteq V_H) \enspace.
    \label{eq:proof-facet-aux3}
\end{align}
Let $v_1 < \dotso < v_{\abs{V_P}}$ be the linear order of the nodes $V_P$ and 
let $e_1 < \dotso < e_{\abs{E_P}}$ be the linear order of the edges $E_P$ 
in the $u'v'$-path $P$.
Now, all $x \in S(uv,U)$ satisfy the equation
\begin{align}\label{eq:proof-facet-6}
    x_{uv} = \sum_{j=1}^{\abs{E_P}} (-1)^{j+1} x_{e_j}
\end{align}
by the following argument. 
It holds that $\abs{E_P}$ is odd, as the path $P$ alternates between the set $U$ where it begins and $V \setminus U$ where it ends.
Thus, we can write 
\begin{align}\label{eq:proof-facet-aux4}
    \sum_{j=1}^{\abs{E_P}} (-1)^{j+1} x_{e_j} 
        = x_{e_1} - \sum_{j=1}^{(\abs{E_P}-1)/2} (x_{e_{2j}} - x_{e_{2j+1}}) \enspace.
\end{align}
Distinguish two cases:
\begin{itemize}
    \item If $x_{uv} = 1$, then $x_e = 1$ for all $e \in E_P$, by the cut inequalities \eqref{eq:lifted-multicut-cut} with respect to $\delta(U)$.
    Therefore, \eqref{eq:proof-facet-aux4} and thus \eqref{eq:proof-facet-6} evaluates to $1 = 1$.
    \item If $x_{uv} = 0$, then the decomposition of $G$ defined by $x$ contains precisely one $(uv,U)$-connected component $H=(V_H,E_H)$ of $G$, 
    by \Cref{lem:connectedness}.
    In particular it holds that
    \begin{align}
        x_{u'v'} = 0 \quad &\forall e = u'v' \in E \cup F \text{ with } u',v' \in V_H \enspace , \label{eq:C4-H-connected-wrt-x-1} \\
        x_{u'v'} = 1 \quad &\forall e = u'v' \in E \cup F \text{ with } u' \in V_H, v' \notin V_H  \enspace . \label{eq:C4-H-connected-wrt-x-2}
    \end{align}
    Without loss of generality we may assume that \eqref{eq:proof-facet-aux2} holds (otherwise exchange $u$ and $v$).
    Consider the nodes $V_P$ (as depicted in \Cref{figure:proof-facet-aux}a).
    It holds that $v_1 = u' \in V_H$, by \eqref{eq:proof-facet-aux2}.
    For every even $j$, $v_j \in V \setminus U$, by definition of $P$. 
    Thus,
    \begin{align}\label{eq:proof-facet-aux5}
        v_{2j} \in V_H \quad \forall \in \{1, \ldots, (\abs{E_P}-1)/2\}
    \end{align}
    by \eqref{eq:proof-facet-aux2}.
    Now consider the edges $E_P$ (as depicted in \Cref{figure:proof-facet-aux}a).
    It holds that $e_1 = v_1v_2$ and $v_1,v_2 \in V_H$, thus, 
    \begin{align}\label{eq:proof-facet-aux6}
        x_{e_1} = 0 \enspace ,
    \end{align}
    by \eqref{eq:C4-H-connected-wrt-x-1}.
    For every $j \in \{1, \ldots, (\abs{E_P}-1)/2\}$ we have $e_{2j} = v_{2j} v_{2j+1}$ and $e_{2j+1} = v_{2j+1}v_{2j+2}$ with $v_{2j},v_{2j+2} \in V_H$, by \eqref{eq:proof-facet-aux5}.
    If $v_{2j+1} \in V_H$, \eqref{eq:C4-H-connected-wrt-x-1} implies
    \begin{align*}
        x_{e_{2j}} = 0 = x_{e_{2j+1}} \enspace.
    \end{align*}
    If otherwise $v_{2j+1} \notin V_H$, \eqref{eq:C4-H-connected-wrt-x-2} implies
    \begin{align*}
        x_{e_{2j}} = 1 = x_{e_{2j+1}} \enspace.
    \end{align*}
    In any case, we have
    \begin{align}\label{eq:proof-facet-aux7}
        x_{e_{2j}} - x_{e_{2j+1}} = 0 \quad \forall j \in \{1, \ldots, (\abs{E_P}-1)/2\} \enspace.
    \end{align}
    Thus, \eqref{eq:proof-facet-6} evaluates to $0 = 0$, by \eqref{eq:proof-facet-aux4}, \eqref{eq:proof-facet-aux6} and \eqref{eq:proof-facet-aux7}.
\end{itemize}

Assume that Condition~\ref{cond:cut-5} does not hold.
Then, there exists a cycle $C = (V_C, E_C)$ in $\widehat G(uv,U)$
such that every $(uv,U)$-connected subgraph $H=(V_H,E_H)$ of $G$ satisfies
\begin{align}
    & V_C \cap U \subseteq V_H \label{eq:proof-facet-aux8}\\
    \text{or} \quad & V_C \setminus U \subseteq V_H \enspace. \label{eq:proof-facet-aux9}
\end{align}
Let $v_0 < \dotso < v_{\abs{V_C}-1}$ be an order on $V_C$ such that $v_0 \in U$ and for all $j \in \{0, \ldots, \abs{E_C}-1\}$ it holds that
\begin{align*}
    e_j = v_jv_{j+1 \bmod \abs{E_C}} \in E_C \enspace.
\end{align*}
Now, all $x \in S(uv,U)$ satisfy the equation
\begin{align}\label{eq:proof-facet-7}
    \sum_{j=0}^{\abs{E_C}-1} (-1)^j x_{e_j} = 0
\end{align}
by the following argument.
It holds that $\abs{E_C}$ is even, as the cycle $C$ alternates between the sets $U$ and $V \setminus U$.
Thus, 
\begin{align}\label{eq:proof-facet-aux10}
    \sum_{j=0}^{\abs{E_C}-1} (-1)^j x_{e_j}
    = \sum_{j=0}^{(\abs{E_C}-2)/2} (x_{e_{2j}} - x_{e_{2j+1}}) \enspace .
\end{align}
Distinguish two cases:
\begin{itemize}
    \item If $x_{uv} = 1$, then $x_e = 1$ for all $e \in E_C$, by the cut inequalities \eqref{eq:lifted-multicut-cut} with respect to \ $\delta(U)$.
    Therefore, \eqref{eq:proof-facet-aux10} and thus \eqref{eq:proof-facet-7} evaluates to $0 = 0$.
    \item If $x_{uv} = 0$, then the decomposition of $G$ defined by $x$ contains precisely one $(uv,U)$-connected component $H=(V_H,E_H)$ of $G$, 
    by \Cref{lem:connectedness}.
    As before, \eqref{eq:C4-H-connected-wrt-x-1} and \eqref{eq:C4-H-connected-wrt-x-2} hold true.
    Without loss of generality we may assume that \eqref{eq:proof-facet-aux8} holds (otherwise exchange $u$ and $v$).
    Consider the nodes $V_C$ (as depicted in \Cref{figure:proof-facet-aux}b).
    For every even $j$, we have that $v_j \in U$, by definition of $C$ and the order. 
    Thus,
    \begin{align}\label{eq:proof-facet-aux11}
        v_{2j} \in V_H \quad \forall j \in \{0, \ldots, (\abs{E_C}-2)/2\}
    \end{align}
    by \eqref{eq:proof-facet-aux8}.
    Now, consider the edges $E_C$ (as depicted in \Cref{figure:proof-facet-aux}b).
    For every $j \in \{0, \ldots, (\abs{E_C}-2)/2\}$ we have $e_{2j} = v_{2j}v_{2j+1}$ and $e_{2j+1} = v_{2j+1}v_{2j+2 \bmod \abs{E_C}}$ with $v_{2j}, v_{2j+2 \bmod \abs{E_C}} \in V_H$, by \eqref{eq:proof-facet-aux11}. 
    If $v_{2j+1} \in V_H$ \eqref{eq:C4-H-connected-wrt-x-1} implies
    \begin{align*}
        x_{e_{2j}} = 0 = x_{e_{2j+1}} \enspace.
    \end{align*}
    If otherwise $v_{2j+1} \notin V_H$, \eqref{eq:C4-H-connected-wrt-x-2} implies
    \begin{align*}
        x_{e_{2j}} = 1 = x_{e_{2j+1}} \enspace.
    \end{align*}
    In any case, we have
    \begin{align}\label{eq:proof-facet-aux12}
        x_{e_{2j}} - x_{e_{2j+1}} = 0 \quad \forall j \in \{0, \ldots, (\abs{E_C}-2)/2\} \enspace.
    \end{align}
    Thus, \eqref{eq:proof-facet-7} evaluates to $0 = 0$, by \eqref{eq:proof-facet-aux10} and \eqref{eq:proof-facet-aux12}.
\end{itemize}
\end{delayedproof}
    

\begin{delayedproof}{prop:tree-lmp}
    For any distinct pair of nodes $u,v \in V$, we set
    \begin{align*}
        x_{uv} = 1 - \prod_{e \in E_{uv}} z_e
    \end{align*}
    which implies
    \begin{align}\label{eq:vars-transform}
        x_{uv} = 0
        & \quad \iff \quad 
        \forall e \in E_{uv} \colon \ z_e = 1 \quad \iff \quad 
        \forall e \in E_{uv} \colon \ x_e = 0 \enspace. 
    \end{align}
    Therefore, we can reformulate problem \eqref{eq:tpp-pbo} in terms of the variables $x_{uv}$ by transforming the objective function according to
    \begin{align*}
        \bar \theta_{uv} \prod_{e \in E_{uv}} z_e 
        = -\bar \theta_{uv} \big(1 - \prod_{e \in E_{uv}} z_e \big) + \bar \theta_{uv} 
        = -\bar \theta_{uv} \, x_{uv} + \bar \theta_{uv} \enspace.
    \end{align*}
    This leads to the linear combinatorial optimization problem
    \begin{align*}
        \min_{x \in \lmc(T)} 
        \sum_{uv \in \binom{V}{2}} \theta_{uv} \, x_{uv} + \bar \theta_{uv} \enspace,
    \end{align*}
    where the definition of $\lmc(T)$ captures the relationship \eqref{eq:vars-transform}.
\end{delayedproof}

\begin{delayedproof}{prop:tree-polytope}
    We show first that $\lmc(T) \subseteq \tppt(T)$.
    For this purpose, let $x \in \lmc(T) \cap \Z^m$ be a vertex of $\lmc(T)$.
    If $x_{uv} > x_{u,\vec{u}(v)} + x_{\vec{u}(v),v}$ for some $u,v \in V$, then $x_{uv} = 1$ and $x_{u,\vec{u}(v)} = x_{\vec{u}(v),v} = 0$.
    This contradicts the fact that $x$ satisfies all cut inequalities with respect to $\vec{u}(v),v$ and the path inequality corresponding to $u,v$.
    If $x_{\vec{u}(v),v} > x_{uv}$ for some $u,v \in V$, then $x_{\vec{u}(v),v} = 1$ and $x_{uv} = 0$.
    This contradicts the fact that $x$ satisfies all cut inequalities with respect to $uv$ and the path inequality associated to $\vec{u}(v),v$.
    It follows that $x \in \tppt(T)$.
    
    Now, we show that $\tppt(T) \subseteq \tpp(T)$.
    Let $x \in \tppt(T)$.
    We need to show that $x$ satisfies all path and cut inequalities.
    Let $u,v \in V$ with $d(u,v) \geq 2$.
    We proceed by induction on $d(u,v)$.
    If $d(u,v) = 2$, then the path and cut inequalities are directly given by the definition of $\tppt(T)$ (for the two possible orderings of $u$ and $v$).
    If $d(u,v) > 2$, then the path inequality is obtained from $x_{uv} \leq x_{u,\vec{u}(v)} + x_{\vec{u}(v),v}$ and the induction hypothesis for the pair $\vec{u}(v),v$, since $d(\vec{u}(v),v) = d(u,v) - 1$.
    Similarly, for any edge $e$ on the path from $u$ to $v$, we obtain the cut inequality with respect to $e$ by using the induction hypothesis and $x_{\vec{u}(v),v} \leq x_{uv}$ such that (without loss of generality) $e$ is on the path from $\vec{u}(v)$ to $v$.
    It follows that $x \in \tpp(T)$.
\end{delayedproof}

\begin{delayedproof}{prop:tree-path-facet}
    First, suppose $d(u,v) = 2$.
    Then, $P_{uv}$ is a path of length $2$ and thus chordless in the complete graph on $V$.
    Hence, the facet-defining property follows directly from \Cref{thm:cycle-facets} \ref{enum:chordless-path}. Now, suppose $d(u,v) > 2$ and let $x \in \lmc(T)$ be such that \eqref{eq:tree-path} is satisfied with equality.
    We show that this implies
    \begin{align}\label{eq:intersection-equality}
        x_{uv} + x_{\vec{u}(v),\vec{v}(u)} = x_{u,\vec{v}(u)} + x_{\vec{u}(v),v} \enspace. 
    \end{align}
    Thus, the face of $\lmc(T)$ defined by \eqref{eq:tree-path} has dimension at most $m - 2$ and hence cannot be a facet.
    In order to check that \eqref{eq:intersection-equality} holds, we distinguish the following three cases.
    If $x_{uv} = 0$, then the cut inequalities \eqref{eq:cut-ineq-tree} yield $x_e = 0$ for all $e \in E_{uv}$ and with the path inequalities \eqref{eq:path-ineq-tree} it follows that all terms in \eqref{eq:intersection-equality} vanish.
    If $x_{uv} = x_{u,\vec{u}(v)} = 1$ and $x_{\vec{u}(v),v} = 0$, then $x_{\vec{u}(v),\vec{v}(u)} = 0$ and $x_{u,\vec{v}(u)} = 1$, so \eqref{eq:intersection-equality} holds.
    Finally, if $x_{uv} = x_{\vec{u}(v),v} = 1$ and $x_{u,\vec{u}(v)} = 0$, then~\eqref{eq:intersection-equality} holds as well, because $x_{\vec{u}(v),\vec{v}(u)} = x_{u,\vec{v}(u)}$ by contraction of the edge $u,\vec{u}(v)$.
\end{delayedproof}

\begin{delayedproof}{prop:tree-cut-facet}
    First, suppose $v$ is not a leaf of $T$ and let $x \in \lmc(T)$ be such that \eqref{eq:tree-cut} is satisfied with equality.
    Since $v$ is not a leaf, there exists a neighbor $w \in V$ of $v$ such that $P_{\vec{u}(v),v}$ is a subpath of $P_{\vec{u}(v),w}$
    We show that $x$ additionally satisfies the equality
    \begin{align}\label{eq:tree-cut-equality}
        x_{uw} = x_{\vec{u}(v),w} 
    \end{align}
    and thus the face of $\lmc(T)$ defined by \eqref{eq:tree-cut} cannot be a facet.
    There are two possible cases: Either $x_{uv} = x_{\vec{u}(v),v} = 1$, then $x_{uw} = x_{\vec{u}(v),w} = 1$ as well, or $x_{uv} = x_{\vec{u}(v),v} = 0$, then $x_{uw} = x_{vw} = x_{\vec{u}(v),w}$ by contraction of the path $P_{uv}$, so \eqref{eq:tree-cut-equality} holds.
    
    Now, suppose $v$ is a leaf of $T$ and let $\Sigma$ be the face of $\lmc(T)$ defined by \eqref{eq:tree-cut}. 
    We show $\dim \Sigma = m-1$ as in the proof of \Cref{thm:dimension} which proved $\dim\lmc(T) = m$ where $m = \abs{E \cup F} = \frac{n(n+1)}{2}$. 
    Of all the lifted multicuts considered in the proof of \Cref{thm:dimension} only one does not satisfy \eqref{eq:tree-cut} with equality, namely the one induced by the decomposition of $T$ into the nodes on the path $P_{\vec{u}(v),v}$ and otherwise singular nodes.
    Therefore, it holds that $\dim \Sigma \geq \dim \lmc(T) - 1$ which concludes the proof.
\end{delayedproof}
    
\begin{delayedproof}{prop:tree-box-facets}
    We apply the more general characterization given by \Cref{thm:upper-box-facets} and \Cref{thm:lower-box-facets}.
    The nodes $u,v \in V$ are a pair of $ww'$-cut-nodes for some nodes $w,w' \in V$ (with at least one being different from $u$ and $v$) if and only if $u$ or $v$ is not a leaf of $V$.
    Thus, the claim follows from \Cref{thm:upper-box-facets}.
    The second assertion follows from \Cref{thm:lower-box-facets} and the fact that we lift to the complete graph on~$V$.
\end{delayedproof}

\begin{delayedproof}{lem:tree-intersection-valid}
    Let $x \in \lmc(T) \cap \Z^m$. In case $x_{u,\vec{v}(u)} = x_{\vec{u}(v),v} = 1$ the inequality is trivially satisfied. 
    Now, suppose that either $x_{u,\vec{v}(u)} = 0$ or $x_{\vec{u}(v),v} = 0$ for some $u,v \in V$ with $d(u,v) \geq 3$.
    Then, since $x$ satisfies all cut inequalities with respect to $u,\vec{v}(u)$, and $\vec{u}(v),v$, and the path inequality with respect to $\vec{u}(v), \vec{v}(u)$, it must hold that $x_{\vec{u}(v),\vec{v}(u)} = 0$.
    Moreover, if even $x_{u,\vec{v}(u)} = 0 = x_{\vec{u}(v),v}$, then, by the same reasoning, we have $x_{uv} = 0$ as well.
    Hence, $x$ satisfies~\eqref{eq:intersection-inequality}.
\end{delayedproof}

\begin{delayedproof}{thm:tree-intersection-facet}
    Let $\Sigma$ be the face of $\lmc(T)$ defined by \eqref{eq:intersection-inequality} for some $u,v \in V$ with $d(u,v) \geq 3$.
    As in the proof of \Cref{prop:tree-cut-facet} we obtain $\dim \Sigma \geq \dim \lmc(T) - 1$ by observing that all but one lifted multicut considered in the proof of \Cref{thm:dimension} satisfy \eqref{eq:intersection-inequality} with equality.
    Indeed, only the lifted multicut induced by the decomposition of $T$ into the nodes on the path $P_{\vec{u}(v),\vec{v}(u)}$ and otherwise singular nodes does not satisfy \eqref{eq:intersection-inequality} with equality.
\end{delayedproof}

\begin{delayedproof}{thm:intersection-arbitrary-g}
    We first prove the ``only if'' part. To that end, suppose that $\{u,v'\}$ is not a $vu'$-separating node set (the case that $\{v,u'\}$ is not a $uv'$-separating node set is analogous). 
    Then, there exists a $vu'$-path $P=(V_P,E_P)$ in $G$ with $u,v' \notin V_P$. 
    Let $\Pi := \{V_P\} \cup \{\{w\} \mid w \in V \setminus V_P\}$ be the decomposition of $G$ into the component $V_P$ and otherwise singular nodes. 
    Let $x := \1_{\phi_{K_n}(\Pi)}$ be the characteristic vector of the multicut of $K_n$ that is induced by $\Pi$. 
    Clearly, it holds that $x_{vu'}=0$ and $x_{uv} = x_{v'u} = x_{u'v'} = 1$ and, hence, \eqref{eq:intersection-arbitrary-g} is not satisfied. Therefore, \eqref{eq:intersection-arbitrary-g} is not valid and in particular not facet-defining for $\lmc(G)$.

    Next, we turn to the ``if'' part. 
    To that end, assume that $\{u,v'\}$ is a $vu'$-separating node set and $\{v,u'\}$ is a $uv'$-separating node set. 
    If it holds that $x_{uv'} = 0$, then there exists a $uv'$ path $P=(V_P,E_P)$ with $x_e = 0$ for all $e \in E_P$. 
    Since $\{v,u'\}$ is a $uv'$-separating node set it holds that $v \in V_P$ or $u' \in V_P$. 
    In the first case we have $x_{uv} = 0$, while in the second case we have $x_{u'v'} = 0$.
    If it holds that $x_{vu'} = 0$, then an analogous argument also yields $x_{uv} = 0$ or $x_{u'v'} = 0$.
    If, further, we have $x_{uv'} = x_{vu'} = 0$, then there exist a $uv'$-path $P=(V_P,E_P)$ and a $vu'$ path $P' = (V_{P'},E_{P'})$ with $x_e = 0$ for all $e \in E_P \cup E_{P'}$.
    By assumption, the paths $P$ and $P'$ must intersect and, hence, all nodes $u,u',v,v'$ are in the same component with respect to $x$, i.e. $x_{uv} = x_{u'v'} = 0$. Therefore, \eqref{eq:intersection-arbitrary-g} is valid for $\lmc(G)$.

    It remains to show that \eqref{eq:intersection-arbitrary-g} is facet-defining for $\lmc(G)$. 
    We show this by applying \Cref{lem:valid-and-facet-defining-for-subset} together with \Cref{thm:tree-intersection-facet}. 
    Because $G$ is connected there exists a path in $G$ that connects a node in $\{u,u'\}$ to a node in $\{v,v'\}$. 
    By assumption, every $uv'$-path in $G$ contains a $uv$-subpath or a $u'v'$-subpath and so does every $vu'$-path. 
    Therefore there must exist a $uv$-path or a $u'v'$-path in $G$. 
    Without loss of generality we may assume there exists a $u'v'$-path in $G$. 
    By adding the nodes $u$ and $v$ and the edges $uu'$ and $vv'$ to that path, we obtain a $uv$-path. 
    We expand this path to a spanning tree $T$ of $G$. 
    Then, it holds that $\vec{u}(v) = u'$ and $\vec{v}(u) = v'$ and the inequality \eqref{eq:intersection-arbitrary-g} is precisely the intersection inequality \eqref{eq:intersection-inequality} for trees with respect to $u$ and $v$. 
    By \Cref{thm:tree-intersection-facet} it is facet-defining for $\lmc(T)$ and \Cref{lem:valid-and-facet-defining-for-subset} yields that it is also facet-defining for $\lmc(G)$.
    For an illustration of this construction, see \Cref{fig:intersection-arbitrary-g}.
\end{delayedproof}

\begin{delayedproof}{lem:path-polytope}
    It holds that $\lmc(T) \subseteq \ppp(n)$ as all inequalities \eqref{eq:path-box} -- \eqref{eq:path-intersection} are valid for $\lmc(T)$ by \Cref{sec:lmp-trees-facets}.
    
    Next, we prove that $\ppp(n) \subseteq \tppt(T)$.
    To this end, let $x \in \ppp(n)$.
    We show that $x$ satisfies all inequalities \eqref{eq:tree-cut}.
    Let $u,v \in V$ with $u < v - 1$.
    We need to prove that both $x_{u+1,v} \leq x_{uv}$ and $x_{u,v-1} \leq x_{uv}$ hold.
    For reasons of symmetry, it suffices to show only $x_{u+1,v} \leq x_{uv}$.
    We proceed by induction on the distance of $u$ from $n$. 
    If $v = n$, then $x_{u+1,n} \leq x_{un}$ is given by~\eqref{eq:path-cut-right}.
    Otherwise, we use \eqref{eq:path-intersection} for $j=u$ and $k=v+1$ and the induction hypothesis on $v+1$:
    \begin{align*}
        x_{uv} + x_{u+1,v+1} & \geq x_{u+1,v} + x_{u,v+1} \\
        & \geq x_{u+1,v} + x_{u+1,v+1} \\
        \implies x_{uv} & \geq x_{u+1,v} \enspace.
    \end{align*}  
    It remains to show that $x$ satisfies all inequalities \eqref{eq:tree-path}.
    Let $u,v \in V$ with $u < v - 1$.
    We proceed by induction on $d(u,v) = u-v$.
    If $d(u,v) = 2$, then \eqref{eq:tree-path} is given by \eqref{eq:path-triangle}.
    If $d(u,v) > 2$, then we use \eqref{eq:path-intersection} for $j=u$ and $k=v$ as well as the induction hypothesis on $u,v-1$, which have distance $d(u,v) - 1$:
    \begin{align*}
        x_{uv} + x_{u+1,v-1} & \leq x_{u+1,v} + x_{u,v-1} \\
        &  \leq x_{u+1,v} + x_{u,u+1} + x_{u+1,v-1} \\
        \implies  x_{uv} & \leq x_{u,u+1} + x_{u+1,v} \enspace.
    \end{align*}
    It remains to show that $x$ satisfies the box inequalities $0 \leq x_{uv} \leq 1$ for all $u,v \in \{0,\dots,n\}$ with $u < v$.
    As $x$ satisfies \eqref{eq:tree-path} and \eqref{eq:tree-cut} it holds that
    \begin{alignat}{3}
        x_{uv} &\leq x_{u,u+1} + x_{u+1,v} &\leq x_{u,u+1} + x_{u,v} 
        \quad &\Rightarrow \quad 0 \leq x_{u,u+1} \label{eq:geq-0-dist-1a} \\
        x_{uv} &\leq x_{u,v-1} + x_{v-1,v} &\leq x_{u,v} + x_{v-1,v} 
        \quad &\Rightarrow \quad 0 \leq x_{v,v-1} \label{eq:geq-0-dist-1b}
    \end{alignat}
    for $u,v \in \{0,\dots,n\}$ with $u < v-1$. 
    Now, applying \eqref{eq:tree-cut} recursively, together with the bases cases \eqref{eq:path-box}, \eqref{eq:geq-0-dist-1a}, and \eqref{eq:geq-0-dist-1b}, we obtain $0 \leq x_{uv} \leq 1$ for all $u,v \in \{0,\dots,n\}$ with $u < v$.

    Altogether, we have shown $x \in \tppt(T)$, which concludes the proof.
\end{delayedproof}

\begin{delayedproof}{thm:path-tdi}
    Let the system defined by \eqref{eq:path-box} -- \eqref{eq:path-intersection} be represented in matrix form as $Ax \leq \alpha$.
    Note that $\ppp(n)$ is non-empty and bounded.
    Thus, to establish total dual integrality, we need to show that for any $\theta \in \Z^{m}$, where $m=\frac{n(n+1)}{2}$ is the number of $x$ variables, the dual program of $\min\{\theta^\top x \mid Ax \leq \alpha\}$ has an integral optimal solution. 
    In the following, we assume $n \geq 3$.
    The case $n=2$ can be verified with a simple calculation as the system \eqref{eq:path-box} -- \eqref{eq:path-intersection} becomes $x_{02} \leq 1$, $x_{12} \leq x_{02}$, $x_{01} \leq x_{02}$, $x_{02} \leq x_{01} + x_{02}$.

    We introduce the following dual variables: $a$ for \eqref{eq:path-box}, $b_i$, $c_i$, $d_i$ for $i \in \{1,\dots,n-1\}$ for \eqref{eq:path-cut-right}, \eqref{eq:path-cut-left} and \eqref{eq:path-triangle} respectively and $e_{i,j}$ for $i,j \in \{0,\dots,n\}$ with $i < j - 2$ for \eqref{eq:path-intersection}. With this we obtain the following dual program
    \begin{align}
        \max && a \notag \\
        \text{subject to}
            && a,b,c,d,e & \leq 0 \\
            && a - b_1 - c_{n-1} + e_{0n} &= \theta_{0n} 
                \label{eq:path-dual-0} \\
            && b_{n-1} - d_{n-1} &= \theta_{n-1,n} 
                \label{eq:path-dual-1} \\
            && b_{n-2} - b_{n-1} + d_{n-1} - e_{n-3,n} &= \theta_{n-2,n} 
                \label{eq:path-dual-2} \\
            && b_i - b_{i+1} + e_{in} - e_{i-1,n} &= \theta_{in} && \forall i \in \{1,\dots,n-3\} 
                \label{eq:path-dual-3} \\
            && c_1 - d_1 &= \theta_{01} 
                \label{eq:path-dual-4} \\
            && c_2 - c_1 + d_1 - e_{03} &= \theta_{02} 
                \label{eq:path-dual-5} \\
            && c_i - c_{i-1} + e_{0i} - e_{0,i+1} &= \theta_{0i} && \forall i \in \{3,\dots,n-1\} 
                \label{eq:path-dual-6} \\
            && -d_i - d_{i+1} + e_{i-1,i+2} &= \theta_{i,i+1} && \forall i \in \{1,\dots,n-2\} 
                \label{eq:path-dual-7} \\
            && d_{i+1} + e_{i-1,i+3} - e_{i-1,i+2} - e_{i,i+3} &= \theta_{i,i+2} && \forall i \in \{1,\dots,n-3\} 
                \label{eq:path-dual-8} \\
            && e_{ij} + e_{i-1,j+1} - e_{i-1,j} - e_{i,j+1}  &= \theta_{ij} && \forall i \in \{1,\dots,n-4\}, 
                \label{eq:path-dual-9} \\
            &&&&& \forall j \in \{i+3,\dots,n-1\}. \notag
    \end{align}
    Observe that \eqref{eq:path-dual-7}, \eqref{eq:path-dual-8} and \eqref{eq:path-dual-9} include all the $e$ variables. In particular, \eqref{eq:path-dual-9} only includes $e$ variables with indices of distance $3$, \eqref{eq:path-dual-8} couples $e$ variables of distance $4$ with those of distance $3$, and finally \eqref{eq:path-dual-9} couples the remaining $e$ variables of distance $k>4$ with those of distance $k-1$ and $k-2$. Therefore, we can express all $e$ variables in terms of the $d$ variables and $\theta$:
    \begin{align}\label{eq:formual-e-variables}
        0 \geq e_{ij} = \sum_{i < k < \ell < j} \theta_{k\ell} + \sum_{i < k < j} d_k 
        \qquad \forall i \in \{0,\dots,n-3\}, j \in \{i+3,\dots,n\} \enspace .
    \end{align}
    For $i\in \{1,\dots,n-2\}$ we can express the variable $b_i$ in terms of the $d$ variables, $\theta$ and $b_{i+1}$ by equations \eqref{eq:path-dual-2}, \eqref{eq:path-dual-3} and \eqref{eq:formual-e-variables}. Together with \eqref{eq:path-dual-1} we obtain 
    \begin{align}\label{eq:formula-b-variables}
        0 \geq b_i = \sum_{i \leq k < \ell \leq n} \theta_{k\ell} + \sum_{i \leq k \leq n-1} d_k
        \qquad \forall i \in \{1,\dots,n-1\} \enspace . 
    \end{align}
    Similarly, equations \eqref{eq:path-dual-4}, \eqref{eq:path-dual-5}, \eqref{eq:path-dual-6} and \eqref{eq:formual-e-variables} yield
    \begin{align}\label{eq:formula-c-variables}
        0 \geq c_i = \sum_{0 \leq k < \ell \leq i} \theta_{k\ell} + \sum_{1 \leq k \leq i} d_k
        \qquad \forall i \in \{1,\dots,n-1\} \enspace . 
    \end{align}
    Lastly, equation \eqref{eq:path-dual-0}, together with \eqref{eq:formual-e-variables}, \eqref{eq:formula-b-variables} and \eqref{eq:formula-c-variables} yields
    \begin{align}
        0 \geq a = \sum_{0 \leq k < \ell \leq n} \theta_{kl} + \sum_{1 \leq k \leq n-1} d_k \enspace .
    \end{align}
    Altogether we can rewrite the dual program as
    \begin{align}\label{eq:path-dual-simplified}
        \max 
            && \sum_{0 \leq k < \ell \leq n} \theta_{kl} &+ \sum_{1 \leq k \leq n-1} d_k \\
        \text{subject to}
            && \sum_{i \leq k \leq j} d_k & \leq - \sum_{i \leq k < \ell \leq j} \theta_{k\ell} && \forall i, j \in \{1,\dots,n-1\}, i < j \notag\\
            && \sum_{i \leq k \leq n-1} d_k & \leq - \sum_{i \leq k < \ell \leq n} \theta_{k\ell} && \forall i \in \{1,\dots,n-1\} \notag\\
            && \sum_{1 \leq k \leq i} d_k & \leq - \sum_{0 \leq k < \ell \leq i} \theta_{k\ell} && \forall i \in \{1,\dots,n-1\} \notag\\
            && \sum_{1 \leq k \leq n-1} d_k & \leq - \sum_{0 \leq k < \ell \leq n} \theta_{kl} \notag \\
            && d & \leq 0 \enspace, \notag
    \end{align}
    where the first inequality is obtained from \eqref{eq:formual-e-variables} by shifting the $i$ and $j$ index by $+1$ and $-1$ respectively. The matrix corresponding to the inequality constraints satisfies the \emph{consecutive-ones} property with respect to its rows. 
    Therefore, the constraint matrix of the system is totally unimodular and, hence, \eqref{eq:path-dual-simplified} admits an integral optimal solution.
\end{delayedproof}
    
\begin{delayedproof}{prop:sequential-set-partition-problem}
    For $i,j \in \{0,\dots,n\}$ with $i < j$ we introduce a new variable $\lambda_{ij} := x_{i-1,j} + x_{i,j+1} - x_{ij} - x_{i-1,j+1}$ where we set $x_{-1,k} = x_{k,n+1} = 1$ for $k \in \{0,\dots,n\}$ and also $x_{-1,n+1}=1$.
    By rearranging the definition of $\lambda$ according to $x_{ij}$, we obtain $x_{ij} = x_{i-1,j} + x_{i,j+1} - x_{i-1,j+1} - \lambda_{ij}$. 
    Applying this formula recursively yields $x_{ij} = 1 - \sum_{k \leq i, j \leq \ell} \lambda_{k\ell}$ for all $i,j \in \{0,\dots,n\}$ with $i < j$.
    By substituting the $x$ variables in inequalities \eqref{eq:path-box} -- \eqref{eq:path-intersection} we obtain $\lambda_{0n} \geq 0$ from \eqref{eq:path-box}, $\lambda_{in} \geq 0$ for $i \in \{1,\dots,n-1\}$ from \eqref{eq:path-cut-right}, $\lambda_{0i} \geq 0$ for $i \in \{1,\dots,n-1\}$ from \eqref{eq:path-cut-left}, $\lambda_{ij} \geq 0$ for $i,j \in \{1,\dots,n-1\}$, $i < j$ from \eqref{eq:path-intersection} and $\sum_{i \leq k \leq j, i \neq j} \lambda_{ij} \leq 1$ for $k \in \{1,\dots,n-1\}$ from \eqref{eq:path-triangle}.
    By further substituting $x$ in the objective function we obtain 
    \[
        \theta^\top x 
        = \sum_{0 \leq i < j \leq n} \theta_{ij} \left(1 - \sum_{k \leq i, j \leq \ell} \lambda_{k\ell}\right) 
        = \sum_{0 \leq i < j \leq n} \theta_{ij} - \sum_{0 \leq k < \ell \leq n} \lambda_{k\ell} \sum_{k \leq i < j \leq \ell} \theta_{ij} \enspace .
    \]
    Altogether, and with the definition of $\Theta$, we obtain that $\min\{\theta^\top x \mid x \in \ppp(n)\}$ is equivalent to 
    \begin{align}
        \min 
            && \Theta_{0n} - \Theta^\top \lambda \label{eq:ssp-variant} \\
        \text{subject to} 
            && \sum_{0 \leq i \leq k \leq j \leq n, i \neq j} \lambda_{ij} & \leq 1 && \forall k \in \{1,\dots,n-1\} \nonumber \\
            && \lambda_{ij} & \geq 0 && \forall i,j \in \{0,\dots,n\}, i \leq j \enspace . \nonumber 
    \end{align}
    Note that this is precisely the dual of \eqref{eq:path-dual-simplified}. 
    By defining the additional variables $\lambda_{ii} := 1 - \sum_{0 \leq i \leq k \leq j \leq n, i \neq j} \lambda_{ij}$ for $i \in \{0,\dots,n\}$ we obtain \eqref{eq:ssp}.
\end{delayedproof}

\begin{delayedproof}{prop:general-intersection-valid}
	We observe that the structure of this proof is similar to the proof of validity for the running intersection inequalities in \citet{dPKha21}.
    Let us partition $K$ in $K_1$, $K_2$, $\dots$, $K_p$ such that $E_{uv} \cap \bigcup_{k \in K_i} E_{u_kv_k}$ form a component for $i = 1$, $\dots$, $p$.
    Then, we want to show that $\sum_{k \in K_i : N_k \neq \emptyset} x_{f_k} \leq \sum_{k \in K_i} x_{u_k v_k}$ for all $i= 1$, $\dots$, $p$, for every feasible vector $x$.
    Equivalently, we can prove that $\sum_{k \in K_i : N_k \neq \emptyset} x_{f_k} - \sum_{k \in K_i} x_{u_k v_k} \leq 0$ for every $i= 1$, $\dots$, $p$.
    We are going to do that by showing 
    \begin{equation}\label{eq:lemma-gen}
        \max_{x \in \lmc(T,\widehat{G})} 
            \sum_{\substack{k \in K_i : \\N_k \neq \emptyset}} x_{f_k} 
            - \sum_{k \in K_i} x_{u_k v_k} = 0 \qquad \text{ for all } i = 1, \dots, p \enspace .
    \end{equation}
    Consider an arbitrary $i \in \{1,\dots,p\}$.
    We divide the proof in two cases: whether there exists $k \in K_i$ such that $x_{f_k} = 1$, or if $x_{f_k} = 0$ for all $k \in K_i$.
    Let us start from the first case.
    Hence, there exists at least one index $k \in K_i$ such that $x_{f_k} = 1$.
    Consider the first index in $K_i$ for which this happens, let it be $k'$.
    The corresponding path inequality \eqref{eq:path-ineq-tree} together with $x_{f_{k'}} = 1$ implies that there exists an edge $e \in E_{f_{k'}} \subseteq E_{u_{k'} v_{k'}}$ such that $x_e = 1$. In turn, the corresponding cut inequality \eqref{eq:cut-ineq-tree}, implies that $x_{u_{k'} v_{k'}} = 1$.
    However, $e \in E_{f_{k'}} \subseteq N_{k'} = E_{uv} \cap E_{u_{k'} v_{k'}} \cap \bigcup_{0<j<k'} E_{u_j v_j}$, hence there exists $j' < k'$ for which $e \in E_{u_{j'} v_{j'}}$.
    Thus, $x_{u_{j'} v_{j'}} = 1$ as well.
    If $k'$ is the only index for which $x_{f_{k'}} = 1$, then $\sum_{{k \in K_i : N_k \neq \emptyset}} x_{f_k} - \sum_{k \in K_i} x_{u_k v_k} \leq -1$.
    Assume that there exists a second index $k'' > k'$ such that $x_{f_{k''}} = 1$.
    Similarly to before, $x_{u_{k''} v_{k''}} = x_{u_{j''} v_{j''}} = 1$, for some $0 < j'' < k''$.
    Note that $j''$ could coincide with $j'$ or $k'$.
    In this case though, the terms $x_{f_{j''}}$ and $x_{u_{j''} v_{j''}}$ would simply cancel out without affecting the value of $\sum_{k \in K_i : N_k \neq \emptyset} x_{f_k} - \sum_{k \in K_i} x_{u_k v_k}$.
    The above argument can be applied recursively until there are no indices $k \in K_i$ left for which $x_{f_k} = 1$. 
    Hence, $\sum_{k \in K_i : N_k \neq \emptyset} x_{f_k} - \sum_{k \in K_i} x_{u_k v_k} \leq -1$ in this case.
    
    Let us move on to the second case, i.e. where $x_{f_k} = 0$ for all $k \in K_i$.
    In that case, showing that \eqref{eq:lemma-gen} holds reduces to proving that $\min_{x \in \lmc(T,\widehat{G})} \sum_{k \in K_i} x_{u_k v_k} = 0$.
    Recall that $\lmc(T,\widehat{G})$ is a binary polytope, which tells us that $\sum_{k \in K_i} x_{u_k v_k} \geq 0$ for all feasible points.
    Moreover, $\sum_{k \in K_i} x_{u_k v_k} = 0$ if and only if $x_{u_k v_k} = 0$ for all $k \in K_i$.
    By putting the two above parts together, we see that \eqref{eq:lemma-gen} is true.
    
    By putting these inequalities together for the distinct components induced by $K_1$, $\dots$, $K_p$, we get that $\sum_{k \in K : N_k \neq \emptyset} x_{f_k} \leq \sum_{k \in K} x_{u_k v_k}$ is a valid inequality for $\lmc(T,\widehat{G})$.
    Then, we add to the above inequality the bounds $x_e \geq 0$ for all $e \in E_{uv} \setminus \bigcup_{k \in K} E_{u_k v_k}$. 
    We are only missing the term $x_{uv}$ in the left-hand side of \eqref{eq:gen-ineq}.
    Note that if $x_{uv} = 0$, then $x_e = 0$ for every $e \in E_{uv}$, which implies that $\sum_{k \in K : N_k \neq \emptyset} x_{f_k} = \sum_{e \in E_{uv} \setminus \bigcup_{k \in K} E_{u_k v_k}} x_e = 0$.
    Therefore, \eqref{eq:gen-ineq} is valid for $\lmc(T,\widehat{G})$. 
    On the other hand, if $x_{uv} = 1$, it follows that there exists $e \in E_{uv}$ such that $x_e = 1$, and at least one of the sums on the right-hand side of \eqref{eq:lemma-gen} is non-zero. 
    If also the sum on the left-hand side of \eqref{eq:lemma-gen} is positive, then the discussion done earlier to show \eqref{eq:lemma-gen} implies that \eqref{eq:gen-ineq} holds.
\end{delayedproof}

\begin{delayedproof}{prop:nec-cond-gen-ineq}
    Let $a^\top x \leq b$ be a generalized intersection inequality corresponding to $uv \in E \cup F$, and $\{u_k, v_k\}$ for $k \in K$.
    
    We first prove that if $a^\top x \leq b$ is facet-defining, then the first property must hold.
    For the sake of contradiction let us assume that Condition~\ref{cond:generalized-intersection-1} is violated. We show that $a^\top x \leq b$ can be obtained as a sum of two different generalized intersection inequalities.
    This implies that $a^\top x \leq b$ is redundant for $\lmc(T,\widehat{G})$.
    The violation of Condition~\ref{cond:generalized-intersection-1} implies that there exists $\bar k \in K$ such that $N_{\bar k} \neq \emptyset$ and $E_{f_{\bar k}}$ is not maximal in $N_{\bar k}$.
    This means that there exists an edge $f'_{\bar k} \in F$ such that $E_{f'_{\bar{k}}} \subseteq N_{\bar{k}}$ and $E_{f_{\bar k}} \subset E_{f'_{\bar k}}$.
    Then, let ${a'}^\top x \leq b$ be the generalized intersection inequality obtained by replacing $f_{\bar k}$ with $f'_{\bar k}$.
    Observe that the inequality $x_{f_{\bar k}} \leq x_{f'_{\bar k}}$ is valid for $\lmc(T,\widehat{G})$, and in particular is a generalized intersection inequality as well.
    It is obtained by choosing $uv = f_{\bar k}$, $|K| = 1$ and $u_1v_1 = f'_{\bar k}$.
    When we sum ${a'}^\top x \leq b$ with $x_{f_{\bar k}} \leq x_{f'_{\bar k}}$ we get precisely $a^\top x \leq b$.
    Therefore, $a^\top x \leq b$ is not facet-defining.
    
    Next, we prove that if $a^\top x \leq b$ is facet-defining, then also the second condition must hold.
    Once again, for the sake of contradiction, assume that there exist distinct indices $i,j \in K$ such that $E_{f_i}, E_{f_j} \subseteq N_i \cap N_j$ and $f_i \neq f_j$.
    Consider the generalized intersection inequality ${a'}^\top x \leq b$ obtained by using the same edges $uv$, $\{u_k, v_k\}$ for $k \in K$ and the same $f'_k$, for $k \neq i,j$. 
    For these indices we choose $f'_i = f'_j = f_i$.
    We construct in an analogous manner the generalized intersection inequality ${a''}^\top x \leq b$, where instead we set $f''_i = f''_k = f_j$. 
    It is easy to see that by summing the inequalities ${a'}^\top x \leq b$ and ${a''}^\top  x \leq b$ we obtain $2 a^\top x \leq 2b$, which is equivalent to $a^\top x \leq b$.
    Hence, Condition~\ref{cond:generalized-intersection-2} must hold. 
\end{delayedproof}
    
\begin{delayedproof}{prop:path-partition-beta-acyclic}
    Let $T = (V,E)$, where $V = \{0,1,\dots,n\}$ and $E = \{ \{i,i+1\} \mid i=0, \dots, n-1 \}$.
    Consider the path partition problem by taking into account also the additional paths between nodes in the pairs of $F$.
    Then, we define a hypergraph $H$ by following the construction explained in the first part of \Cref{sec:multilinear}.
    Note that the node $\bar v_1$ in the hypergraph corresponding to the edge $\{0,1\}$ in $T$ is a nest point since all the paths containing $\{0,1\}$ form a chain for set inclusion.
    If it were not so, then there would exist two edges $\bar e$, $\bar f$ in $H$ such that $\bar v_1 \in \bar e$, $\bar v_1 \in \bar f$, $\bar e \not\subseteq \bar f$, $\bar f \not\subseteq \bar e$.
    This means that there exists two nodes in $H$, $\bar v'$ and $\bar v''$, such that $\bar v' \in \bar e \setminus \bar f$ and $\bar v'' \in \bar f \setminus \bar e$.
    These two nodes in $H$ correspond to two edges in $T$, let us denote them by $\{i, i+1\}$ and $\{j, j+1\}$, with $i \neq j$.
    When we translate it to the lifted multicut setting, it implies that there exists two paths $P_{0k}$, $P_{0l}$ in $T$ starting from the node $0$ and ending in two different nodes $k$, $l$ such that the edge $\{i, i+1\}$ is on the path $P_{0k}$ but not on the path $P_{0l}$, and the edge $\{j, j+1\}$ is on the path $P_{0l}$ but not on the path $P_{0k}$.
    This contradicts the assumption that $T$ is a path starting at node $0$. 
    
    Next, we remove the node $\bar v_1$, which is a nest point, from $H$.
    Once it has been removed, we similarly remark that the node $\bar v_2$ representing $\{1,2\}$ becomes a nest point for the hypergraph that was obtained by removing $\bar v_1$ from $H$.
    Hence, we remove $\bar v_2$ from this new hypergraph.
    This argument can be repeated recursively until we have removed all the nodes, thus obtaining the empty hypergraph.
    Therefore, by using \Cref{thm:charact-beta}, we can conclude that the original hypergraph $H$ was indeed $\beta$-acyclic.
\end{delayedproof}
    

\begin{delayedproof}{prop:cut-not-facet}
    We show that for a given $f = vw \in F$ and a $vw$-cut $\{e,e'\}$ where $e$ and $e'$ do not share a node Condition~\ref{cond:cut-4} of \Cref{thm:cut-facets} is violated. 
    For an example, see \Cref{figure:violated-conditions}i.

    Let $s,t \in \Z_n$ such that $e=\{s,s+1\}$ and $e'=\{t,t+1\}$. 
    Without loss of generality we may assume that $w \in [s+1,t]$ and $v \in [t+1,s]$ by potentially interchanging $e$ and $e'$. 
    Since $e$ and $e'$ do not share a node we have either $v \neq s$ and $w \neq t$ or we have $v \neq t+1$ and $w \neq t+1$. 
    In the first case consider the path along the nodes $\{v, t, s, w\}$ in the second case consider the path along the nodes $\{v, s+1, t+1, w\}$. 
    In either case all $(vw,U)$-connected components satisfy \eqref{eq:proof-facet-aux2} or \eqref{eq:proof-facet-aux3} and hence Condition~\ref{cond:cut-4} is violated.
\end{delayedproof}

\begin{delayedproof}{prop:clique-web-not-facet}
    In case $p=2$, $q=1$, $r=0$, $S=\{u\}$ and $T=\{u-1,u+1\}$ for some $u \in \Z_n$ we have $W=\left\{\{u-1,u+1\}\right\}$ and the clique-web inequality \eqref{eq:clique-web-inequality} is the inequality
    \begin{align}\label{eq:path-length-2}
        x_{u-1,u+1} \leq x_{u-1,u} + x_{u,u+1} 
    \end{align}
    which is the path inequality \eqref{eq:path} corresponding to the $u-1,u+1$-path with edges $\{\{u-1,u\},\{u,u+1\}\}$. 
    This inequality is facet-defining by \Cref{cor:path-length-2-facet}. 
    
    Now assume the clique-web inequality \eqref{eq:clique-web-inequality} does not coincide with a path inequality of a path of length $2$. 
    We show that if $x \in X_n$ satisfies \eqref{eq:clique-web-inequality} with equality, then $x$ also satisfies \eqref{eq:path-length-2} with equality for all $u \in S$. 
    By assumption these inequalities are different from \eqref{eq:clique-web-inequality}.
    Therefore, the clique-web inequality \eqref{eq:clique-web-inequality} is not facet-defining for $\lmc(C)$. 

    Let $x \in X_n$ such that $x$ satisfies \eqref{eq:clique-web-inequality} with equality.
    For sake of contradiction, suppose \eqref{eq:path-length-2} is strict, i.e. $x_{u,u+1} = x_{u-1,u} = 1$ or $x_{u-1,u+1} = 0$ and $x_{u,u+1} + x_{u-1,u} = 1$. 
    Assume that the second case holds, i.e. $x_{u,u+1} = 1$ and $x_{u-1,u} = 0$ or $x_{u,u+1} = 0$ and $x_{u-1,u} = 1$.
    Assume without loss of generality that $x_{u,u+1} = 1$ and $x_{u-1,u} = 0$.
    By the cut inequalities \eqref{eq:cut}, it holds that $1 - x_{u-1,u+1} \leq (1 - x_{u,u+1}) + (1 - x_e) \Rightarrow x_e = 0$ for all $e \in E \setminus \{\{u,u+1\}, \{u-1,u\}\}$. 
    It follows that $1 = x_{u,u+1} > \sum_{e \in E \setminus \{\{u,u+1\}\}} x_e = 0$, i.e. the cycle inequality \eqref{eq:cycle} corresponding to $\{u,u+1\}$ is violated, which is a contradiction to $x \in X_n$. 
    Therefore, \eqref{eq:path-length-2} can only by strict if $x_{u,u+1} = x_{u-1,u} = 1$.
    
    For all $w \in \Z_n \setminus \{u\}$ the set $\{\{u,u+1\}, \{u-1,u\}\}$ is a $uw$-cut in $C$ and the corresponding cut inequality \eqref{eq:cut}, together with $x_{u,u+1} = x_{u-1,u} = 1$, yields $x_{uw} = 1$. 
    Therefore, it holds that $\sum_{w \in T} x_{uw} = |T| = p$.
    Further, due to $x_e \leq 1$ for all $e \in \binom{\Z_n}{2}$ and $|W| = \frac{p(p-1)}{2} - pr$ it holds that $\sum_{e \in W} x_e \leq \frac{p(p-1)}{2} - pr$.
    In case $q=1$, i.e. $S=\{u\}$ for some $u \in \Z_n$, \eqref{eq:clique-web-inequality} can be written as
    \begin{align*}
        \sum_{e \in W} x_{e} - \sum_{w \in T} x_{uw} 
        & \leq \frac{p(p-1)}{2} - pr - p = \frac{p(p-2r-3)}{2} \\ 
        & < \frac{(p-1)(p-2r-2)}{2} = \frac{(p-q)(p-q-2r-1)}{2} \enspace,
    \end{align*}
    contradicting that $x$ satisfies \eqref{eq:clique-web-inequality} with equality. 
    In case $q \geq 2$, let $\mathring{S} := S \setminus \{u\}$ for some $u \in S$. 
    We get 
    \begin{align*}
        & \sum_{e \in W} x_e 
        + \sum_{v,w \in S, v \neq w} x_{vw}
        - \sum_{v \in S, w \in T} x_{vw} \\
        = &
        \underbrace{
            \sum_{e \in W} x_e 
            + \sum_{v,w \in \mathring{S}, v \neq w} x_{vw}
            - \sum_{v \in \mathring{S}, w \in T} x_{vw}
        }_{\overset{(*)}{\leq} \frac{(p-q+1)(p-q-2r)}{2}}
        + \underbrace{\sum_{v \in \mathring{S}} x_{uv}}_{\leq q-1}
        - \underbrace{\sum_{w \in T} x_{uw}}_{=p} \\
        \leq &
        \frac{(p-q)(p-q-2r-1)}{2} - 1 < \frac{(p-q)(p-q-2r-1)}{2} \enspace,
    \end{align*}
    i.e. \eqref{eq:clique-web-inequality} is not satisfied with equality, in contradiction to our assumption. 
    The inequality $(*)$ holds as it is a clique-web inequality with respect to $r$, $\mathring{S}$, $T$ and $W$ which is valid by \Cref{lem:clique-web-valid}.
\end{delayedproof}

\begin{delayedproof}{thm:2-chorded-valid-and-facet}
    The proof is due to \citet{Groetschel1990}. For completeness, since we refer to this proof in our following results, we reproduce the proof of validity:

    For all $i \in \Z_k$ the triangle inequality $x_{v_{i-1}v_{i+1}} - x_{v_{i-1}v_i} - x_{v_iv_{i+1}} \leq 0$ and the box inequality $x_{v_{i-1}v_{i+1}} \leq 1$ hold (a triangle inequality is a cycle inequality for a cycle of length three). 
    Summing all those inequalities we obtain 
    \[
        2 \sum_{i \in \Z_k} \left(x_{v_iv_{i+2}} - x_{v_iv_{i+1}}\right) 
        \leq k \enspace .
    \]
    Dividing by $2$ and rounding the right hand side down to the nearest integer yields \eqref{eq:2-chorded-cycle}.
\end{delayedproof}

\begin{delayedproof}{prop:2-chorded-not-facet}
    First, assume that \ref{item:2-chorded-not-true-not-facet} is not satisfied.
    Then there exist distinct nodes $u,w \in \{v_0,\dots,v_{k-1}\}$ such that $uw \neq v_iv_{i+1}$ for all $i \in \Z_k$ and such that $]u,w[ \; \cap \{v_0,\dots,v_{k-1}\} = \emptyset$. 
    Let $x \in X_n$ be the characteristic vector of a lifted multicut that satisfies \eqref{eq:2-chorded-cycle} with equality. 
    We show that $x$ also satisfies the equality 
    \begin{align}\label{eq:2-chorded-intersection}
        x_{uw} + x_{u-1,w+1} = x_{u,w+1} + x_{u-1,w}
    \end{align}
    and hence \eqref{eq:2-chorded-cycle} is not facet-defining.
    This equality corresponds to the intersection inequality \eqref{eq:cycle-intersection-inequality} with respect to $\{u,w+1\}$. 
    By \Cref{cor:intersection-cycle}, it holds that $x_{uw} + x_{u-1,w+1} \leq x_{u,w+1} + x_{u-1,w}$. 
    For sake of contradiction we assume that \eqref{eq:2-chorded-intersection} is not satisfied, i.e. above inequality is strict. 
    Then, we have either 
    \begin{enumerate}[(a)]
        \item \label{item:2-chord-case-a2}
        $x_{u,w+1} + x_{u-1,w} = 1$ and $x_{uw} + x_{u-1,w+1} = 0$ or
        \item \label{item:2-chord-case-b2}
        $x_{u,w+1} + x_{u-1,w} = 2$ and $x_{uw} + x_{u-1,w+1} \leq 1$.
    \end{enumerate}
    First, assume Case \ref{item:2-chord-case-a2} holds.
    Due to $x_{uw} = 0$ the nodes $u$ and $w$ are in the same component of the decomposition of $C$ with respect to $x$. 
    Due to $x_{u-1,w+1} = 0$ the nodes $u-1$ and $w+1$ are also in the same component.
    Due to $x_{u,w+1} + x_{u-1,w} = 1$ either the nodes $u$ and $w+1$ or the nodes $u-1$ and $w$ are in the same component with respect to $x$. 
    Then, in either case, the nodes $u,u-1,w$ and $w+1$ are all in the same component with respect to $x$ in contradiction to $x_{u,w+1} + x_{u-1,w} = 1$.
    Therefore, Case~\ref{item:2-chord-case-a2} can not occur and, thus, Case~\ref{item:2-chord-case-b2} must hold. 
    Due to $x_{uw} + x_{u-1,w+1} \leq 1$ we have that $x_{uw} = 0$ or $x_{u-1,w+1} = 0$. 
    By the cut inequalities \eqref{eq:cut} it follows that the path along the nodes $[u, w]$ or the path along the nodes $[w+1, u-1]$ is not cut with respect to $x$. 
    It follows that $x_{u-1,u} = x_{w,w+1} = 1$ holds, since otherwise the path along the nodes $[u, w+1]$ or the path along the nodes $[w, u-1]$ would be not cut with respect to $x$ contradicting $x_{u,w+1} = x_{u-1,w} = 1$. 

    Let $i_1,i_2 \in \Z_k$ such that $u=v_{i_1}$ and $w=v_{i_2}$. 
    Due to the assumption $uw \neq v_iv_{i+1}$ for all $i \in \Z_k$ and $]u,w[ \; \cap \{v_0,\dots,v_{k-1}\} = \emptyset$ the set $\{\{u-1,u\}, \{w,w+1\}\}$ is a $v_jv_{j-1}$-cut and a $v_jv_{j+1}$-cut for $j=i_1,i_2$. 
    With $x_{u-1,u} = x_{w,w+1} = 1$ from above, the cut inequalities \eqref{eq:cut} yield $x_{v_jv_{j-1}} = x_{v_jv_{j+1}} = 1$ for $j=i_1,i_2$. 
    For the corresponding triangles induced by the nodes $\{v_{j-1},v_j,v_{j+1}\}$, it holds that $x_{v_{j-1}v_{j+1}} - x_{v_{j-1}v_j} - x_{v_jv_{j+1}} \leq -1$ for $j=i_1,i_2$.
    As in the proof of \Cref{thm:2-chorded-valid-and-facet}, adding all triangle inequalities and box inequalities yields 
    \[
        2 \sum_{i=0}^{k-1} \left( x_{v_iv_{i+2}} - x_{v_iv_{i+1}} \right)
        \leq k-2 \enspace .
    \]
    Dividing by $2$ and rounding down yields 
    \[
        \sum_{i=0}^{k-1} \left( x_{v_iv_{i+2}} - x_{v_iv_{i+1}} \right)
        \leq \left\lfloor \frac{k}{2} \right\rfloor - 1 
        < \left\lfloor \frac{k}{2} \right\rfloor \enspace ,
    \]
    i.e. \eqref{eq:2-chorded-cycle} is not satisfied with equality in contradiction to our assumption. 
    Therefore every $x \in X_n$ that satisfies \eqref{eq:2-chorded-cycle} with equality also satisfies \eqref{eq:2-chorded-intersection} with equality and \eqref{eq:2-chorded-cycle} cannot be facet-defining.
    
    \medskip

    Next, assume \ref{item:2-chorded-geq-k-not-facet} is not satisfied, i.e. $k \geq 6$.
    Due to \Cref{thm:dimension}, \Cref{thm:2-chorded-valid-and-facet} and $\lmc(C) \subseteq \mc(K_n)$, the 2-chorded cycle inequality is not facet-defining for $\lmc(C)$ for even $k$. 
    From now on let $k \geq 7$ odd. 
    In that case the right hand side of \eqref{eq:2-chorded-cycle} becomes $d := \frac{k-1}{2}$.

    If $v$ is not true to $C$ the 2-chorded cycle inequality is not facet-defining for $\lmc(C)$ by \ref{item:2-chorded-not-true-not-facet}, so we may assume that $v$ is true to $C$, i.e. $v_i \in [v_0, v_{i+1}[$ for $i=0,\dots,k-2$. 
    Let $x \in X_n$ be the characteristic vector of a lifted multicut that satisfies \eqref{eq:2-chorded-cycle} with equality. 
    We show that then $x$ also satisfies the equalities $x_{v_iv_{i+d}} = 1$ for all $i \in \Z_k$, i.e. \eqref{eq:2-chorded-cycle} is not facet-defining. 
    To that end, let $i \in \Z_k$ be fixed and assume $x_{v_iv_{i+d}} = 0$. 
    By the cut inequalities \eqref{eq:cut}, the path along the nodes $[v_i,v_{i+d}]$ or the path along the nodes $[v_{i+d},v_i]$ is not cut with respect to $x$.
    By the path inequalities one of the following cases holds
    \begin{enumerate}[(a)]
        \item \label{item:2-chord-case-a}
        $x_{v_{i+j}v_{i+j+1}} = 0$ for $j=0,\dots,d-1$ and $x_{v_{i+j}v_{i+j+2}} = 0$ for $j=0,\dots,d-2$,
        \item \label{item:2-chord-case-b}
        $x_{v_{i-j}v_{i-j-1}} = 0$ for $j=0,\dots,d$ and $x_{v_{i-j}v_{i-j-2}} = 0$ for $j=0,\dots,d-1$.
    \end{enumerate}
    In Case~\ref{item:2-chord-case-a} at most $d+2$ summands of the left hand side of \eqref{eq:2-chorded-cycle} with coefficient $+1$ can have value $1$. 
    Therefore \eqref{eq:2-chorded-cycle} can only be satisfied with equality if $\sum_{j=0}^{d} x_{v_{i-j}v_{i-j-1}} \leq 2$ holds. 
    In case $\sum_{j=0}^{d} x_{v_{i-j}v_{i-j-1}} = 0$ we have $x_{v_jv_{j+1}}=x_{v_jv_{j+2}}=0$ for all $j \in \Z_k$ and clearly \eqref{eq:2-chorded-cycle} is not satisfied with equality. 
    The case $\sum_{j=0}^{d} x_{v_{i-j}v_{i-j-1}} = 1$ cannot occur since a cycle inequality with respect to the cycle along the nodes $v_0,\dots,v_{k-1}$ would be violated.
    In case $\sum_{j=0}^{d} x_{v_{i-j}v_{i-j-1}} = 2$ there are exactly two indices $i_1,i_2 \in \Z_k$ with $x_{v_{i_1}v_{i_1+1}} = x_{v_{i_2}v_{i_2+1}} = 1$ and $x_{v_\ell v_{\ell+1}} = 0$ for all $\ell  \in \Z_k \setminus \{i_1,i_2\}$. 
    It follows $x_{v_{i_1-1}v_{i_1+1}} = x_{v_{i_1}v_{i_1+2}} = x_{v_{i_2-1}v_{i_2+1}} = x_{v_{i_2}v_{i_2+2}} = 1$ and $x_{v_\ell v_{\ell + 2}} = 0$ for all $\ell \in \Z_k \setminus \{i_1-1,i_1,i_2-1,i_2\}$. 
    Therefore, the left hand side of \eqref{eq:2-chorded-cycle} evaluates to $4-2 = 2$ (or $3-2 = 1$ if $i_1$ and $i_2$ are just one apart) and since $k \geq 7$ the inequalities \eqref{eq:2-chorded-cycle} is not satisfied with equality. 

    In Case~\ref{item:2-chord-case-b} at most $d+1$ summands of the left hand side of \eqref{eq:2-chorded-cycle} with coefficient $+1$ can have value $1$ and \eqref{eq:2-chorded-cycle} can only be satisfied with equality if $\sum_{j=0}^{d-1} x_{v_{i+j}v_{i+j+1}} \leq 1$ holds. 
    Only the case ``$=0$'' can occur since otherwise a cycle inequality would be violated and, as above, it follows that \eqref{eq:2-chorded-cycle} is not satisfied with equality.

    All together, we have shown that all $x \in X_n$ that satisfy \eqref{eq:2-chorded-cycle} with equality also satisfy $x_{v_iv_{i+d}} = 0$ for $i \in \Z_k$ and therefore \eqref{eq:2-chorded-cycle} is not facet-defining for $\lmc(C)$.
\end{delayedproof}

\begin{delayedproof}{prop:num-half-chorded-odd-cycle}
    There are $k!$ permutations of $k$ nodes $v_0,\dots,v_{k-1}$. 
    Respectively $2k$ permutations represent the same cycle. 
    Therefore, in total there are 
    \[
        \sum_{k=5 \text{ odd}}^n \binom{n}{k} \frac{k!}{2k} =
        \sum_{k=5 \text{ odd}}^n \frac{n!}{(n-k)! \; 2k}
    \]
    distinct half-chorded odd cycle inequalities.
\end{delayedproof}

\begin{delayedproof}{lem:half-chorded-valid}
    We show this by constructing the inequality by summing other valid inequalities and then adjusting the right hand side downward to the next integer, similarly to the proof of \Cref{thm:2-chorded-valid-and-facet}.

    Let $5 \leq k \leq n$ with $k$ odd, let $d=\frac{k-1}{2}$, and let $v\colon\Z_k \to \Z_n$ injective.
    For $i \in \Z_k$ consider the valid cycle inequality
    \begin{align}\label{eq:cycle-ineq-for-half-chorded}
        x_{v_iv_{i+d}} - \sum_{\ell=0}^{d-1} x_{v_{i+\ell}v_{i+\ell+1}} \leq 0
    \end{align}
    for the cycle along the nodes $v_i, v_{i+1},\dots,v_{i+d}$. 
    Summing inequalities \eqref{eq:cycle-ineq-for-half-chorded} and $d-1$ times the box inequalities $x_{v_iv_{i+d}} \leq 1$ for $i \in \Z_k$ we obtain
    \[
        d \sum_{i \in \Z_k} \left( x_{v_iv_{i+d}} - x_{v_iv_{i+1}} \right) 
        \leq k (d-1) \enspace.
    \]
    The claim follows by dividing both sides by $d=\tfrac{k-1}{2}$ and rounding down:
    \[
        \sum_{i \in \Z_k} \left( x_{v_iv_{i+d}} - x_{v_iv_{i+1}} \right)
        \leq \left\lfloor (k-3) \frac{k}{k-1} \right\rfloor
        = \left\lfloor (k-3) + \frac{k-3}{k-1} \right\rfloor
        = k - 3 \enspace.
    \]
    It is easy to see that the obtained inequality is a Chv\'atal inequality of rank 1. In particular, we choose multipliers $\frac{1}{d}$ for the considered cycle inequalities, $\frac{d-1}{d}$ for the considered box inequalities, and $0$ for all the remaining inequalities of canonical relaxation of $\mc(K_n)$. It follows immediately that the multipliers are non-negative by construction.
\end{delayedproof}

\begin{delayedproof}{lem:roots-of-half-chorded}
    Let $r := \abs{I}$. 
    For $r = 0$ it holds $W = \emptyset$, i.e. $x = \chi(W) = 0$ and the left hand side of \eqref{eq:half-chorded-odd-cycle-ineq} evaluates to $0$ and the inequality is not satisfied with equality.

    Next, assume $r=1$, i.e. $I=\{j\}$ for some $j \in \Z_k$. 
    Then, since $v$ is true to $C$, the path along the nodes $[v_{j+1},v_{j}]$ is not cut with respect to $x$. 
    By the path inequalities \eqref{eq:path} none of the edges $v_iv_{i+1}$ and $v_iv_{i+d}$ are cut with respect to $x$. 
    As for $r=0$, the left hand side of \eqref{eq:half-chorded-odd-cycle-ineq} evaluates to $0$ and the inequality is not satisfied with equality.

    For $r \geq 2$ an edge $v_iv_{i+1}$ is cut if and only if $i \in I$. 
    It follows that $\sum_{i \in \Z_k} x_{v_iv_{i+1}} = r$. 
    Therefore, inequality \eqref{eq:half-chorded-odd-cycle-ineq} can be written as $\sum_{i \in \Z_k} x_{v_iv_{i+d}} - r \leq k - 3$. Clearly, it holds that $\sum_{i \in \Z_k} x_{v_iv_{i+d}} \leq k$. Thus, for $r \geq 4$ the left hand side of \eqref{eq:half-chorded-odd-cycle-ineq} is $\leq k - 4$ and \eqref{eq:half-chorded-odd-cycle-ineq} is not satisfied with equality.

    Next, assume $r=2$, i.e. $I=\{j,\ell\}$ for some $j,\ell \in \Z_k$ with $j \neq \ell$. 
    By the above observation, the inequality \eqref{eq:half-chorded-odd-cycle-ineq} is satisfied with equality if and only if all but one of the half-chords $v_iv_{i+d}$ for $i \in \Z_k$ are cut with respect to $x$. 
    This is satisfied if and only if $j - \ell \in \{d,d+1\}$.

    Lastly, assume $r=3$. 
    The inequality \eqref{eq:half-chorded-odd-cycle-ineq} is satisfied with equality if and only if all of the half-chords $\{v_i, v_{i+d}\}$ for $i=0,\dots,k-1$ are cut. 
    This is satisfied if and only if for all $j,\ell \in I$ with $j \neq \ell$ it holds that $j-\ell \in \{1,\dots,d\}$ or $\ell - j \in \{1,\dots,d\}$.
\end{delayedproof}

\begin{delayedproof}{lem:half-chorded-facet}
    Let $S := \{x \in X_n \mid x \text{ satisfies \eqref{eq:half-chorded-odd-cycle-ineq} with equality} \}$ and let $\Sigma := \conv S$ be the face that is defined by \eqref{eq:half-chorded-odd-cycle-ineq}. 
    Let $\Sigma'$ be a facet of $\lmc(C)$ with $\Sigma \subseteq \Sigma'$ and suppose $\Sigma'$ is defined by an inequality $a^\top x \leq \beta$ for some $a \in \R^{\binom{\Z_n}{2}}$ and $\beta \in \R$. 
    Let $S' := \{x \in X_n \mid a^\top x = \beta \}$ be the set of integral points in the facet $\Sigma'$, i.e. $S \subseteq S'$. 
    We show that \eqref{eq:half-chorded-odd-cycle-ineq} is a positive scalar multiple of $a^\top x \leq \beta$ and hence $\Sigma=\Sigma'$, i.e. \eqref{eq:half-chorded-odd-cycle-ineq} defines a facet of $\lmc(C)$. 
    In particular we need to show that there exists $\alpha > 0$ with $\alpha = a_{v_iv_{i+d}} = - a_{v_iv_{i+1}}$ for $i \in \Z_k$ with $d:=\frac{k-1}{2}$, $\beta = \alpha (k-3)$, and $a_f = 0$ for all other edges $f$. 
    We start by showing $a_f = 0$ for all $f \in \binom{\Z_n}{2}\setminus E^*$ with $E^*:=\bigl\{v_iv_{i+1}, v_iv_{i+d} \mid i \in \Z_k \bigr\}$ in two steps:
    \begin{enumerate}[(i)]
        \item \label{item:step-1-half-chorded}
        $a_{uw} = 0$ for $u \in \Z_n \setminus \{v_0,\dots,v_{k-1}\}$ and $w \in \Z_n \setminus\{u\}$,
        \item \label{item:step-2-half-chorded}
        $a_{vw} = 0$ for $vw \subseteq \{v_0,\dots,v_{k-1}\}$ and $vw \notin E^*$.
    \end{enumerate}

    \noindent \textit{Proof of \ref{item:step-1-half-chorded}}.
    Let $i \in \Z_k$ such that $u \in \;]v_i,v_{i+1}[$ and let $v^*:=v_{i-d}$ (this $i$ is unique since $v$ is true to $C$). 
    First we show $a_{uw} = 0$ for $w \in \; ]u,v^*]$. 
    We define $x^1(w) := \chi(\{u-1,u,w,v^*\})$ and $x^2(w) := \chi(\{u-1,w,v^*\})$. 
    By \Cref{lem:roots-of-half-chorded}, both $x^1(w)$ and $x^2(w)$ satisfy \eqref{eq:half-chorded-odd-cycle-ineq} with equality.
    Indeed, we have $u,u-1 \in [v_i,v_{i+1}[$ and $v^* \in [v_{i-d},v_{i-d+1}[$, i.e. $I = \{i, i-d, j\}$ where $j \in \Z_k$ such that $w \in [v_j,v_{j+1}[$. 
    In case $j \in \{i,i-d\}$ we have $I=\{i,i-d\}$ and condition \ref{cond:roots-half-chorded-a} is satisfied. 
    Otherwise condition \ref{cond:roots-half-chorded-b} is satisfied. 
    Therefore, it holds that $x^1(w),x^2(w) \in S \subseteq S'$, i.e. $a^\top x^1(w) = a^\top x^2(w) = \beta$.
    Lastly, it holds that $a^\top x(w) = 0$ for $x(w) := x^1(w) - x^2(w)$. 
    By construction we have for all $w \in \; ]u,v^*]$ that $x(w)_{uw'} = 1$ for $w' \in \; ]u,w]$ and $x(w)_f=0$ for all other edges $f$. 
    For $z(w) := x(w) - x(w-1)$ (with $x(u) = 0$), it holds that $z(w)_{uw} = 1$ and $z(w)_f = 0$ for all other edges $f$. 
    It follows that $0 = a^\top z(w) = a_{uw}$ for all $w \in \; ]u,v^*]$.

    For $w \in \; ]v^*,u[$ an analogous construction yields $a_{uw} = 0$ and Claim \ref{item:step-1-half-chorded} follows.

    Note that the presented construction does not work for $u=v_i$ for some $i$ because in that case we would have $I=\{i-1,i,i+d,j\}$ and $x^1(w)$ does not satisfy \eqref{eq:half-chorded-odd-cycle-ineq} with equality, by \Cref{lem:roots-of-half-chorded}.
    
    \noindent \textit{Proof of \ref{item:step-2-half-chorded}}. 
    We need to show $a_{v_iv_{i+\ell}} = 0$ for all $i \in \Z_k$ and $\ell \in \{2,\dots,d-1\}$. 
    To that end we define the following vectors:
    \begin{align*}
        x^1(i, \ell) &:= \chi(\{v_i, v_{i+d}-1, v_{i-\ell}\}) \\
        x^2(i, \ell) &:= \chi(\{v_i-1, v_{i+d}-1, v_{i-\ell}\}) \\
        x^3(i, \ell) &:= \chi(\{v_i-1, v_{i+\ell}-1, v_{i-d}\}) \\
        x^4(i, \ell) &:= \chi(\{v_i, v_{i+\ell}-1, v_{i-d}\})
    \end{align*}
    for $i \in \Z_k$ and $\ell \in \{2,\dots,d-1\}$. All these vectors satisfy condition \ref{cond:roots-half-chorded-b} of \Cref{lem:roots-of-half-chorded} and thus satisfy \eqref{eq:half-chorded-odd-cycle-ineq} with equality. 
    Therefore, it holds that $x^j(i, \ell) \in S \subseteq S'$, and hence $a^\top x^j(i, \ell) = \beta$, for $j=1,2,3,4$. 
    Further, we define $x(i, \ell) := x^1(i, \ell) - x^2(i, \ell) + x^3(i, \ell) - x^4(i, \ell)$ which satisfies $a^\top x(i, \ell) = 0$. 
    By construction, it holds that $x(i,\ell)_{v_iw} = 1$ for $w \in [v_{i+\ell}, v_{i+d}[\; \cup \; ]v_{i-d},v_{i-\ell}]$ and $x(i,\ell)_f = 0$ for all other edges $f$. 
    Next, we define $z(i,d-1) := x(i,d-1)$ and $z(i, \ell) := x(i,\ell) - x(i,\ell+1)$ for $\ell \in \{2,\dots,d-2\}$. 
    For $\ell \in \{2,\dots,d-1\}$, it holds that $z(i,\ell)_{v_iw} = 1$ for $w \in [v_{i+\ell}, v_{i+\ell+1}[ \; \cup \; ]v_{i-\ell-1},v_{i-\ell}]$ and $z(i,\ell)_f=0$ for all other edges $f$. Thus, $a^\top z(i,\ell) = 0$ yields
    \begin{align}\label{eq:half-chorded-sum-coeff-zero}
        \sum_{w\in [v_{i+\ell}, v_{i+\ell+1}[} a_{v_iw} + \sum_{w \in \;]v_{i-\ell-1},v_{i-\ell}]} a_{v_iw} = 0 \enspace .
    \end{align}
    By Claim~\ref{item:step-1-half-chorded}, it holds that $a_{v_iw} = 0$ for $w \notin \{v_0,\dots,v_{k-1}\}$ and \eqref{eq:half-chorded-sum-coeff-zero} yields $a_{v_iv_{i+\ell}} + a_{v_i,v_{i-\ell}} = 0$. 
    
    Now consider the cycle in $K_n$ along the nodes $v_i,v_{i+\ell},v_{i+2\ell},\dots,v_{i-\ell},v_i$. 
    By the above observation the coefficient $a_f$ of the edges $f$ along that cycle have alternating sign. 
    The length of this cycle is the smallest integer $p$ such that $\ell \cdot p = 0 \; (\text{mod} \; k)$. 
    If $p$ was even, say $p=2q$ it would follow $\ell \cdot q = 0 \; (\text{mod} \; k)$ since $k$ is odd and $p$ would not be smallest. 
    Therefore, the cycle has odd length $p$.
    As the coefficients $a_f$ of all edges $f$ on that cycle have the same absolute value and an alternating sign, all these coefficients must be zero. 
    This proves Claim~\ref{item:step-2-half-chorded}. 
    For an illustration of this proof, see \Cref{fig:half-chorded-facet-proof}.
    
    Next, we show that there exists $\alpha \in \R$ with $\alpha = -a_{v_iv_{i+1}}$ for all $i \in \Z_k$. 
    For that we consider the vectors $x^1(i, 2), x^2(i, 2) \in S \subseteq S'$ from the proof of Claim~\ref{item:step-2-half-chorded}, and set $x := x^1(i, 2) - x^2(i, 2)$ for which it holds that $a^\top x = 0$. 
    By construction we have $x_{v_iw} = -1$ for $w \in \; ]v_{i-2},v_i[$ and $x_{v_iw} = 1$ for $w \in \; ]v_i,v_{i+d}[$ and $x_f = 0$ for all other $f \in \binom{\Z_n}{2}$. 
    By \ref{item:step-1-half-chorded} and \ref{item:step-2-half-chorded} we have $a_{v_iw} = 0$ for all $w \in \; ]v_{i-2},v_{i+d}[ \setminus \{v_{i-1}, v_i, v_{i+1}\}$ and $a^\top x = 0$ yields $a_{v_iv_{i-1}} = a_{v_iv_{i+1}}$. 
    This holds for all $i \in \Z_k$, i.e. there exists $\alpha \in \R$ with $\alpha = -a_{v_iv_{i+1}}$ for all $i \in \Z_k$.

    We continue by showing $\alpha = a_{v_iv_{i+d}}$ for all $i \in \Z_k$. 
    For that we consider the vectors $x^1 := \chi(\{v_i-1, v_i, v_{i+d}\})$ and $x^2 := \chi(\{v_i-1,v_{i+d}\})$ which satisfy Conditions~\ref{cond:roots-half-chorded-b} and \ref{cond:roots-half-chorded-a} of \Cref{lem:roots-of-half-chorded} respectively, i.e. $x^1,x^2 \in S \subseteq S'$. 
    We define $x := x_1 - x_2$ which satisfies $a^\top x = 0$ and by construction we have $x_{v_iw} = 1$ for $w \in \;]v_i, v_{i+d}]$. 
    By \ref{item:step-1-half-chorded} and \ref{item:step-2-half-chorded}, it holds that $a_{v_iw} = 0$ for all $w \in \; ]v_i, v_{i+d}] \setminus \{v_{i+1},v_{i+d}\}$ and $a^\top x = 0$ yields $a_{v_iv_{i+1}} + a_{v_iv_{i+d}} = 0$. 
    With the definition of $\alpha$ from above it follows that $\alpha = a_{v_iv_{i+d}}$ for $i \in \Z_k$.

    Lastly, by plugging any of the previously considered vectors $x \in S \subseteq S'$ in $a^\top x = \beta$ we obtain $\beta = \alpha(k-3)$. 
    It holds that $\alpha \neq 0$, because otherwise $a^\top x \leq \beta$ would be $0^\top x \leq 0$ which is obviously not facet-defining. 
    Further, it holds that $\alpha > 0$ because for $\alpha < 0$ the vector $0 \in X_n$ would violate the inequality $a^\top 0 \leq \alpha (k-3)$. 
    All together we have shown that \eqref{eq:half-chorded-odd-cycle-ineq} is indeed a positive scalar multiple of $a^\top x \leq \beta$ and, thus, is facet-defining for $\lmc(C)$.
\end{delayedproof}

\begin{figure}
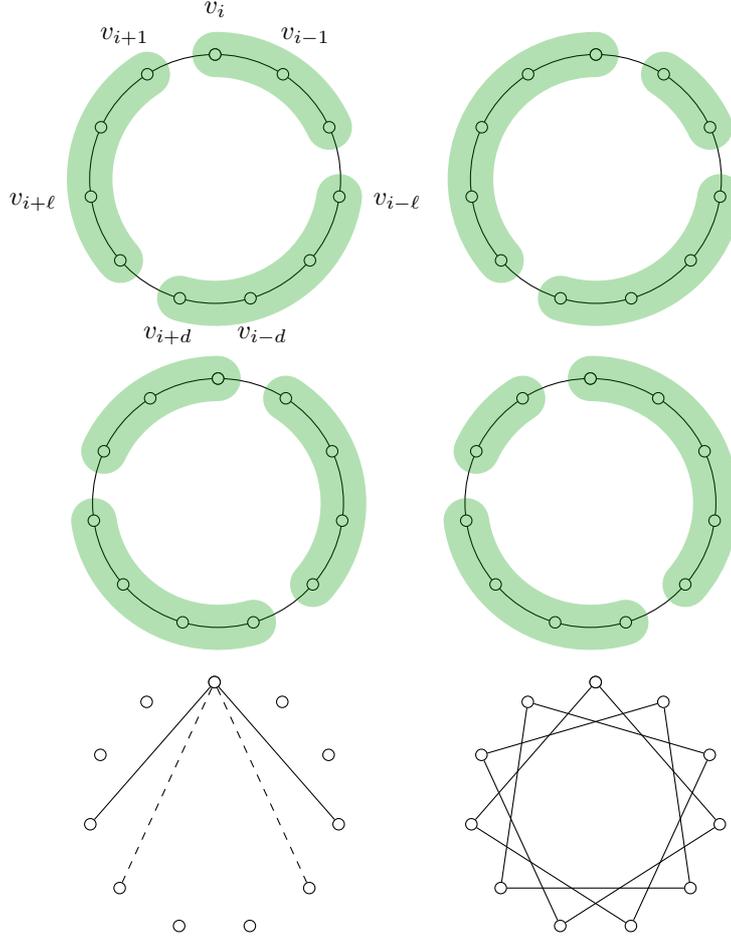

    \centering
    \input{figures/cycle/half-chorded-facet-proof-visualization-1.tex}
    \input{figures/cycle/half-chorded-facet-proof-visualization-2.tex} \\
    \hspace{27pt}
    \input{figures/cycle/half-chorded-facet-proof-visualization-3.tex}
    \hspace{16pt}
    \input{figures/cycle/half-chorded-facet-proof-visualization-4.tex} \\
    \hspace{25pt}
    \input{figures/cycle/half-chorded-facet-proof-visualization-5.tex}
    \hspace{35pt}
    \input{figures/cycle/half-chorded-facet-proof-visualization-6.tex}
    \caption{Illustration of the construction of the proof of \ref{item:step-2-half-chorded} from the proof of \Cref{lem:half-chorded-facet}. 
    In the shown example we have $n=k=11$, i.e. $d=5$, and $\ell=3$ (note that for $n=k$ the set $\Z_n \setminus \{v_0,\dots,v_{k-1}\}$ is empty and \ref{item:step-1-half-chorded} does not occur). 
    The first four graphs illustrate the decompositions of $C$ with respect to the vectors $x^1(i,\ell)$ to $x^4(i,\ell)$. 
    In the fifth graph the four drawn edges are precisely the edges $f$ with $x(i,\ell)_f = 1$. 
    The continuously drawn edges are the edges $f$ with $z(i,\ell)_f=1$, i.e. the edges $\{v_i,v_{i+\ell}\}$ and $\{v_i,v_{i-\ell}\}$. 
    The sixth graph is the cycle of odd length that contains all edges $\{v_i,v_{i+\ell}\}$ for $i \in \Z_k$.}
    \label{fig:half-chorded-facet-proof}
\end{figure}

\begin{delayedproof}{lem:half-chorded-not-true-not-facet}
    This proof follows the same line of the proof of \Cref{prop:2-chorded-not-facet}\ref{item:2-chorded-not-true-not-facet}.
    Let $u,w \in \{v_0,\dots,v_{k-1}\}$ distinct, as in the proof of \Cref{prop:2-chorded-not-facet}\ref{item:2-chorded-not-true-not-facet}, i.e. $uw \neq v_iv_{i+1}$ for all $i \in \Z_k$ and $]u,w[ \; \cap \{v_0,\dots,v_{k-1}\} = \emptyset$.
    We show that all $x \in X_n$ that satisfy \eqref{eq:half-chorded-odd-cycle-ineq} with equality also satisfy 
    \begin{align}\label{eq:half-chorded-intersection}
        x_{uw} + x_{u-1,w+1} = x_{u,w+1} + x_{u-1,w}
    \end{align}
    with equality and, therefore, \eqref{eq:half-chorded-odd-cycle-ineq} is not facet-defining for $\lmc(C)$. 
    Let $x \in X_n$ such that $x$ satisfies \eqref{eq:half-chorded-odd-cycle-ineq} with equality and assume $x$ does not satisfy \eqref{eq:half-chorded-intersection}. 
    Then, as in the proof of \Cref{prop:2-chorded-not-facet}\ref{item:2-chorded-not-true-not-facet}, it follows that $x_{u-1,u} = x_{w,w+1} = 1$. 
    Now let $i_1,i_2 \in \Z_k$ be the indices such that $u=v_{i_1}$ and $w=v_{i_2}$. 
    As in the proof of \Cref{prop:2-chorded-not-facet}\ref{item:2-chorded-not-true-not-facet}, it follows that $x_{v_jv_{j-1}} = x_{v_jv_{j+1}} = 1$ for $j=i_1,i_2$. 
    It holds that
    \begin{align}\label{eq:cycle-leq-minus-1}
        x_{v_iv_{i+d}} - \sum_{\ell = 0}^{d-1} x_{v_{i+\ell}v_{i+\ell+1}} \leq -1
    \end{align}
    for all $i \in \Z_k$ with $i_1 \in \{i+1,i+2,\dots,i+d-1\}$ or $i_2 \in \{i+1,i+2,\dots,i+d-1\}$. 
    There are at least $d$ many such $i$. 
    Adding all cycle inequalities \eqref{eq:cycle-ineq-for-half-chorded} and $(d-1)$ times the box inequalities $x_{v_iv_{i+d}} \leq 1$ for $i \in \Z_k$, as in the proof of \Cref{lem:half-chorded-valid}, we obtain
    \[
        d \sum_{i \in \Z_k} \left( x_{v_iv_{i+d}} - x_{v_iv_{i+1}} \right)
        \leq k (d-1) - d \enspace,
    \]
    where the $-d$ on the right hand side is due to \eqref{eq:cycle-leq-minus-1}.
    Dividing both sides by $d$ rounding the right hand side down to the nearest integer yields
    \[
        \sum_{i = 1}^k \left(x_{v_iv_{i+d}} - x_{v_iv_{i+1}}\right) 
        \leq \left\lfloor \frac{k(d-1)}{d} - 1 \right\rfloor = k - 4 \enspace.
    \]
    Therefore $x$ does not satisfy \eqref{eq:half-chorded-odd-cycle-ineq} with equality in contradiction to our assumption and the claim follows.
\end{delayedproof}

\begin{delayedproof}{prop:num-half-chorded-odd-cycle-facet}
    For every set of $k$ nodes of $C$ there is only one half-chorded odd cycle inequality \eqref{eq:half-chorded-odd-cycle-ineq} that is facet-defining for $\lmc(C)$ because each enumeration of $k$ nodes that is true to $C$ yields the same half-chorded odd cycle inequality. 
    There are $2^{n-1}$ ways to select an odd number of nodes. There are $n$ and $\frac{n(n-1)(n-2)}{6}$ ways to select just $1$ and $3$ nodes respectively.
    Together the claim follows.
\end{delayedproof}

\begin{delayedproof}{thm:half-chorded-facet-general-g}
    By \Cref{lem:half-chorded-valid}, the inequality is valid for $\mc(K_V)$ and thus, by \Cref{prop:inclusion-property}, valid for $\lmc(G)$.
    We show facet-definingness by extending the proof of \Cref{lem:half-chorded-facet}. 
    To that end we identify the $n$ nodes $V_C$ with $\Z_n$ in such a way that $E_C = \{\{v,v+1\} \mid v \in \Z_n\}$. 
    Let $K_n=(\Z_n,\binom{\Z_n}{2})$ be the subgraph of $K_V$ that is induced by the node set $\Z_n$. 
    For a vector $x \in \{0,1\}^{\binom{\Z_n}{2}}$ we define the \emph{extension} of $x$ as $\bar{x} \in \{0,1\}^{\binom{V}{2}}$ with
    \[
        \bar{x}_e = \begin{cases}
            x_e & \text{for } e \in \binom{\Z_n}{2} \\
            1   & \text{otherwise} \enspace.
        \end{cases}
    \]
    If $x$ is the characteristic vector of a multicut of $K_n$ lifted from $C$ that is induced by the partition $\Pi$ of $\Z_n$, then $\bar{x}$ is the characteristic vector of the multicut of $K_V$ lifted from $G$ that is induced by the partition $\bar{\Pi} = \Pi \cup \{\{v\} \mid v \in V \setminus \Z_n\}$ of $V$. 
    
    As in the proof of \Cref{lem:half-chorded-facet}, let $\Sigma$ be the face of $\lmc(G)$ that is defined by \eqref{eq:half-chorded-odd-cycle-ineq}, let $\Sigma'$ be a facet of $\lmc(G)$ with $\Sigma \subseteq \Sigma'$ and suppose $\Sigma'$ is defined by the inequality $a^\top x \leq \beta$ with $a \in \R^{\binom{V}{2}}$ and $\beta \in \R$. 
    Let $S := \Sigma \cap \{0,1\}^{\binom{V}{2}}$ and $S' := \Sigma' \cap \{0,1\}^{\binom{V}{2}}$ be the sets of the characteristic vectors of multicuts of $K_V$ lifted from $G$ that satisfy \eqref{eq:half-chorded-odd-cycle-ineq} and $a^\top x \leq \beta$ with equality.

    For every characteristic vector $x$ of a multicut lifted from $C$ to $K_n$ that satisfies \eqref{eq:half-chorded-odd-cycle-ineq} (restricted to $\binom{\Z_n}{2}$) with equality, the extension $\bar{x}$ also satisfies \eqref{eq:half-chorded-odd-cycle-ineq} with equality. 
    By considering the extension of the vectors $\chi(W)$ used in the proof of \Cref{lem:half-chorded-facet}, we obtain that $a_e = 0$ for all edges $e \in \binom{\Z_n}{2} \setminus \{v_i v_{i+1},v_iv_{i+d} \mid i \in \Z_k\}$ and that there exists $\alpha \in \R$ with $\alpha = a_{v_i v_{i+d}} = - a_{v_i v_{i+1}}$ for all $i=0,\dots,k-1$. 
    It remains to show $a_{uw} = 0$ for all $uw \in \binom{V}{2}$ with $uw \not\subseteq \Z_n$.

    \begin{claim}\label{claim:1}
        For every $uw \in \binom{V}{2}$ with $\{u,w\} \not\subseteq \Z_n$ there exists a $uw$-path $P=(V_P,E_P)$ in $G$ such that there exists $x \in S$ corresponding to the partition $\Pi = \{\pi_1,\dots,\pi_p\}$ with $V_P = \pi_j$ for some $j\in \{1,\dots,p\}$.
    \end{claim}

    \begin{proof}[Proof of \Cref{claim:1}]
    First assume that there exists a $uw$-path in $G-\Z_n$ where $G-\Z_n$ denotes the graph that is obtained by removing the nodes $\Z_n$ and all edges incident to at least one node in $\Z_n$ from $G$. 
    Let $P=(V_P,E_P)$ be such a path. 
    Let $\Pi^*$ be a decomposition of $C$ such that the characteristic vector $x^*$ of the induced multicut of $K_n$ lifted from $C$ satisfies \eqref{eq:half-chorded-odd-cycle-ineq} (restricted to $\binom{\Z_n}{2})$) with equality. 
    Then, $\Pi := \Pi^* \cup \{V_P\} \cup \{\{v\} \mid v \in V \setminus (\Z_n \cup V_P)\}$ is a decomposition of $G$. 
    Let $x$ be the characteristic vector of the multicut of $K_V$ lifted from $G$ corresponding to $\Pi$. 
    It holds that $x_e = x^*_e$ for all $e \in \binom{\Z_n}{2}$ and, hence, $x$ satisfies \eqref{eq:half-chorded-odd-cycle-ineq} with equality.

    Next assume that there does not exist a $uw$-path in $G-\Z_n$. 
    Since $G$ is connected there exist $u',w' \in \Z_n$ such that there is a $uu'$-path $P_u = (V_u,E_u)$ in $G-(\Z_n \setminus \{u'\})$ and a $ww'$-path $P_w = (V_w,E_w)$ in $G-(\Z_n \setminus \{w'\})$. 
    Note that in case $u \in \Z_n$ it holds that $u=u'$ and $E_u = \emptyset$, and in case $w \in \Z_n$ it holds that $w=w'$ and $E_w = \emptyset$.
    As there exists no $uw$-path in $G-\Z_n$, it holds that $E_u \cap E_w = \emptyset$. 
    Let $i_1,i_2 \in \Z_k$ such that $u' \in \; ]v_{i_1},v_{i_1+1}]$ and $w' \in \; ]v_{i_2},v_{i_2+1}]$. 
    We may assume $i_2 - i_1 \leq d$ (mod $k$) by possibly interchanging $u$ and $w$. 
    Let $x^* := \chi(\{u'-1,w', v_{i_1+d+1}\})$ be the characteristic vector of the multicut of $K_n$ lifted from $C$ corresponding to the partition $\{[u',w'], \;]w', v_{i_1+d+1}], \;]v_{i_1+d+1},u'[\}$ of $\Z_n$. 
    By \Cref{lem:roots-of-half-chorded}, $x^*$ satisfies \eqref{eq:half-chorded-odd-cycle-ineq} (restricted to $\binom{\Z_n}{2}$) with equality (in the case where $w' = v_{i_1+d+1}$ holds, condition \ref{cond:roots-half-chorded-a} is satisfied, otherwise condition \ref{cond:roots-half-chorded-b} is satisfied). 
    Now let $P=(V_P,E_P)$ be the $uw$-path with $V_P = V_u \cup V_w \cup [u',w']$ and $E_P = E_u \cup E_w \cup \{\{v,v+1\} \mid v \in [u',w'[$. 
    Let $\Pi := \{V_P, \;]w', v_{i_1+d+1}], \;]v_{i_1+d+1},u'[\} \cup \{\{v\} \mid v \in V \setminus (\Z_n \cup V_u \cup V_w)\}$ be a partition of $V$. 
    By construction $\Pi$ is a decomposition of $G$. 
    Let $x$ be the characteristic vector of the multicut lifted from $G$ to $K_V$ that corresponds to the decomposition $\Pi$. 
    Again, by construction, it holds that $x_e = x^*_e$ for $e \in \binom{\Z_n}{2}$ and therefore $x$ satisfies \eqref{eq:half-chorded-odd-cycle-ineq} with equality and the path $P$ meets the requirements from \Cref{claim:1}. This concludes the proof of \Cref{claim:1}. 
    For an illustration we refer to \Cref{fig:proof-zero-lifting-half-chorded}.
    \end{proof}

    For $uw \in \binom{V}{2}$ with $\{u,w\} \not\subseteq \Z_n$ let $P=(V_P,E_P)$ be a shortest $uw$-path that meets the requirements of \Cref{claim:1}, and let $d(u,w)$ be the length of $P$. 
    Further, let $\Pi=\{\pi_1,\dots,\pi_p\}$ be a decomposition of $G$ that meets the requirements of \Cref{claim:1}. 
    In particular, let $j \in \{1,\dots,p\}$ with $\pi_j = V_P$. 
    Define $\Pi' := \{\pi_i \mid i \in \{1,\dots,p\}\setminus \{j\}\} \cup \{\pi_j \cap \Z_n\} \cup \{\{v\} \mid v \in \pi_j \setminus \Z_n\}$. 
    Clearly $\Pi'$ is a decomposition of $G$. 
    Let $x$ and $x'$ be the characteristic vectors of the multicuts of $K_V$ lifted from $G$ corresponding to $\Pi$ and $\Pi'$ respectively. 
    By construction and \Cref{claim:1}, $x$ satisfies \eqref{eq:half-chorded-odd-cycle-ineq} with equality. 
    Further, it holds that $x'_e = x_e$ for all $e \in \binom{\Z_n}{2}$ and therefore also $x'$ satisfies \eqref{eq:half-chorded-odd-cycle-ineq} with equality. 
    It follows that $x,x' \in S \subseteq S'$ and $z:= x' - x$ satisfies $a^\top z = 0$. 
    By construction, it holds that $z_{st} = 1$ for all $s,t \in V_P$ with $s \neq t$ and $\{s,t\} \not\subseteq \Z_n$ and $z_e = 0$ for all other edges $e$.
    Thus, $a^\top z = 0$ yields
    \begin{align}\label{eq:sum-s-t-in-V_P-zero}
        \sum_{s,t \in V_P, s \neq t, \{s,t\} \not\subseteq \Z_n} a_{st} = 0 \enspace.
    \end{align}
    We prove $a_{uw} = 0$ for all $uw \in \binom{V}{2}$ with $\{u,w\} \not\subseteq \Z_n$ by induction over $d(u,w)$.
    For $d(u,w) = 1$ equality \eqref{eq:sum-s-t-in-V_P-zero} yields the desired $a_{uw} = 0$. 
    Now assume that it holds $a_{st} = 0$ for all $st \in \binom{V}{2}$ with $\{s,t\} \not\subseteq \Z_n$ and $d(s,t) < d(u,v)$.
    From the proof of \Cref{claim:1}, it is easy to see that for all $s,t \in V_P$ with $s \neq t$, $\{s,t\} \not \subseteq \Z_n$ and $st \neq uw$ it holds that $d(s,t) < d(u,w)$. 
    This assumption, together with \eqref{eq:sum-s-t-in-V_P-zero}, yields $a_{uw} = 0$.

    Like in the proof of \Cref{lem:half-chorded-facet} it follows that $a^\top x \leq \beta$ is a positive scalar multiple of \eqref{eq:half-chorded-odd-cycle-ineq}, i.e. $\Sigma=\Sigma'$, and \eqref{eq:half-chorded-odd-cycle-ineq} is facet-defining for $\lmc(G)$.
\end{delayedproof}

\begin{figure}
    \centering
    \input{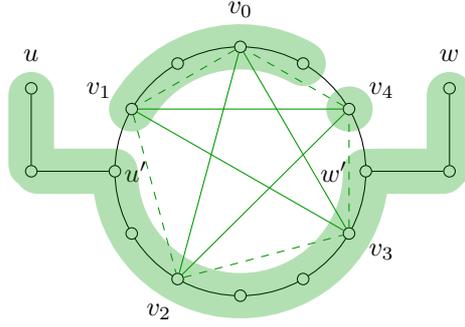}
    \caption{Depicted is a decomposition of a graph $G$ that satisfies the condition of \Cref{claim:1} in the proof of \Cref{thm:half-chorded-facet-general-g}.
    The depicted graph (black edges) has a cycle $C$ of length $12$. 
    The green edges depict the support graph of the half-chorded odd cycle inequality \eqref{eq:half-chorded-odd-cycle-ineq} with respect to $v_0,\dots,v_4$. 
    The green areas illustrate the decomposition of $G$ that is constructed in the proof: 
    it contains one component that corresponds to a $uw$-path and it satisfies \eqref{eq:half-chorded-odd-cycle-ineq} with equality.}
    \label{fig:proof-zero-lifting-half-chorded}
\end{figure}

\begin{delayedproof}{prop:half-chorded-separation}
    Let $k \geq 5$ odd, let $d=\frac{k-1}{2}$, and let $v:\Z_k \to \Z_n$.
    Define $w: \Z_k \to \Z_n$ with $w_i = v_{i \cdot d}$ for $i \in Z_k$.
    Then it holds that $\bigl\{\{w_i,w_{i+1}\} \mid i \in \Z_k \bigr\} = \bigl\{\{v_i,v_{i+d}\} \mid i \in \Z_k \bigr\}$ and $\bigl\{\{w_i,w_{i+2}\} \mid i \in \Z_k \bigr\} = \bigl\{\{v_i,v_{i+1}\} \mid i \in \Z_k \bigr\}$. 
    In particular, the edges $\{w_i,w_{i+1}\}$ for $i \in \Z_k$ induce a cycle and the edges $\{w_i,w_{i+2}\}$ for $i \in \Z_k$ are the 2-chords of that cycle.
    The half-chorded odd cycle inequality \eqref{eq:half-chorded-odd-cycle-ineq} with respect to $v$ can be written in terms of $w$ as 
    \begin{align}\label{eq:half-chorded-separation}
        \sum_{i \in \Z_k} (x_{w_iw_{i+2}} - x_{w_i,w_{i+1}} + 1) \geq 3 \enspace.
    \end{align}
    As detailed in \citet{muller1996partial}, for a given $x \in [0,1]^{\binom{\Z_n}{2}}$ one can decide in polynomial time whether there exists an inequality of the form \eqref{eq:half-chorded-separation} that is violated by $x$. 
    Note that \citet{muller1996partial} considers the clique partitioning problem instead of the multicut problem, i.e. before the algorithm can be applied, inequality \eqref{eq:half-chorded-separation} has to be reformulated by substituting $1-x$ for $x$.
\end{delayedproof}

\begin{delayedproof}{prop:num-star-inequalities}
    For any node $v$ there are $2^{n-2}$ ways to pick an odd number of nodes of $\Z_n \setminus \{v\}$ that are true to the cycle $C$, yielding the $n2^{n-2}$ part. The $-n(n-1)/2$ part is to account for the fact that we are counting the box inequalities twice. 
\end{delayedproof}

\begin{delayedproof}{prop:num-glider-inequalities}
    For any node $w \in \Z_n$ there are $2^{n-1}-n-\frac{(n-1)(n-2)}{2}$ ways to pick three or more nodes of $\Z_n \setminus \{w\}$ that are true to $C$. Additionally, there are $\frac{n(n-1)}{2}$ box inequalities. Together, there are 
    \[
        n \left( 2^n - n - \frac{(n-1)(n-2)}{2} \right) + \frac{n(n-1)}{2}
    \]  
    glider inequalities and the claim holds.
\end{delayedproof}

\begin{delayedproof}{lem:star-glider-valid}
    For $k=0$ the claim is trivial so from now on we assume $k \geq 1$. 
    Let $x \in X_n$ and let $c := \sum_{i=0}^k x_{v_iv_{i+1}}$ be the number of edges in the sail that are cut. 
    First we consider $c=0$, i.e. $x_{v_iv_{i+1}} = 0$ for all $i\in \{0,\dots,k\}$. 
    If $x_{v_iw_i^j} = 0$ for all $i\in \{1,\dots,k\}$ and all even $j\in\{1,..,m_i\}$ the inequality obviously holds. 
    Otherwise, let $i_-$ and $i_+$ be the smallest and largest indices $i \in \{1,\dots,k\}$ respectively, such that there exists an even $j \in \{1,\dots,m_i\}$ with $x_{v_iw_i^j} = 1$. 
    Let $j_-$ be the largest even index $j \in \{1,\dots,m_{i_-}\}$ with $x_{v_{i_-}w_{i_-}^j} = 1$ and let $j_+$ be the smallest even index $j \in \{1,\dots,m_{i_+}\}$ with $x_{v_{i_+}w_{i_+}^j} = 1$. 
    By the path inequalities \eqref{eq:path} the $v_{i_-}w_{i_-}^{j_-}$ path via $v_0$ and the $v_{i_+}w_{i_+}^{j_+}$ path via $v_{k+1}$ are both cut with respect to $x$. 
    For convenience we define
    \begin{align}\label{eq:s_i_definition}
        S(i) := \sum_{j=1}^{m_i} (-1)^j x_{v_iw_i^j} \enspace.
    \end{align}
    By the cut inequalities \eqref{eq:cut} the following holds 
    \begin{itemize}
        \item If $i_-=i_+$ we have $x_{v_{i_-}w_{i_-}^j} = 1$ for $j\in \{j_+,\dots,j_-\}$, i.e. $S(i_-) \leq 1$. 
        Otherwise, by the definition of the indices $i_-, i_+$, we have $x_{v_iw_i^j} = 1$ for $i=i_-$, $j \in \{1,\dots,j_-\}$ and for $i = i_+$, $j\in \{j_+,\dots,m_{i_+}\}$, i.e. $S(i_-) \leq 0$ and $S(i_+) \leq 0$. 
        \item For $i \in \{i_-+1,\dots,i_+-1\}$ we have $x_{v_iw_i^j} = 1$ for $j\in \{1,\dots,m_i\}$, i.e. $S(i) = -1$.
    \end{itemize}
    For $i \in \{1,\dots,i_--1\} \cup \{i_++1,\dots,k\}$ we have, by definition of $i_-$ and $i_+$, $x_{v_iw_i^j} = 0$ for all even $j \in \{1,\dots,m_i\}$, i.e. $S(i) \leq 0$. 
    Altogether, we have $\sum_{i=1}^k S(i) \leq 1$, i.e. inequality \eqref{eq:star-glider-inequality} holds.

    Now assume $c \geq 1$. 
    Let $i_- := \min \{i \in \{0,\dots,k\}: x_{v_iv_{i+1}} = 1\}$ and $i_+ := \max \{i \in \{0,\dots,k\}: x_{v_iv_{i+1}} = 1\}$. Due to $x_{v_{i_-}v_{i_-+1}} = 1$ and $x_{v_{i_+}v_{i_++1}} = 1$ and the path inequalities \eqref{eq:path}, there exist an edge $e_-$ on the paths along the nodes $[v_{i_-},v_{i_-+1}]$ and an edge $e_+$ on the path along the nodes $[v_{i_+},v_{i_++1}]$ with $x_{e_-} = x_{e_+} = 1$.
    For $i \in \{i_-+1,\dots,i_+\}$ the set $\{e_-,e_+\}$ is a $v_iw_i^j$-cut for $j=1,\dots,m_i$ and the cut inequality \eqref{eq:cut} with respect to that cut yields $x_{v_iw_i^j} = 1$ for $j=1,\dots,m_i$, i.e. it holds that $S(i) = -1$ for $i \in \{i_-+1,\dots,i_+\}$. 
    Next we show $S(i) \leq 0$ for all $i \in \{1,\dots,k\} \setminus \{i_-+1,\dots,i_+\}$. 
    Let $i \in \{1,\dots,i_-\}$. 
    If $x_{v_iw_i^j} = 0$ for all even $j\in \{1,\dots,m_i\}$ then $S(i) \leq 0$ holds. 
    Otherwise, if there exists an even $j \in \{1,\dots,m_i\}$ with $x_{v_iw_i^j} = 1$, the path inequalities \eqref{eq:path} yield that the $v_iw_i^j$-path along the nodes $[w_i^j,v_i]$ is cut. 
    Together with $x_{v_{i_-}v_{i_-+1}}=1$ the cut inequalities \eqref{eq:cut} yield $x_{v_iw_i^{j'}} = 1$ for all $j'\in \{1,\dots,j\}$ and $S(i) \leq 0$ holds. 
    Analogously one can show that $S(i) \leq 0$ holds for $i \in \{i_++1,\dots,k\}$.
    
    By the definition of $c$, $i_-$ and $i_+$, there are at least $c-1$ indices $i \in \{i_-+1,\dots,i_+\}$, i.e. $\sum_{i=1}^k S(i) \leq -(c-1)$ and with this \eqref{eq:star-glider-inequality} holds.
\end{delayedproof}

\begin{delayedproof}{thm:star-glider-facet}
    For $k=0$ the claim holds by \Cref{cor:upper-box-facet} so we may assume $k \geq 1$. 

    Let $S \subseteq X_n$ be the set of characteristic vectors of lifted multicuts that satisfy \eqref{eq:star-glider-inequality} with equality and let $\Sigma = \conv S$ be the face of $\lmc(C)$ that is defined by \eqref{eq:star-glider-inequality}. 
    Assume $\Sigma'$ is a facet of $\lmc(C)$ with $\Sigma \subseteq \Sigma'$ and suppose $\Sigma'$ is defined by the inequality $a^\top x \leq b$ with $a \in \R^{\binom{\Z_n}{2}}$ and $b \in \R$. 
    Let $S' \subseteq X_n$ be the set of characteristic vectors $x$ of lifted multicuts with $a^\top x = b$, i.e. $\Sigma' = \conv S'$. 
    We show that $a^\top x \leq b$ is a positive scalar multiple of \eqref{eq:star-glider-inequality} and, thus, $\Sigma = \Sigma'$ which yields that \eqref{eq:star-glider-inequality} is indeed facet-defining. 
    In particular, we need to show that there exists $\alpha > 0$ such that $b = \alpha$,  $a_{v_iv_{i+1}} = \alpha$ for $i=0,\dots,k$, $a_{v_iw_i^j} = (-1)^j \alpha$ for $i=1,\dots,k$, $j=1,\dots,m_i$ and $a_e = 0$ for all other edges $e \in \binom{\Z_n}{2} \setminus E^*$ with $E^* := \{v_iv_{i+1} \mid i\in \{0,\dots,k\}\} \cup \{v_iw_i^j \mid i \in \{1,\dots,k\}, j\in\{1,\dots,m_i\}\}$.
    
    To that end we introduce some notation. 
    For a connected subset $V \subseteq \Z_n$ of $C$ let $\Pi = \{V\} \cup \{\{u\} \mid u \in \Z_n \setminus V\}$ be the decomposition of $C$ that consists of the component $V$ and otherwise singular nodes. 
    Let $\psi(V) := \1_{\phi_{K_n}(\Pi)}$ be the characteristic vector of the multicut that is induced by the decomposition $\Pi$, i.e. $\psi(V)_e = 0 \iff e \subseteq V$.
    In particular, for $s,t \in \Z_n$, $s \neq t$, it holds that $\psi([s,t])_e = 0 \iff e \subseteq [s,t]$ and we have
    \begin{align}\label{eq:psi-sum-unit-vec}
        \psi([s,t[) + \psi(]s,t]) - \psi(]s,t[) - \psi([s,t]) = \1_{\{st\}} \enspace .
    \end{align}
    Note that the set $]s,t[$ is potentially empty.
    In that case we have $\psi(\emptyset) = 1$, the all one vector.

    If, for a given edge $st \in \binom{\Z_n}{2}$, we can show that 
    \begin{align}\label{eq:psi-in-S}
        \psi([s,t[), \psi(]s,t]), \psi(]s,t[), \psi([s,t]) \in S \subseteq S' \enspace,
    \end{align}
    then, by \eqref{eq:psi-sum-unit-vec}, it follows that $a^\top \1_{\{st\}} = 0$ and therefore $a_{st} = 0$. 
    With this preparation, we show that $a_{st}=0$ holds for ${st} \in \binom{\Z_n}{2} \setminus E^*$ by distinguishing between different cases:

    \begin{enumerate}
        \item 
        First, assume $s,t \in \;]v_k,v_1[$. 
        By potentially interchanging $s$ and $t$ we may assume $[s,t] \subseteq \;]v_k, v_1[$. 
        Then all edges $e \in E^*$ are cut with respect to the four characteristic vectors. 
        As the left hand side of \eqref{eq:star-glider-inequality} includes the coefficient $+1$ once more than the coefficient $-1$, all four characteristic vectors satisfy \eqref{eq:star-glider-inequality} with equality.
        Therefore, \eqref{eq:psi-in-S} is satisfied and by the argument above, it follows that $a_{st} = 0$.

        \item \label{item:s-t-in-v1-vk}
        Next, assume $s,t \in [v_1,v_k]$ and distinguish further cases:
        \begin{enumerate}
            \item \label{item:st-between-neighbouring-vi}
            If there exist $i \in \{1,\dots,k-1\}$ with $s,t \in [v_i,v_{i+1}]$ we may assume that $[s,t] \subseteq [v_i,v_{i+1}]$ by potentially interchanging $s$ and $t$. 
            Then, as before, \eqref{eq:psi-in-S} is satisfied since all edges $e \in E^*$ are cut with respect to the four characteristic vectors (by $st \notin E^*$ it holds that $st \neq v_iv_{i+1}$). 
            It follows that $a_{st} = 0$. 
            (Note that in the cases $s,t \in [v_0,v_1]$ and $s,t \in [v_k,v_{k+1}]$ the same argument yields $a_{st} = 0$).

            \item
            If there does not exist $i \in \{1,\dots,k-1\}$ with $s,t \in [v_i,v_{i+1}]$ we may assume $[t,s] \subseteq [v_1,v_k]$, again, by potentially interchanging $s$ and $t$. 
            Then, there exists $i,j \in \{2,\dots,k-1\}$ with $i \leq j$ and $s \in \; ]v_{j},v_{j+1}]$, $t \in [v_{i-1},v_{i}[$. 
            With respect to $\psi([s,t])$ an edge $e \in E^*$ is cut if and only if it is incident to a node $v_\ell$ with $\ell \in \{i,\dots,j\}$. Therefore, $\psi([s,t])$ satisfies \eqref{eq:star-glider-inequality} with equality.
            If it holds that $s = v_{j+1}$ then, additionally, the edges $e \in E^*$ that are incident to $v_{j+1}$ are cut with respect to $\psi(]s,t])$ and $\psi(]s,t[)$. 
            If it holds that $t = v_{i-1}$ then additionally the edges $e \in E^*$ that are incident to $v_{i-1}$ are cut with respect to $\psi([s,t[)$ and $\psi(]s,t[)$. 
            Either way, all four characteristic vectors satisfy \eqref{eq:star-glider-inequality} with equality, i.e. \eqref{eq:psi-in-S} holds and we obtain $a_{st} = 0$.
        \end{enumerate}

        \item
        It remains to consider the cases $s \in [v_1,v_k] \land t \in \; ]v_k,v_1[$ and $t \in [v_1,v_k] \land s \in \; ]v_k,v_1[$.
        By potentially interchanging $s$ and $t$, we may assume $s \in [v_1,v_k] \land t \in \; ]v_k,v_1[$.
        We distinguish further cases:
        \begin{enumerate}
            \item \label{item:s-is-vi}
            First, suppose $s=v_i$ for some $i \in \{1,\dots,k\}$. 
            We distinguish even further:
            \begin{enumerate}
                \item \label{item:vi-t-1}
                Consider that $t \in \; ]w_i^{m_i},v_1[$. 
                In case $i=1$ and $t \in \; ]v_0,v_1[$ the claim follows by \eqref{item:st-between-neighbouring-vi}. 
                So we may assume $i \neq 1$ or $t \in \; ]w_i^{m_i},v_0[$.
                Then, \eqref{eq:psi-in-S} is satisfied: 
                for $x=\psi([s,t])$ or $x=\psi([s,t[)$ an edge $e \in E^*$ is cut with respect to $x$ if and only if it is incident to a node $v_\ell$ for $\ell \in \{0,\dots,i-1\}$. 
                For $x=\psi(]s,t])$ or $x=\psi(]s,t[)$ and edge $e \in E^*$ is cut with respect to $x$ if and only if it is incident to a node $v_\ell$ for $\ell \in \{0,\dots,i\}$. 
                It follows that $a_{st} = 0$.

                \item 
                The case $t \in \; ]v_k,w_i^1[$ follows analogously to \ref{item:vi-t-1} by interchanging $s$ and $t$.
                 
                \item 
                It remains to consider $t \in [w_i^1,w_i^{m_i}]$. 
                Due to $st \notin E^*$ there exists $j \in \{1,\dots,m_i-1\}$ such that $t \in \; ]w_i^j,w_i^{j+1}[$. 
                First assume that $j$ is odd. 
                Then, \eqref{eq:psi-in-S} is satisfied by the following argument:
                As in \ref{item:vi-t-1}, all edges $e \in E^*$ that are incident to a node $v_\ell$ for $\ell \in \{0,\dots,i-1\}$ are cut with respect to all four characteristic vectors. 
                For $x=\psi([s,t])$ or $x=\psi([s,t[)$ additionally the edges $v_iw_i^\ell$ for $\ell \in \{j+1,\dots,m_i\}$ are cut with respect to $x$. 
                Since $j$ is odd $x$ satisfies \eqref{eq:star-glider-inequality} with equality. 
                For $x=\psi(]s,t])$ or $x=\psi(]s,t[)$ additionally all edges $e \in E^*$ that are incident to $v_i$ are cut with respect to $x$ and $x$ satisfies \eqref{eq:star-glider-inequality} with equality.
                It follows $a_{st} = 0$.
                
                If, otherwise, $j$ is even, interchanging $s$ and $t$ yields $a_{st} = 0$ by an analogous argument.
            \end{enumerate}

            \item 
            Otherwise, there exists an $i \in \{1,\dots,k-1\}$ such that $s \in \; ]v_i,v_{i+1}[$.
            We again distinguish further cases:
            \begin{enumerate}
                \item \label{item:t-not-vi-1}
                Consider $t \in \;]w_{i+1}^{m_{i+1}},v_1[$.
                Then, \eqref{eq:psi-in-S} is satisfied since for all four characteristic vectors an edge $e \in E^*$ is cut with respect to the respective vector if and only if $e$ is incident to a node $v_\ell$ for $\ell \in \{1,\dots,i\}$.
                It follows that $a_{st} = 0$.

                \item \label{item:t-not-vi-2}
                As before, the case $t \in \; ]v_k,w_i^1[$ follows analogously to \ref{item:t-not-vi-1} by interchanging $s$ and $t$.

                \item 
                It remains to consider $t \in \{w_i^1\} \cap \{w_{i+1}^{m_{i+1}}\}$ which only occurs if $w_i^1 = w_{i+1}^{m_{i+1}}$.
                For $x = \psi([s,t])$ or $x = \psi(]s,t])$ an edge $e \in E^*$ is cut with respect to $x$ if and only if it is adjacent to a node $v_\ell$ for $\ell \in \{1,\dots,i\}$ and therefore $x$ satisfies \eqref{eq:star-glider-inequality} with equality.
                Then, for $x := \psi(]s,t]) - \psi([s,t])$ it holds that $a^\top x = 0$. By construction we have that $x_{su} = 1$ for all $u \in \; ]s,t]$ and $x_e = 0$ for all other edges $e \in \binom{\Z_n}{2} \setminus \{su \mid u \in \; ]s,t]\}$. From $a^\top x = 0$ we obtain
                \begin{align} \label{eq:sum-a-su-zero}
                    \sum_{u \in \; ]s,t]} a_{su} = 0 \enspace .
                \end{align}
                By \ref{item:s-t-in-v1-vk} and \ref{item:t-not-vi-2} we have $a_{su} = 0$ for all $u \in \; ]s,t[$ and \eqref{eq:sum-a-su-zero} yields $a_{st} = 0$.
            \end{enumerate}
        \end{enumerate}
    \end{enumerate}

    It remains to show that there exists $\alpha > 0$ such that $b = \alpha$,  $a_{v_iv_{i+1}} = \alpha$ for $i=0,\dots,k$, $a_{v_iw_i^j} = (-1)^j \alpha$ for $i=1,\dots,k$, $j=1,\dots,m_i$.
    We show that $\alpha := a_{v_0v_1}$ satisfies this.
    For $i=1,\dots,k$, it is easy to see that $\psi([w_i^{m_i},v_i]), \psi([w_i^{m_i},v_i[) \in S \subseteq S'$ and, therefore, $x:=\psi([w_i^{m_i},v_i[) - \psi([w_i^{m_i},v_i])$ satisfies $a^\top x = 0$. 
    It holds that $x_{v_iu} = 1$ for $u \in [w_i^{m_i},v_i[$ and $x_e = 0$ for all other edges $e$. 
    Together with $a_e = 0$ for all $e \in \binom{\Z_n}{2} \setminus E^*$, $a^\top x = 0$ yields 
    \begin{align}\label{eq:sum-coeff-zero-1}
        a_{v_{i-1}v_i} + a_{v_iw_i^{m_i}} = 0 \quad \text{for } i = 1,\dots,k \enspace .
    \end{align}
    Similarly, by considering $\psi([v_i,w_i^1]), \psi(]v_i,w_i^1]) \in S \subseteq S'$ we obtain
    \begin{align}\label{eq:sum-coeff-zero-2}
        a_{v_iv_{i+1}} + a_{v_iw_i^1} = 0 \quad \text{for } i=1,\dots,k \enspace .
    \end{align}
    For $i=1,\dots,k$ and odd $j \in \{1,\dots,m_i-2\}$ it holds that $\psi([w_i^j,v_i]), \psi([w_i^{j+2},v_i]) \in S \subseteq S'$ and $a^\top(\psi([w_i^{j+2},v_i]) - \psi([w_i^j,v_i])) = 0$ yields 
    \begin{align}\label{eq:sum-coeff-zero-3}
        a_{v_iw_i^j} + a_{v_iw_i^{j+1}} = 0 \quad \text{for } i=1,\dots,k, \; \text{odd } j \in \{1,\dots,m_i-2\} \enspace .
    \end{align}
    For $i=1,\dots,k$ and odd $j \in \{3,\dots,m_i\}$ it holds that $\psi([v_i,w_i^j]), \psi([v_i,w_i^{j-2}]) \in S \subseteq S'$ and $a^\top(\psi([v_i,w_i^{j-2}]) - \psi([v_i,w_i^j]))$ yields 
    \begin{align}\label{eq:sum-coeff-zero-4}
        a_{v_iw_i^j} + a_{v_iw_i^{j-1}} = 0 \quad \text{for } i=1,\dots,k, \; \text{odd } j \in \{3,\dots,m_i\} \enspace .
    \end{align}
    All together, inequalities \eqref{eq:sum-coeff-zero-1} -- \eqref{eq:sum-coeff-zero-4} yield $a_{v_iv_{i+1}} = \alpha$ for $i=0,\dots,k$ and $a_{v_iw_i^j} = (-1)^j \alpha$ for $i=1,\dots,k$, $j=1,\dots,m_i$.
    From this it follows that $b = a^\top 1 = \alpha$.
    Clearly, it holds that $\alpha \neq 0$ since otherwise we have $a = 0$ and $0^\top x \leq 0$ does not define a facet.
    Lastly, $\alpha > 0$ holds because otherwise the all zeros vector $0 \in \lmc(C)$ would not satisfy \eqref{eq:star-glider-inequality}.
    
    It follows that $a^\top x \leq b$ is indeed a scalar multiple of \eqref{eq:star-glider-inequality} and that the star-glider inequality \eqref{eq:star-glider-inequality} defines a facet of $\lmc(C)$.
\end{delayedproof}

\bibliography{references}
\bibliographystyle{plainnat}

\end{document}